\RequirePackage{fix-cm}

\documentclass[twocolumn]{svjour3}
\smartqed
\usepackage{graphicx}

\makeatletter
\def\cl@chapter{\@elt {theorem}}
\makeatother

\usepackage{amsmath,amsfonts, amssymb, amsthm}
\usepackage{breakcites}
\usepackage{csquotes}
\usepackage[dvipsnames]{xcolor}
\usepackage{hyperref}
\hypersetup{
    colorlinks=true,
    linkcolor=blue,
    urlcolor=blue,
    citecolor=blue
}
\usepackage{cleveref}
\usepackage{graphicx}
\usepackage{subfig}
\usepackage[justification=centering]{caption}
\usepackage{stmaryrd}
\usepackage{tikz}
\usetikzlibrary{shapes,calc,positioning}
\usepackage[algoruled,linesnumbered]{algorithm2e}
\usepackage{comment}

\makeatletter

\SetCommentSty{mycommfont}
\newcommand{\nosemic}{\renewcommand{\@endalgocfline}{\relax}}
\newcommand{\dosemic}{\renewcommand{\@endalgocfline}{\algocf@endline}}

\let\oldnl\nl
\newcommand{\nonl}{\renewcommand{\nl}{\let\nl\oldnl}}
\makeatother

\newcommand{\mres}{\mathbin{\vrule height 1.6ex depth 0pt width
0.13ex\vrule height 0.13ex depth 0pt width 1.3ex}}

\newcommand{\NN}{\mathbb{N}}
\newcommand{\RR}{\mathbb{R}}
\newcommand{\Diff}{\mathrm{D}}
\newcommand{\LD}{\mathrm{L}^2(\RR^2)}

\newcommand{\eqdef}{\ensuremath{\stackrel{\mbox{\upshape\tiny def.}}{=}}}
\newcommand{\abs}[1]{\left\lvert#1\right\rvert}

\renewcommand{\div}{\mathrm{div}}

\newcommand{\Diffrad}{\mathrm{D}_{\mathrm{rad}}}
\newcommand{\Difforth}{\mathrm{D}_{\mathrm{orth}}}

\definecolor{ao(english)}{rgb}{0.0, 0.5, 0.0}
\newcommand{\revision}[1]{{#1}}
\newcommand{\revisionbis}[1]{{#1}}

\DeclareMathSymbol{\Phi}{\mathalpha}{operators}{8}

\journalname{Journal of Mathematical Imaging and Vision}

\begin{document}

\title{Towards off-the-grid algorithms for total variation regularized inverse problems\thanks{This work was supported by a grant from R\'egion Ile-De-France and by
the ANR CIPRESSI project, grant ANR-19-CE48-0017-01 of the French
Agence Nationale de la Recherche.}}

\titlerunning{Towards off-the-grid algorithms for total variation regularization}        

\author{Yohann De Castro        \and
        Vincent Duval \and
        Romain Petit 
}

\authorrunning{Y. De Castro \and V. Duval \and R. Petit} 

\institute{Y. De Castro \at
              Institut Camille Jordan, CNRS UMR 5208, \'Ecole Centrale de Lyon, F-69134 \'Ecully, France \\
              \email{yohann.de-castro@ec-lyon.fr}           
           \and
           V. Duval \at
           CEREMADE, CNRS, UMR 7534, Universit\'e Paris-Dauphine, PSL University, 75016 Paris, France \\
           INRIA-Paris, MOKAPLAN, 75012 Paris, France \\
           \email{vincent.duval@inria.fr}
           \and
           R. Petit \at
           CEREMADE, CNRS, UMR 7534, Universit\'e Paris-Dauphine, PSL University, 75016 Paris, France \\
           INRIA-Paris, MOKAPLAN, 75012 Paris, France \\
           \email{romain.petit@inria.fr}
}

\date{\today}

\maketitle

\begin{abstract}
We introduce an algorithm to solve linear inverse problems regularized with the total (gradient) variation in a gridless manner. Contrary to most existing methods, that produce an approximate solution which is piecewise constant on a fixed mesh, our approach exploits the structure of the solutions and consists in iteratively constructing a linear combination of indicator functions of simple polygons.
\end{abstract}

\section{Introduction}

By promoting solutions with a certain specific structure, the regularization of a variational inverse problem is a way to encode some prior knowledge on the signals to recover. Theoretically, it is now well understood which regularizers tend to promote signals or images which are sparse, low rank or piecewise constant. Yet, paradoxically enough, most numerical solvers are not designed with that goal in mind, and the targeted structural property (sparsity, low rank or piecewise constancy) only appears ``in the limit'', when the algorithm converges.

Several recent works have focused on incorporating structural
properties in optimization algorithms. In the context of~$\ell^1$-based sparse spikes recovery, it was proposed to switch from, e.g.
standard proximal methods (which require the introduction of an
approximation grid)
to algorithms which operate directly in a continuous domain:
interior point methods solving a reformulation of the problem
\cite{candesMathematicalTheorySuperresolution2014,castroExactSolutionsSuper2017} or a Frank-Wolfe / conditional gradient algorithm
\cite{brediesInverseProblemsSpaces2013} approximating a
solution in a greedy way.
More generally, the conditional gradient algorithm has drawn a lot
of interest from the data science community, for it provides
iterates which are a sum of a small number of atoms which are
promoted by the regularizer (see the review paper
\cite{jaggiRevisitingFrankWolfeProjectionFree2013}).

In the present work, we explore the extension of these fruitful approaches to  the total (gradient) variation regularized inverse problem
\begin{equation}
	\underset{u \in \LD}{\min} ~ T_{\lambda}(u)\eqdef\frac{1}{2}\,\left\Vert\Phi u-y\right\Vert^2+\lambda\,\abs{\Diff u}(\mathbb{R}^2)\,,
    \label{prob}
		\tag{$\mathcal{P}_{\lambda}$}
\end{equation}
where $|\Diff u|(\RR^2)$ denotes the total variation of (the gradient of) $u$ and $\Phi\colon \LD\rightarrow \RR^m$ is a continuous linear map such that
\begin{equation}
	\forall u\in \LD,~\Phi u=\int_{\mathbb{R}^2}u(x)\,\varphi(x)\,\mathrm{d}x\,,
   \label{measurement_op}
\end{equation}
with $\varphi\in \left[\LD\right]^m\cap C^0(\mathbb{R}^2,\mathbb{R}^m)$. \revisionbis{Such variational problems have been widely used in imaging for the last decades, following the pioneering work of Rudin, Osher and Fatemi \cite{rudinNonlinearTotalVariation1992}. A typical application is the reconstruction of an unknown image~$u_0$ from a set of noisy linear measurements ${y=\Phi u_0+w}$, where~${w\in\RR^m}$ is some additive noise.

The total variation term in \eqref{prob} is known to promote piecewise constant solutions. It has been shown that some solutions are sums of at most~$m$ indicator functions of simple sets (see \cite{brediesSparsitySolutionsVariational2019,boyerRepresenterTheoremsConvex2019}). However, when~$u_0$ is a simple piecewise constant image, there are evidences that solutions are usually made of a much smaller number of shapes. In such situations, it is highly desirable to design numerical solvers preserving this structure, and able to accurately estimate the jump set of solutions. This could be particularly relevant for specific applications, like astronomical and cell imaging.}

\subsection{Previous works}
\revisionbis{
Many algorithms have been proposed to solve \eqref{prob}. The vast majority of them rely on the introduction of a fixed spatial discretization, and of a discrete version of the total variation (see \cite{chambolleChapterApproximatingTotal2021} for a review). These approaches often yield reconstruction artifacts, such as anisotropy or blur (see the previous reference, \cite{tabtiSymmetricUpwindScheme2018}, and the experiments section below). Most importantly, existing algorithms often fail to preserve the structure exhibited by solutions of \eqref{prob}, which is discussed above.

To circumvent these issues, mesh adaptation techniques were introduced in \cite{violaUnifyingResolutionindependentFormulation2012,bartelsSingularSolutionsGraded2021}. The refinement rules they propose are, however, either heuristic or too restrictive to faithfully recover edges. In any case, they still rely on a discretization of the whole domain, and hence do not provide a compact representation of the reconstruced image.

In \cite{ongieOfftheGridRecoveryPiecewise2016}, a method for recovering  piecewise constant images from few Fourier samples is introduced. Its orginiality is to produce a continuous domain representation of the image, assuming its edge set is a trigonometric curve. However, this approach heavily relies on relations satisfied by the Fourier coefficients of the image. As such, it does not seem possible to adapt it to handle other types of measurements.}

\subsection{Contributions}
\revisionbis{
Our goal is to design an algorithm which does not suffer from some grid bias, while providing a continuous domain representation of solutions. To this aim, we construct an approximate solution built from the above-mentioned atoms, namely indicator functions of simple sets. As shown in the experiments section, this approach is particularly suited for reconstructing simple piecewise constant images. On more complex natural images, traditional grid-based methods perform better. In \Cref{pres_fw,pres_sliding}, we introduce a theoretical iterative algorithm, whose output provably converges to a solution of \eqref{prob}. The exploratory nature of our work lies in the numerical methods we propose to carry out several steps of this algorithm. Although experiments suggest they perform well, several questions concerning their theoretical analysis remain.}

\section{Preliminaries}

In the following, for any function $u:\RR^2\to\RR$, we shall use the notation
\begin{equation*}
	U^{(t)}\eqdef\begin{cases}\left\{x\in\RR^2\,\big\rvert\,u(x)\geq t\right\} & \text{if } t\geq 0\,, \\
\left\{x\in\RR^2\,\big\rvert\,u(x)\leq t\right\} & \text{otherwise.}\end{cases}\end{equation*}

\subsection{Functions of bounded variation and sets of finite perimeter}
\label{preli_sets}

Let $u\in\mathrm{L}^1_{loc}(\RR^2)$. The total variation of $u$ is given by
\begin{equation*}
  J(u)\eqdef\underset{z \in C_c^{\infty}(\RR^2,\RR^2)}{\text{sup}}~ -\int_{\RR^2}u\,\div z ~\text{ s.t. }~ \|z\|_{\infty}\leq 1\,.
\end{equation*}
If $J(u)$ is finite, then $u$ is said to have bounded variation, and the distributional gradient of $u$, denoted~$\Diff u$, is a finite vector-valued Radon measure. We moreover have $|\Diff u|(\RR^2)=J(u)<+\infty$.

A measurable set $E\subset\RR^2$ is said to be of finite perimeter if $P(E)\eqdef J(\mathbf{1}_E)<+\infty$. The reduced boundary $\partial^*E$ of a set of finite perimeter $E$ is defined as the set of points $x\in\text{Supp}\left(\left|\Diff\mathbf{1}_E\right|\right)$ at which
\begin{equation*}
	\nu_E(x)\eqdef \underset{r\to 0^+}{\text{lim}}-
\frac{\Diff\mathbf{1}_E(B(x,r))}{\left|\Diff\mathbf{1}_E\right|(B(x,r))}
\end{equation*}
exists and is moreover such that $\|\nu_E(x)\|=1$.

From \cite[Proposition 3.1]{giustiMinimalSurfacesFunctions1984}, we know that if $E$ has finite perimeter, there exists a Lebesgue representative of $E$ with the property that \begin{equation*}
\forall x\in\partial E$, $0<|E\cap B(x,r)|<|B(x,r)|\,.
\end{equation*}
In the following, we always consider such a representative and consequently obtain $\text{Supp}(\Diff\mathbf{1}_E)=\overline{\partial^*E}=\partial E$.

\revision{We now introduce the notion of indecomposable and simple sets, which are the measure-theoretic analogues of connected and simply connected sets (see \cite{ambrosioConnectedComponentsSets2001} for more details). A set of finite perimeter~$E\subset\RR^2$ is said to be decomposable if there exists a partition of $E$ in two sets of positive measure~$A$ and~$B$ with~${P(E)=P(A)+P(B)}$. We say that $E$ is indecomposable if it is not decomposable. Any indecomposable set of finite measure whose complement is also indecomposable is called simple. If~${E\subset\RR^2}$ has finite perimeter and finite measure it can be decomposed (up to Lebesgue negligible sets) into an at most countable union of pairwise disjoint
indecomposable sets, i.e.
\begin{equation}
   \label{mcomponents}
	E=\bigcup\limits_{i\in I}E_i,\,P(E)=\sum\limits_{i\in I}P(E_i)\text{ and }\forall i,\,|E_i|>0\,.
\end{equation}
Each~$E_i$ can in turn be decomposed as
\begin{equation}
   \label{jordan}\begin{aligned}
E_i&=\mathrm{int}(\gamma_i^+)\setminus\bigcup\limits_{j\in J_i}\mathrm{int}(\gamma_{i,j}^-)\,, \\
\text{with }
P(E_i)&=P(\mathrm{int}(\gamma^+_i))+\sum\limits_{j\in J_i}P(\mathrm{int}(\gamma_{i,j}^-))\,,
\end{aligned}\end{equation}
where for all $i\in I$ and $j\in J_i$, $\gamma_i^+$ and $\gamma_{i,j}^-$ are rectifiable Jordan curves.}

\begin{figure}
   \centering
   \begin{tikzpicture}[x=0.75pt,y=0.75pt,yscale=-1,xscale=1]

\draw  [fill={rgb, 255:red, 223; green, 223; blue, 223 }  ,fill opacity=1 ] (166.5,61.5) .. controls (180.36,55.74) and (253.95,50.13) .. (261.37,75.77) .. controls (268.79,101.4) and (238.81,83.98) .. (212.98,110.63) .. controls (187.15,137.29) and (106.58,123.05) .. (101.18,94.25) .. controls (95.78,65.45) and (152.64,67.26) .. (166.5,61.5) -- cycle ;
\draw  [fill={rgb, 255:red, 255; green, 255; blue, 255 }  ,fill opacity=1 ] (165.8,88.41) .. controls (173.96,104.41) and (162.23,109.95) .. (148.97,104.41) .. controls (135.71,98.87) and (134.43,101.95) .. (145.4,85.33) .. controls (156.37,68.71) and (157.64,72.4) .. (165.8,88.41) -- cycle ;
\draw  [fill={rgb, 255:red, 255; green, 255; blue, 255 }  ,fill opacity=1 ] (226.67,67.2) .. controls (241.73,60.42) and (229.31,105.85) .. (214.26,85.51) .. controls (199.2,65.16) and (211.62,73.98) .. (226.67,67.2) -- cycle ;
\draw  [fill={rgb, 255:red, 224; green, 224; blue, 224 }  ,fill opacity=1 ] (288.83,67.38) .. controls (300.6,61.42) and (353.58,55.47) .. (341.8,67.38) .. controls (330.03,79.28) and (330.03,85.24) .. (341.8,103.1) .. controls (353.58,120.97) and (300.6,120.97) .. (288.83,103.1) .. controls (277.06,85.24) and (277.06,73.33) .. (288.83,67.38) -- cycle ;

\draw (139,41.4) node [anchor=north west][inner sep=0.75pt]    {$\gamma _{1}^{+}$};
\draw (114,89.4) node [anchor=north west][inner sep=0.75pt]    {$\gamma _{1,1}^{-}$};
\draw (187,73.4) node [anchor=north west][inner sep=0.75pt]    {$\gamma _{1,2}^{-}$};
\draw (266,98.4) node [anchor=north west][inner sep=0.75pt]    {$\gamma _{2}^{+}$};

\end{tikzpicture}
   \caption{Decomposition of a set (gray area) as in \eqref{mcomponents}, \eqref{jordan}}
\end{figure}
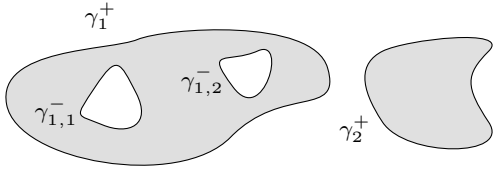

\subsection{Subdifferential of the total variation}

In the rest of this document, $J$ is considered as a mapping from~${\LD}$ to~${\RR\cup\{+\infty\}}$. This mapping is convex, proper and lower semi-continuous. We have the following useful characterizations of $\partial J(0)$:
\begin{equation*}
\begin{aligned}
\partial J(0)= \bigg\{&\eta\in\LD \,\bigg\rvert\,\\
&\forall u\in \LD,~\left|\int_{\mathbb{R}^2}\eta\,u\right|\leq |\Diff u|(\mathbb{R}^2)\bigg\}\,,
\end{aligned}
\end{equation*}
\begin{equation*}
\begin{aligned}
\partial J(0)= \bigg\{&\eta\in \LD\,\bigg\rvert\, \forall E\subset{\mathbb{R}^2},\,0<|E|<+\infty\\
&\text{and }P(E)<+\infty\implies\left|\int_{\mathbb{R}^2}\eta\,\frac{\mathbf{1}_E}{P(E)}\right|\leq 1\bigg\}\,.
\end{aligned}
\end{equation*}
Moreover, the subdifferential of $J$ at $u\in\LD$ is given by:
\begin{equation*}
		\partial J(u)=\left\{\eta\in \partial J(0)\,\bigg\rvert\,\int_{\RR^2}\eta\,u=|\Diff u|(\RR^2)\right\}.
\end{equation*}
We also have the following useful result:
\begin{proposition}[see e.g. \cite{chambolleGeometricPropertiesSolutions2016}] Let~${u\in\LD}$ be such that $J(u)<\infty$ and~${\eta\in\LD}$. Then $\eta\in\partial J(u)$ if and only if $\eta\in\partial J(0)$ and the level sets of $u$ satisfy
	\begin{equation}
\left\{
\begin{aligned}
		\forall t>0,~P(U^{(t)})&=\int_{U^{(t)}}\eta\,, \\
		\forall t<0,~P(U^{(t)})&=-\int_{U^{(t)}}\eta\,.
	\end{aligned}\right.
	\end{equation}
	\label{lvlset_sousdiff}
\end{proposition}
\subsection{Dual problem and dual certificates}
The Fenchel-Rockafellar dual of \eqref{prob} is the following finite dimensional problem
\begin{equation}
\underset{p \in \RR^m}{\text{max}} ~  \langle p,y\rangle -\frac{\lambda}{2}\,\|p\|^2 \quad \text{s.t.} \quad \Phi^*p\in\partial J(0)\,,
  \tag{$\mathcal{D}_{\lambda}$}
  \label{dual}
\end{equation}
which has a unique solution (it is in fact equivalent to the projection of $\frac{y}{\lambda}$ on the closed convex set of vectors $p$ such that $\Phi^*p\in\partial J(0)$). Moreover, strong duality holds as stated by the following proposition
\begin{proposition}
  Problems \eqref{prob} and \eqref{dual} have the same value and any solution $u_{\lambda}$ of \eqref{prob} is linked with the unique solution $p_{\lambda}$ of \eqref{dual} by the extremality condition
  \begin{equation}
  \left\{\begin{aligned}
    \Phi^*p_{\lambda}&\in \partial J(u_{\lambda})\,, \\
    p_{\lambda} &= -\frac{1}{\lambda}\left(\Phi u_{\lambda}-y\right)\,.
  \end{aligned}\right.
	\label{eq_dualite}
  \end{equation}
	\label{prop_dualite}
\end{proposition}
\begin{remark}
	\Cref{prop_dualite} implies in particular that all solutions of \eqref{prob} have the same total variation and the same image by $\Phi$.
\end{remark}

\subsection{Distributional curvature}
We denote by $\mathcal{H}^1$ the $1$-dimensional Hausdorff measure on $\RR^2$, and for every Borel set~${A\subset\RR^2}$, by $\mathcal{H}^1\mres A$ the measure $\mathcal{H}^1$ restricted to $A$, i.e. such that for every Borel set $E$ we have
\begin{equation*}
	\left(\mathcal{H}^1\mres A\right)(E)=\mathcal{H}^1(A\cap E)\,.
\end{equation*}

If $E\subset\RR^2$ is a set of finite perimeter, then the distributional  curvature vector of $E$ is~${\mathbf{H}_E:C^{\infty}_c(\RR^2,\RR^2)\to\RR}$ defined by
 \begin{equation*}
    \forall T\in C^{\infty}_c(\RR^2,\RR^2),~\langle \mathbf{H}_E,T\rangle=\int_{\partial^*E}\text{div}_E\,T\,d\mathcal{H}^{1}\,,
\end{equation*}
  where $\text{div}_E\,T$ denotes the tangential divergence of $T$ on~$E$ given by
  \begin{equation*}\text{div}_E\,T=\text{div}\,T - (\Diff T\,\nu_E)\cdot \nu_E\,,\end{equation*}
where $\Diff T$ denotes the differential of $T$. $E$ is said to have locally integrable distributional curvature if there exists a function  ${H_E\in\mathrm{L}^1_{loc}(\partial^*E;\mathcal{H}^1)}$ such that
\begin{equation*}\mathbf{H}_E=H_E\,\nu_E\,\mathcal{H}^1\mres\partial^*E\,.\end{equation*}

For instance, if $E$ is an open set with $C^2$ boundary, it has a locally summable distributional curvature which is given by the (classical) scalar mean curvature.

\section{A modified Frank-Wolfe algorithm}
\label{pres_fw}

\begin{algorithm}
 \KwData{objective $f$, domain $C$, starting point $x_0\in C$}
 \KwResult{point $x^*$}
 \While{true}{
  find $s_k\in \underset{s\in C}{\text{Argmin}}f(x_k)+df(x_k)\left[s-x_k\right]$\;

  \eIf{$df(x_k)\left[s_k-x_k\right]=0$}{
   output $x^*\gets x_k$, which is optimal\;
   }{
   $\gamma_k\gets \frac{2}{k+2}$\;
   \tcp{tentative update}
   $\tilde{x}_{k+1}\gets x_k+\gamma_k (s_k-x_k)$\;
   \tcp{final update}
   choose any $x_{k+1}$ such that $f(x_{k+1})\leq f(\tilde{x}_{k+1})$\;
   }
 }
 \caption{Frank-Wolfe algorithm}
 \label{fw}
\end{algorithm}

In the spirit of \cite{brediesInverseProblemsSpaces2013,boydAlternatingDescentConditional2017,
denoyelleSlidingFrankWolfeAlgorithm2019} which introduced variants of the conditional gradient algorithm for sparse spikes recovery in a continuous domain, we derive a modified Frank-Wolfe
algorithm allowing to iteratively solve \eqref{prob} in a gridless manner.

\subsection{Description}
\label{sec_fw}

The Frank-Wolfe algorithm (see Algorithm \ref{fw}) allows to minimize a convex differentiable function $f$ over a weakly compact convex subset $C$ of a Banach space. Each step of the algorithm consists in minimizing a linearization of $f$ on $C$, and building the next iterate as a convex combination of the obtained point and the current iterate.

An important feature of the algorithm is that while the classical update (Line 8) is to take $x_{k+1}$ to be equal to $\tilde{x}_{k+1}$, all convergence guarantees
are preserved if one chooses any $x_{k+1}\in C$ such that~${f(x_{k+1})\leq
f(\tilde{x}_{k+1})}$ instead.

Even though $T_{\lambda}$ is not differentiable, it is possible to recast problem \eqref{prob} into that framework by performing an epigraphical lift (see \Cref{derivation_algo}). In this setting, the linear minimization step which is at the core of the algorithm amounts to solving the following problem
\begin{equation}
	\underset{u \in \LD}{\text{min}} ~ \int_{\mathbb{R}^2}\eta\,u  \quad \text{ s.t. } ~  \abs{\Diff u}(\mathbb{R}^2)\leq 1\,,
  \label{cheeger_relaxed}
\end{equation}
for an iteration-dependent function $\eta \in \LD$. Denoting $u^{[k]}$ the $k$-th iterate, this function is given by \begin{equation*}
\eta^{[k]}\eqdef-\frac{1}{\lambda}\Phi^*\left(\Phi u^{[k]}-y\right).
\end{equation*}
As is usual when using the Frank-Wolfe algorithm, we notice that since the objective of \eqref{cheeger_relaxed} is linear and the total variation unit ball is convex and compact (in the weak
$\LD$ topology), at least one of its extreme points is optimal. A result due
to Fleming \cite{flemingFunctionsGeneralizedGradient1957} (see
also~\cite{ambrosioConnectedComponentsSets2001}) states that those extreme
points are exactly the functions of the form  $\pm\mathbf{1}_E/P(E)$ where
$E\subseteq\mathbb{R}^2$ is a simple set with $0<|E|<+\infty$. This means the linear minimization step can be carried out by finding a simple set solving the following geometric variational problem:
\begin{equation}
   \underset{E \subseteq \mathbb{R}^2}{\text{max}} ~
   \frac{\left|\int_{E}\eta\right|}{P(E)}
    \quad \text{ s.t. } ~ 0<|E|<+\infty,~P(E)<+\infty\,.
  \label{cheeger}
\end{equation}
Since Problem~\eqref{cheeger} is reminiscent of the Cheeger problem \cite{pariniIntroductionCheegerProblem2011}, which, given a domain~$\Omega\subseteq\mathbb{R}^2$, consists in finding the subsets $E$
of $\Omega$ minimizing the ratio~${P(E)/|E|}$, we refer to it as the~``Cheeger problem'' in the rest of the paper, and to any of its solutions as a ``Cheeger set''.

In view of the above, we derive Algorithm \ref{sfw}, which produces a sequence of functions that are linear combinations of indicators of simple sets, and which is a valid application of Algorithm \ref{fw} to \eqref{prob}, in the sense that \Cref{conv_res} holds.

\begin{algorithm}
\KwData{measurement operator $\Phi$, observations $y$, regularization parameter $\lambda$}
\KwResult{function $u^*$}
$u^{[0]}\gets 0$\;
$N^{[0]}\gets 0$\;
\While{true}{
$\eta^{[k]}\gets -\frac{1}{\lambda}\Phi^*\left(\Phi u^{[k]}-y\right)$\;
$E_*\gets \underset{E \text{ simple}}{\text{Argmax}}~\frac{\left|\int_E\eta^{[k]}\right|}{P(E)} \text{ s.t. } 0<|E|<+\infty$\;
\eIf{$\left|\int_{E_*}\eta^{[k]}\right|\leq P(E_*)$}{
output $u^*\gets u^{[k]}$, which is optimal\;}
{$E^{[k+1]}\gets(E_1^{[k]},...,E_{N^{[k]}}^{[k]},E_*)$\;
$a^{[k+1]}\gets \underset{a\in\mathbb{R}^{N^{[k]}+1}}{\text{argmin}}
T_{\lambda}\left(\sum\limits_{i=1}^{N^{[k]}+1}a_i
\,\mathbf{1}_{E^{[k+1]}_i}\right)$\;
remove atoms with zero amplitude\;
$N^{[k+1]}\gets$ number of atoms in $E^{[k+1]}$\;
$u^{[k+1]}\gets\sum\limits_{i=1}^{N^{[k+1]}}a^{[k+1]}_i\,\mathbf{1}_{E^{[k+1]}_i}$\;}
}
\caption{modified Frank-Wolfe algorithm applied to \eqref{prob}}
\label{sfw}
\end{algorithm}

\begin{remark}
   We use here a so-called \enquote{fully corrective} variant of
   Frank-Wolfe, meaning that instead of choosing the next iterate ${u^{[k+1]}}$ as a convex
   combination of~${\pm\mathbf{1}_{E_*}/P(E_*)}$ and the previous iterate ${u^{[k]}}$ as in Line~7 of Algorithm \ref{fw}, we optimize (Line 10 of Algorithm \ref{sfw}) the objective over $\text{Vect}\left(\left(\mathbf{1}_{E_i}\right)_{i=1}^{N^{[k]}+1}\right)$, which decreases the objective more than the standard update, and hence does not break convergence guarantees.
\end{remark}

\begin{remark}
   \revision{Line 10 of Algorithm \ref{sfw} can always be reduced to the resolution of a LASSO-type problem (possibly changing $E^{[k+1]}$ and constraining the sign of the components of $a$). Indeed, given $N\in\mathbb{N}^*$ and~$E_1,...,E_N$ a collection of simple sets, assuming that we have
   \begin{equation}
      \forall a\in\RR^N,~\left|\Diff\left(\sum\limits_{i=1}^N a_i\,\mathbf{1}_{E_i}\right)\right|(\mathbb{R}^2)=\sum\limits_{i=1}^N |a_i|\,P(E_i)\,,
      \label{tv_sum}
   \end{equation}
   then we get that
   \begin{equation*}
      T_{\lambda}(u)=\frac{1}{2}||\Phi_E\,a -y||^2+\lambda\,\sum\limits_{i=1}^N P(E_i)\,|a_i|\,,
   \end{equation*}
   with
   \begin{equation*}
   \Phi_E\eqdef\left[\left(\int_{E_i}\varphi_j\right)_{\substack{1\leq i\leq N\\1\leq j\leq m}}\right]^T\in\RR^{m\times N}\,.
   \end{equation*}
   Hence, finding the vector $a$ minimizing $T_{\lambda}(u)$ with the sets~${E_1,...,E_N}$ fixed amounts to solving a finite dimensional least squares problem with a weighted $\ell^1$ norm penalization (the weights are here the perimeters of the sets $(E_i)_{i=1}^N$).

   Identity \eqref{tv_sum} holds as soon as $\mathcal{H}^1(\partial^*E_i\cap\partial^*E_j)=0$ for every $i\neq j$. Although this is generically satisfied, and that we never observe experimentally configurations where this fails, we describe in \Cref{lasso_amplitudes} how to change $E^{[k+1]}$ to reduce Line 10 to a LASSO-type problem at the price of constraining the sign of the components of $a$.}
\end{remark}

\begin{remark}
   The stopping condition is here replaced by
   \begin{equation*}
      \underset{E}{\sup}~\frac{
      \left|\int_E\eta^{[k]}\right|
      }{P(E)}\leq 1 \text{, with } \eta^{[k]}=-\frac{1}{\lambda}\Phi^*\left(\Phi u^{[k]}-y\right),
   \end{equation*}
   which is equivalent to~${\eta^{[k]}\in\partial J(0)}$. Since the optimality of $a^{[k]}$ at Line 10 always ensures $\int_{\RR^2}\eta^{[k]}\,u^{[k]}=J(u^{[k]})$, this yields ${\eta^{[k]}\in\partial J(u^{[k]})}$ and hence \eqref{eq_dualite} holds,
   which means $u^{[k]}$ solves \eqref{prob}.
\end{remark}

\subsection{Convergence results}

As already mentioned, Algorithm \ref{sfw} is a valid application of Algorithm \ref{fw} to \eqref{prob}, in the sense that the following property holds (see
\cite{jaggiRevisitingFrankWolfeProjectionFree2013}):
\begin{proposition}
   Let $(u^{[k]})_{k\geq 0}$ be a sequence produced by Algorithm \ref{sfw}. Then there exists $C>0$ such that for any solution $u^*$ of Problem \eqref{prob},
\begin{align}
   \forall k\in\mathbb{N}^*,~T_{\lambda}(u^{[k]})-T_{\lambda}(u^*)\leq \frac{C}{k}\,.
\end{align}
\label{conv_res}
\end{proposition}
\begin{remark}
   As discussed in \cite{jaggiRevisitingFrankWolfeProjectionFree2013}, the
   linear minimization step (solving \eqref{cheeger_relaxed} or
   equivalently~\eqref{cheeger}) can be solved approximately. In fact if there exists $\delta>0$ such that for every $k$ the set computed at~${\text{Line 5 }}$ is an $\epsilon_k$-maximizer of \eqref{cheeger} with $\epsilon_k=\frac{\gamma}{k+2}\delta$, then
   \begin{align}
      \forall k\in\mathbb{N}^*,~T_{\lambda}(u^{[k]})-T_{\lambda}(u^*)\leq \frac{2\,\gamma}{k+2}(1+\delta)\,,
   \end{align}
   where $\gamma$ is the curvature constant of the objective used in the reformulation of $\eqref{prob}$. One can in fact show that this curvature constant is smaller than a quantity which is proportional to $\|\Phi\|^2\,\left(\|y\|^2/\lambda\right)^2$.
   \label{approx_jaggi}
\end{remark}
We first provide a general property of minimizing sequences (see e.g.  \cite{iglesiasNoteConvergenceSolutions2018} for a proof), which hence applies to the sequence of iterates produced by Algorithm \ref{sfw}.
\begin{proposition}
   Let $(u_n)_{n\geq 0}$ be a minimizing sequence for \eqref{prob}. Then there exists a subsequence (not relabeled) which converges weakly in $\LD$ and strongly in $\mathrm{L}^1_{loc}(\RR^2)$ to a solution $u_*$ of \eqref{prob}. Moreover, we have~${\Diff u_n\overset{\ast}{\rightharpoonup}\Diff u_*}$ and $|\Diff u_n|(\RR^2)\to |\Diff u|(\RR^2)$.
   \label{conv_min_seq}
\end{proposition}

We now provide additional properties of sequences produced by Algorithm \ref{sfw}. We first begin by noticing that if $(u^{[k]})_{k\geq 0}$ is such a sequence, then the optimality condition at Line 10 ensures that \begin{equation*}
\forall k,~\forall i\in\{1,...,N^{[k]}\},~ P(E_i^{[k]})=\left|\int_{E_i^{[k]}}\eta^{[k]}\right|.
\end{equation*}
But from \Cref{conv_res} and \Cref{conv_min_seq}
we have the existence of a (not relabeled) subsequence which converges strongly in $\mathrm{L}^1_{loc}(\RR^2)$ and weakly in $\LD$ towards a solution $u^*$ of \eqref{prob}. The weak convergence of $(u^{[k]})_{k\geq 0}$ in $\LD$
implies that~${\underset{n\to+\infty}{\text{lim}}\Phi u^{[k]}=\Phi u^*}$, which in turns yields the strong convergence in $\LD$ of $(\eta^{[k]})_{k\geq 0}$ towards the solution $\eta^*$ of \eqref{dual}. We can then use the following lemma to show all the sets $E_i^{[k]}$ are included in some common ball.
\begin{lemma}
   Let $(\eta_k)_{k\geq 0}$ be a sequence of functions converging strongly to $\eta_{\infty}$ in $\LD$. For all $k\geq 0$, we denote
   \begin{equation*}
      \mathcal{F}_k\eqdef\left\{E\text{ simple}\,\bigg\rvert\,0<|E|<+\infty,~P(E)=\left|\int_{E}\eta_k\right|\right\},
   \end{equation*}
   and $\mathcal{F}=\cup_{k\geq 0}\mathcal{F}_k$. Then there exist positive real numbers $R$ and $C$ such that
   \begin{equation*}\forall E\in\mathcal{F},~P(E)\leq C \text{ and }E\subset B(0,R)\,.\end{equation*}
   \label{lemma_bounds}
\end{lemma}
\begin{proof}
   This proof is based on \cite[Section 5]{chambolleGeometricPropertiesSolutions2016}.

   \medskip
   \noindent\textbf{Upper bound on the perimeter:} the family of functions $\{\eta_k^2,~k\in\NN\}\cup\{\eta_{\infty}^2\}$ being equi-integrable, for all~${\epsilon>0}$ there exists $R_1>0$ such that
   \begin{equation*}
      \forall k,~\int_{\mathbb{R}^2\setminus B(0,R_1)}\eta_k^2\leq \epsilon^2\,.
   \end{equation*}
   Let $E\in\mathcal{F}$. Then there exists $k$ s.t.~${P(E)= \left|\int_{E}\eta_k\right|}$ and we have:
\begin{equation*}
   \begin{aligned}
      \left|\int_{E}\eta_k\right|&\leq \left|\int_{E\cap B(0,R_1)}\eta_k\right|+\left|\int_{E\setminus B(0,R_1)}\eta_k\right|\\
		&\leq~~\sqrt{|B(0,R_1)|}~||\eta_k||_{L^2}\\
      &~~~+\sqrt{|E\setminus B(0,R_1)|}\,\sqrt{\int_{\mathbb{R}^2\setminus B(0,R_1)}\eta_k^2}\\
      & \leq \underset{k}{\sup}~||\eta_k||_{L^2}\,\sqrt{|B(0,R_1)|}+\epsilon\,\sqrt{|E\setminus B(0,R_1)|}\,.
   \end{aligned}
\end{equation*}
Moreover,
\begin{equation*}
   \sqrt{|E\setminus B(0,R_1)|}\leq \frac{1}{\sqrt{c_2}}(P(E)+P(B(0,R_1)))\,,
\end{equation*}
where $c_2\eqdef 4\pi$ is the isoperimetric constant. Hence taking $\epsilon\eqdef\frac{\sqrt{c_2}}{2}$ and defining
\begin{equation*}
   C\eqdef 2\left(\sqrt{|B(0,R_1)|}~\underset{k}{\text{sup}}~||\eta_k||_{L^2}+\frac{1}{2}P(B(0,R_1))\right),
\end{equation*}
we have $P(E)\leq \frac{1}{2}P(E) + \frac{C}{2}$ and hence $P(E)\leq C$.

\medskip

\noindent\textbf{Inclusion in a ball:} we still take $\epsilon=\frac{\sqrt{c_2}}{2}$ and fix a real $R_2>0$ such that $\int_{\mathbb{R}^2\setminus B(0,R_2)}\eta_k^2\leq \epsilon^2$ for all $k$. Now let $E\in\mathcal{F}$ and $k$ such that
\begin{equation*}
   P(E)= \left|\int_{E}\eta_k\right|.
\end{equation*}
Let us show that $E\cap B(0,R_2)\neq \emptyset$. By contradiction, if $E\cap B(0,R_2)= \emptyset$, we would have:
\begin{equation*}
   \begin{aligned}
   P(E)= \left|\int_{E\setminus B(0,R_2)}\eta_k\right|&\leq \sqrt{\int_{\mathbb{R}^2\setminus B(0,R_2)}\eta_k^2}~\sqrt{|E|}\\
   &\leq\frac{\epsilon}{\sqrt{c_2}}P(E)\,.
\end{aligned}
\end{equation*}
Dividing by $P(E)$ (which is positive since $0<|E|<\infty$) yields a contradiction. Since $E$ is simple, the perimeter bound yields $\text{diam}(E)\leq C$, which shows~${E\subset B(0,R)}$ with $R\eqdef C+R_2$.
\end{proof}
We have now shown there exists $R>0$ such that for all $k$
we have $\text{Supp}(u^{[k]})\subset
B(0,R)$, which means the strong $\mathrm{L}^1_{loc}$ convergence of $(u^{[k]})_{k\geq 0}$ towards $u^*$ is in fact a strong $\mathrm{L}^1$ convergence. This slightly improved convergence result is summarized in the following proposition:
\begin{proposition}
   Let $(u_n)_{n\geq 0}$ be a sequence produced by Algorithm \ref{sfw}. Then there exists a (not relabeled) subsequence and $R>0$ such that $\text{Supp}(u_n)\subset B(0,R)$ for all $n$. Moreover, this subsequence  converges strongly in~${\mathrm{L}^1(\RR^2)}$ to a solution $u_*$ of \eqref{prob} (and by \Cref{conv_min_seq} weakly in $\LD$, with
   moreover~${\Diff u_n\overset{\ast}{\rightharpoonup}\Diff u_*}$ and~${|\Diff u_n|(\RR^2)\to |\Diff u|(\RR^2)}$).
   \label{conv_l1}
\end{proposition}

\begin{corollary}
   Let $(u_n)_{n\geq 0}$ be a subsequence such as in \Cref{conv_l1}. Up to another extraction, for almost every~$t\in\RR$, we have
   \begin{equation*}
      \underset{n\to+\infty}{\text{lim}}~|U_n^{(t)}\triangle\, U_*^{(t)}|=0 ~\text{ and }~ \partial U_*^{(t)}\subseteq\underset{n\to+\infty}{\text{lim inf}}~\partial U_n^{(t)}\,,
   \end{equation*}
   where\footnote{For more details on this type of set convergence, see e.g.
   \cite[Chapter 4]{rockafellarVariationalAnalysis1998}.}
   \begin{equation*}
      \underset{n\to+\infty}{\text{lim inf}}~\partial U_n^{(t)}\eqdef\big\{x\in\RR^2\,\big\rvert\,\underset{n\to+\infty}{\limsup}~\mathrm{dist}(x,\partial U_n^{(t)})=0\big\}\,.
   \end{equation*}
   \label{coroll_conv_level_set}
\end{corollary}
   \begin{proof}
   The strong convergence of $(u_n)_{n\geq 0}$ towards a solution~$u_*$ in $\mathrm{L}^1(\RR^2)$ and Fubini's theorem give
   \begin{equation*}
      0=\underset{n\to+\infty}{\text{lim}}~\int_{\RR^2}|u_n-u_*|=\underset{n\to+\infty}{\text{lim}}~\int_{\RR}\left|U_n^{(t)}\triangle\, U_*^{(t)}\right|\,dt\,.
   \end{equation*}
   Hence, up to the extraction of a further subsequence, that we do not relabel, we get that
   \begin{equation*}
      \underset{n\to+\infty}{\text{lim}}~|U_n^{(t)}\triangle\, U_*^{(t)}|=0
      \text{ for almost every }t\in\RR\,.
   \end{equation*}
   We now fix such $t\in\RR$ and let $x\in\partial U_*^{(t)}$. We want to show that $x\in\underset{n\to +\infty}{\liminf}~\partial U_n^{(t)}$, which is equivalent to
   \begin{equation*}
      \underset{n\to+\infty}{\text{lim sup}}~\text{dist}\big(x,\,\partial U_n^{(t)}\big)=0\,.
   \end{equation*}
   By contradiction, if the last identity does not hold,
   we have the existence of $r>0$ and of $\varphi$ such that
   \begin{equation*}
      \forall n\in\mathbb{N},~B(x,r)\cap \partial U_{\varphi(n)}^{(t)}=\emptyset\,.
   \end{equation*}
   Hence for all $n$, we either have
   \begin{equation*}
      B(x,r)\subset U_{\varphi(n)}^{(t)} ~\text{ or }~ {B(x,r)\subset \RR^2\setminus U_{\varphi(n)}^{(t)}}\,.
   \end{equation*}
   If $B(x,r)\subset U_{\varphi(n)}^{(t)}$ for a given $n$ then
   \begin{equation*}
      \begin{aligned}
      \left|U_{\varphi(n)}^{(t)}\triangle\,
   U_*^{(t)}\right|\geq \left|U_{\varphi(n)}^{(t)}\setminus U_*^{(t)}\right|&\geq \left|B(x,r)\setminus U_*^{(t)}\right|\\&\geq C \left|B(x,r)\right|.
   \end{aligned}
   \end{equation*}
   The last inequality, which is a weak regularity property of $U_*^{(t)}$,
   holds for all $r$ smaller than some $r_0>0$, for some constant $C$ that is
   independent of $r$ and $x$ (see \cite[Prop. 7]{chambolleGeometricPropertiesSolutions2016}). We can in the same way show  \begin{equation*}\left|U_{\varphi(n)}^{(t)}\triangle\,
   U_*^{(t)}\right|\geq C \left|B(x,r)\right|\end{equation*} if $B(x,r)\subset \RR^2\setminus U_{\varphi(n)}^{(t)}$ and hence get the inequality for all $n$. Using that $\underset{n\to+\infty}{\text{lim}}~|U_n^{(t)}\triangle\, U_*^{(t)}|=0$, we get a contradiction.
   \label{remark_conv_level_set}
\end{proof}

\section{Sliding step}
\label{pres_sliding}
Several works \cite{brediesInverseProblemsSpaces2013,
raoForwardBackwardGreedy2015,boydAlternatingDescentConditional2017,
denoyelleSlidingFrankWolfeAlgorithm2019} have advocated for the use of a special final update, which helps identify the sparse structure of the sought-after signal. Loosely speaking, it would amount in our case to running, at the very end of an iteration, the gradient flow of the mapping
\begin{equation}
   (a,E)\mapsto T_{\lambda}\left(\sum\limits_{i=1}^{N^{[k+1]}}a_i \mathbf{1}_{E_i}\right)
   \label{sliding_obj}
\end{equation}
initialized with $(a^{[k+1]},E^{[k+1]})$, so as to find a set of parameters  at which the objective is smaller. Formally, this would correspond\footnote{The formulas given in \eqref{grad_flow} can be formally obtained by using the notion of shape derivative, see \cite[Chapter 5]{henrotShapeVariationOptimization2018}.} to finding a curve
\begin{equation*}t\mapsto (a_i(t),E_i(t))_{i=1}^{N^{[k+1]}}\end{equation*} such that for all $t$
\begin{equation}\left\{
   \begin{aligned}
      a_i'(t)&=-\lambda\left(\text{sign}(a_i(t))\,P(E_i(t))-\int_{E_i(t)}\eta(t)\right),\\
      V_i(t)&=-\lambda\,|a_i(t)|\left(H_{E_i(t)}-\text{sign}(a_i(t))\,\eta(t)\right),\\
   \end{aligned}\right.
   \label{grad_flow}
\end{equation}
where $V_i(t)$ denotes the normal velocity of the boundary of $E_i$ at time $t$ and \begin{equation*}\eta(t)=-\frac{1}{\lambda}\Phi^*\left(\Phi u(t)-y\right),~u(t)=\sum\limits_{i=1}^{N^{[k+1]}} a_i(t)\,\mathbf{1}_{E_i(t)}\,.\end{equation*}
\revision{The study of this gradient flow (existence, uniqueness) is out of the scope of this paper.

For our purpose, it is enough to introduce a sliding step which improves the objective by performing a local descent on
\begin{equation*}
(a,E)\mapsto T_{\lambda}\left(\sum\limits_{i=1}^{N^{[k+1]}}a_i\,\mathbf{1}_{E_i}\right)
\end{equation*} initialized with $(a^{[k+1]},E^{[k+1]})$, that is to find a set of parameters $(a_i,E_i)_{i=1}^{N^{[k+1]}}$ such that~$E_i$ is simple for all $i$ with
\begin{equation}
   T_{\lambda}\left(\sum\limits_{i=1}^{N^{[k+1]}}a_i\,\mathbf{1}_{E_i}\right)\leq T_{\lambda}\left(\sum\limits_{i=1}^{N^{[k+1]}}a_i^{[k+1]}\,\mathbf{1}_{E_i^{[k+1]}}\right).
   \label{obj_decrease}
\end{equation}
The resulting algorithm, which is Algorithm \ref{sfw_sliding}, is a valid application of Algorithm \ref{fw} to \eqref{prob}. Moreover, Line 14 ensures that all convergence guarantees derived for Algorithm \ref{sfw} remain valid.}

\begin{algorithm}
\KwData{measurement operator $\Phi$, observations $y$, regularization parameter $\lambda$}
\KwResult{function $u^*$}
$u^{[0]}\gets 0$\;
$N^{[0]}\gets 0$\;
\While{true}{
$\eta^{[k]}\gets -\frac{1}{\lambda}\Phi^*\left(\Phi u^{[k]}-y\right)$\;
$E_*\gets \underset{E \text{ simple}}{\text{Argmax}}~\frac{\left|\int_E\eta^{[k]}\right|}{P(E)} \text{ s.t. } 0<|E|<+\infty$\;
\eIf{$\left|\int_{E_*}\eta^{[k]}\right|\leq P(E_*)$}{
output $u^*\gets u^{[k]}$, which is optimal\;}
{$E^{[k+1]}\gets(E_1^{[k]},...,E_{N^{[k]}}^{[k]},E_*)$\;
$a^{[k+1]}\gets \underset{a\in\mathbb{R}^{N^{[k]}+1}}{\text{argmin}}
T_{\lambda}\left(\sum\limits_{i=1}^{N^{[k]}+1}a_i
\,\mathbf{1}_{E^{[k+1]}_i}\right)$\;
remove atoms with zero amplitude\;
$N^{[k+1]}\gets$ number of atoms in $E^{[k+1]}$\;
perform a local descent on
$(a,E)\mapsto T_{\lambda}\left(\sum\limits_{i=1}^{N^{[k+1]}}a_i\,\mathbf{1}_{E_i}\right)$ initialized with $(a^{[k+1]},E^{[k+1]})$\;
repeat the operations of Lines 10-12\;
$u^{[k+1]}\gets\sum\limits_{i=1}^{N^{[k+1]}}a^{[k+1]}_i\,\mathbf{1}_{E^{[k+1]}_i}$\;}
}
\caption{modified Frank-Wolfe algorithm applied to \eqref{prob} (with sliding)}
\label{sfw_sliding}
\end{algorithm}

The sliding step (Line 13 of Algorithm \ref{sfw_sliding}) was first introduced in
\cite{brediesInverseProblemsSpaces2013}. It allows in practice to considerably improve the convergence speed of the algorithm, and also produces sparser solutions: if the solution is expected to be a linear combination of a few indicator functions, removing the sliding step will typically produce iterates made of a much larger number of indicator functions, the majority of them correcting the crude approximations of the support of the solution made during the first iterations.

In \cite{denoyelleSlidingFrankWolfeAlgorithm2019}, the introduction of this step allowed the authors to derive improved convergence guarantees (i.e. finite time convergence) in the context of sparse spikes recovery. Their proof relies on the fact that at every iteration, a ``critical point'' of the objective can be reached at the end of the sliding step. In our case, the above mentioned existence issues make the adaptation of these results difficult. However, if the existence of a curve (formally) satisfying \eqref{grad_flow} could be guaranteed for all times, then one would expect it to converge when $t$ goes to infinity to a critical point of the mapping defined in \eqref{sliding_obj}, in the sense of the following definition.
\begin{definition}
  Let $N\in\mathbb{N}^*$, $a\in\RR^N$ and $E_1,...,E_N$ be subsets of $\RR^2$ such that $|E_i|<+\infty$, ${P(E_i)<+\infty}$ for all~${i\in\{1,...,N\}}$ and \eqref{tv_sum} holds. We say that $(a_i,E_i)_{i=1}^N$ is a critical point of the mapping
  \begin{equation*}(a,E)\mapsto T_{\lambda}\left(\sum\limits_{i=1}^{N}a_i \mathbf{1}_{E_i}\right)\end{equation*} if for all $i\in\{1,...,N\}$ we either have $a_i\neq 0$ and
  \begin{equation}\left\{
     \begin{aligned}
        P(E_i)&=\mathrm{sign}(a_i)\,\int_{E_i}\eta\,, \\
        H_{E_i}&=\mathrm{sign}(a_i)~\eta\,,
     \end{aligned}\right.
     \label{opti_cond_sliding}
  \end{equation}
 or $a_i=0$ and~${\left|\int_{E_i}\eta\right|\leq P(E_i)}$, where \begin{equation*}\eta\eqdef-\frac{1}{\lambda}\Phi^*\left(\Phi u-y\right)\,,~u\eqdef\sum\limits_{i=1}^N a_i\mathbf{1}_{E_i}\,.\end{equation*}
\label{critical_point}
\end{definition}
In \Cref{remark_haussdorff}, we discuss how assuming a critical point is indeed reached at the end of the sliding step for every iteration could be used to derive additional properties of sequences produced by Algorithm \ref{sfw_sliding}. We stress that if, for a given iteration, a critical point is reached at the end of the sliding step, then Line 14 can be skipped, since the first equality in \eqref{opti_cond_sliding} and the inequality given above in the case of a zero amplitude ensure $a^{[k+1]}$ is already optimal for the problem to be solved.

\begin{remark}
   \label{remark_haussdorff}
   As mentioned above, the introduction of the sliding step is supposed to allow the derivation of improved convergence properties. If its output is a critical point in the sense of \Cref{critical_point}, a first remark we can make is that for all $i\in\{1,...,N^{[k]}\}$ the set $E_i^{[k]}$ has distributional curvature
   $\text{sign}(a_i^{[k]})\,\eta^{[k]}$. This can be exploited to obtain ``uniform'' density estimates for the level sets of~$u^{[k]}$ in the spirit of \cite[Corollary 17.18]{maggiSetsFinitePerimeter2012}. One could then wonder whether this weak regularity of the level sets could be used to prove \begin{equation}
   \underset{n\to+\infty}{\limsup}~\partial U_n^{(t)}\subseteq\partial U_*^{(t)}\,,\label{lim_sup}
   \end{equation}
   where
   \begin{equation*}
   \underset{n\to+\infty}{\limsup}~\partial
   U_n^{(t)}\eqdef\big\{x\in\RR^2\,\big\rvert\,\underset{n\to+\infty}{\liminf}~\mathrm{dist}(x,\partial U_n^{(t)})=0\big\}\,.
   \end{equation*}
   This, combined with the result of \Cref{coroll_conv_level_set} and the fact $(\partial
   U_n^{(t)})_{n\geq 0}$ is uniformly bounded would mean that
   \begin{equation*}
      \underset{n\to +\infty}{\lim}~\partial U_n{(t)}=\partial U_*^{(t)}
   \end{equation*}
   in the Hausdorff sense (see \cite{rockafellarVariationalAnalysis1998} for
   more details).

   A major obstacle towards this result is that, although \Cref{lemma_bounds} provides a uniform upper bound on the perimeter of the atoms involved in the definition of the iterates, to
   our knowledge, it does not seem possible to derive such a bound for the perimeter of their level sets, which prevents one from using the potential weak-* convergence of $\Diff\mathbf{1}_{U_n^{(t)}}$ towards $\Diff\mathbf{1}_{U_*^{(t)}}$.
   \end{remark}

\section{Implementation}
\label{implementation}

The implementation\footnote{A complete implementation of Algorithm \ref{sfw_sliding} can be found online at \texttt{https://github.com/rpetit/tvsfw} (see also \texttt{https://github.com/rpetit/PyCheeger}).} of Algorithm \ref{sfw_sliding} requires two oracles to carry out the operations described on Lines 5 and 13 \revision{(recall Line 10 can always be reduced to a LASSO-type problem which can efficiently be solved by existing solvers)}: a first one that, given a weight function~$\eta$, returns a solution of
\eqref{cheeger},
and a second one that, given a collection of real numbers and simple sets, returns another
such collection with a lower objective value. Our approach for designing these oracles relies on polygonal approximations: we fix an integer~${n\geq 3}$ (that might be iteration-dependent), look for a maximizer of $\mathcal{J}$ defined by
\begin{equation*}\mathcal{J}(E)\eqdef\frac{1}{P(E)}\left|\int_{E}\eta\right|\end{equation*} among simple polygons with at most $n$ sides, and perform the sliding step by finding a collection of real numbers and simple polygons satisfying \eqref{obj_decrease}. This choice is mainly motivated by our goal to solve \eqref{prob} ``off-the-grid'', which naturally leads us to consider purely Lagrangian methods which do not rely on the introduction of a pre-defined discrete grid.

\subsection{Polygonal approximation of Cheeger sets}
\label{polygonal_approx_cheeger}
\revision{In the following, we fix an integer $n\geq 3$ and denote
\begin{equation*}
   \mathcal{X}_n=\left\{x\in\RR^{n\times 2}\,\big\rvert\,[x_1,x_2],...,[x_n,x_1]\text{ is simple}\right\}.
\end{equation*}
We recall that a polygonal curve is said to be simple if non-adjacent sides do not intersect. If $x\in\mathcal{X}_n$, then\footnote{If $i>n$ we define $x_i\eqdef x_{i\,\text{mod}\,n}$, i.e. $x_{n+1}=x_1$.}~${\cup_{i=1}^n[x_i,x_{i+1}]}$ is a Jordan curve. It hence divides the plane in two regions, one of which is bounded. We denote this region $E_x$ (it is hence a simple polygon). When $x$ spans $\mathcal{X}_n$, $E_x$ spans $\mathcal{P}_n$, the set of simple polygons
with at most $n$
sides. The sets we wish to approximate in this section (in order to carry out Line 5 in Algorithm \ref{sfw_sliding}) are the
maximizers of $\mathcal{J}$ over $\mathcal{P}_n$. We prove their existence in \Cref{ex_cheeger_polygon}.}

The approximation method presented thereafter consists of several steps. First, we solve a discrete version of \eqref{cheeger_relaxed}, where the minimization is performed over the set of piecewise constant functions on a fixed grid. Then, we extract a level set of the solution, and obtain a simple polygon whose edges are located on the edges of the grid. Finally, we use a first order method initialized with the previously obtained polygon to locally maximize $\mathcal{J}$.

\subsubsection{Fixed grid step}
Every solution of \eqref{cheeger_relaxed} has its support included in some ball (indeed if $u$ solves \eqref{cheeger_relaxed}, then there exists $\alpha$ such that~${\alpha\,\eta\in \partial J(u)}$, and the result follows
from \Cref{lvlset_sousdiff} and \Cref{lemma_bounds}). We can hence solve \eqref{cheeger_relaxed} in~${[-R,R]^2}$ \revision{(with Dirichlet boundary conditions)} for a sufficiently large $R>0$. We now
proceed as in \cite{carlierApproximationMaximalCheeger2009}. Let~$N$ be a positive integer and~${h\eqdef 2R/N}$. We denote~$E^h$ the set of $N$ by $N$ matrices. For every matrix~${u=(u_{i,j})_{(i,j)\in[1,N]^2}\in E^h}$ we define
\begin{equation}
   \partial_x^h u_{i,j}\eqdef
      u_{i+1,j}-u_{i,j}
   ~~~~~
   \partial_y^h u_{i,j}\eqdef
      u_{i,j+1}-u_{i,j}
      \label{discrete_grad}
\end{equation}
for all $(i,j)\in[0,N]^2$,
with the convention $u_{i,j}=0$ if either $i$ or $j$ is in $\{0,N+1\}$. We now define
\begin{equation*}
   \nabla^h u_{i,j}\eqdef\left(\partial_x^h u_{i,j},\partial_y^h u_{i,j}\right),
\end{equation*} and set
\begin{equation*}
   J^h(u)\eqdef h\sum\limits_{i=0}^N\sum\limits_{j=0}^N ||\nabla^h u_{i,j}||_2=h\,\|\nabla^h\,u\|_{2,1}\,.
\end{equation*}

We then solve the following discretized version of~\eqref{cheeger_relaxed} for increasingly small values of $h$
\begin{equation}
    \underset{u \in E^h}{\text{min}}~ h^2\,\langle \overline{\eta}^h,u\rangle ~\text{ s.t. }~ J^h(u)\leq 1\,,
  \label{cheeger_relaxed_discrete}
\end{equation}
where $\overline{\eta}^h=\left(\frac{1}{h^2}\int_{C^h_{i,j}}\eta\right)_{(i,j)\in[1,N]^2}$ and $(C^h_{i,j})_{(i,j)\in[ 1,N]^2}$ is a partition of $[-R,R]^2$ composed of squares of equal size, i.e.
\begin{equation*}
   C^h_{i,j}\eqdef[-R+(i-1)h,-R+ih]\times [-R+(j-1)h,-R+jh]\,.
\end{equation*}
For convenience reasons, we will also use the above expression to define $C^h_{i,j}$ if $i$ or $j$ belongs to $\{0,N+1\}$.

In practice we solve \eqref{cheeger_relaxed_discrete} using the primal-dual algorithm introduced in \cite{chambolleFirstOrderPrimalDualAlgorithm2011}: we take $(\tau,\sigma)$
such that $\tau\,\sigma\,\|D\|^2<1$ holds with $D\eqdef h\nabla^h$ and define
\begin{equation}
   \left\{
   \begin{aligned}
      \phi^{n+1}&=\text{prox}_{\sigma  \|\cdot\|_{2,\infty}}(\phi^n+\sigma\,D\bar{u}^n)\,, \\
      u^{n+1}&=(u^n-\tau\,D^*\phi^{n+1})-\tau\,h^2\,\bar{\eta}^h\,, \\
      \bar{u}^{n+1}&=2\,u^{n+1}-u^n\,,
   \end{aligned}\right.
\end{equation}
where $\text{prox}_{\sigma \|\cdot\|_{2,\infty}}$ is given by:
\begin{equation*}
   \text{prox}_{\sigma \|\cdot\|_{2,\infty}}(\phi)=\phi - \sigma\,\text{proj}_{\{\|\cdot\|_{2,1}\leq 1\}}\left(\frac{\phi}{\sigma}\right).
\end{equation*}
See \cite{condatFastProjectionSimplex2016a} for the computation of the
projection onto the $(2,1)$-unit ball.

The following proposition shows that, when the grid becomes finer,
solutions of \eqref{cheeger_relaxed_discrete} converge to a solution of \eqref{cheeger_relaxed}. Its proof is almost the same as the one of \cite[Theorem 4.1]{carlierApproximationMaximalCheeger2009}. Since the latter however gives a slightly different result about the minimization of a quadratic objective (linear in our case) on the total variation unit ball, we decided to include it in \Cref{proof_cheeger_grid} for the sake of completeness.
\begin{proposition}
   \label{cheeger_grid}
   Let $u^h$ be the piecewise constant function on~$(C^h_{i,j})_{(i,j)\in[ 1,N]^2}$, extended to $0$ outside $[-R,R]^2$ associated to a solution of \eqref{cheeger_relaxed_discrete}. Then there exists a (not relabeled) subsequence converging strongly in $\mathrm{L}^1(\RR^2)$ and weakly in $\LD$ to a solution
   $u^*$ of \eqref{cheeger_relaxed} when $h\to 0$, with moreover $\Diff u^h\overset{\ast}{\rightharpoonup}\Diff u$.
\end{proposition}

Since we are interested in finding a simple set $E$ that approximately solves \eqref{cheeger}, and now have a good way of approximating solutions of \eqref{cheeger_relaxed}, we make use of the following result:
\begin{proposition}
   Let $u$ be a solution of \eqref{cheeger_relaxed}. Then the level sets of~$u$ are such that for all $t\in\RR^*$ with~${|U^{(t)}|>0}$, the set $U^{(t)}$ solves \eqref{cheeger}.
\end{proposition}
\begin{proof}
   This is a direct consequence of \Cref{lvlset_sousdiff}.
\end{proof}
If we have $v^h$ converging strongly in $\mathrm{L}^1(\RR^2)$ to a solution $v^*$ of \eqref{cheeger_relaxed}, then up to the extraction of a (not relabeled) subsequence, for almost every $t\in\RR$ we have that \begin{equation*}\underset{h\to 0}{\text{lim}}~\left|V_h^{(t)}\triangle V_*^{(t)}\right|=0\,.\end{equation*}
The above results hence show we can construct a sequence of sets
$(E_k)_{k\geq 0}$ such that $|E_k\triangle E_*|$ converges to~$0$, with $E_*$
a solution of \eqref{cheeger}. However, this convergence only implies that \begin{equation*}
\underset{k\to\infty}{\text{lim sup}}~\mathcal{J}(E_k)\leq \mathcal{J}(E_*)\,,
\end{equation*}
and given that $E_k$ is a union of squares this inequality is likely to be
strict, with the perimeter of~$E_k$ not converging to the perimeter
of $E_*$. From \Cref{approx_jaggi}, we know we have to design a numerical
method that allows to find a set at which the value of $\mathcal{J}$ is
arbitrarily close to $\mathcal{J}(E_*)$. This hence motivates the
introduction of the refinement step described in the next subsection.

\revision{As a final remark, we note that, even for $k$ large enough, $E_k$ could be non-simple. However, \revisionbis{using the notations of \Cref{preli_sets}}, since for every set of finite perimeter $E$,
$\mathcal{J}(E)$ is a convex combination of the \begin{equation*}
\left(\mathcal{J}(\mathrm{int}(\gamma_i^+))\right)_{i\in I},~\left(\mathcal{J}(\mathrm{int}(\gamma_{i,j}^-))\right)_{i\in I,j\in J_i}\,,
\end{equation*}
there is a simple set $F$ in the decomposition of
$E$ which is such that $\mathcal{J}(F)\geq \mathcal{J}(E)$. In practice, such a set can be found by extracting all the contours of the binary image~$\mathbf{1}_{E}$, and finding the one with highest objective value. This procedure guarantees that the output of the fixed grid step is a simple polygon. We stress that in all our experiments, $v^h$ is close to being (proportional to) the indicator of a simple set for~$h$ large enough, so that its non-trivial level sets are all simple.}

\subsubsection{Refinement step}
\label{refinement_cheeger}

We use a shape gradient algorithm (see \cite{allaireChapterShapeTopology2021}) to refine the output of the fixed grid step. It
consists in iteratively constructing a sequence of simple polygons by finding at each step a displacement of steepest ascent for $\mathcal{J}$, along which the vertices of the previous polygon are moved. Given \revision{$x^t\in\mathcal{X}_n$} and a step size~$\alpha^t$, we define the next iterate by:
\begin{equation}
   \begin{aligned}
      x_j^{t+1}&\eqdef x_j^t + \alpha^t\,\theta_j^t\,,\\
      \theta_j^t &\eqdef \frac{1}{P(E_{x^t})}
      \left(\theta_{\text{area},j}^t
      -\frac{\int_{E_{x^t}}\eta}{P(E_{x^t})}\,\theta_{\text{per},j}^t\right)\,, \\
      \theta_{\text{area},j}^t&\eqdef w_{j}^{t-}\nu_{j-1}^t+w_{j}^{t+}\nu_{j}^t \,,\\
      \theta_{\text{per},j}^t&\eqdef-(\tau_{j}^t-\tau_{j-1}^t)\,,
   \end{aligned}
   \label{cheeger_update}
\end{equation}
where, for all $j$, $\tau_j^t$ and $\nu_j^t$ are respectively the unit tangent and outer normal vectors on $[x_j^t,x_{j+1}^t]$ and
\begin{equation*}
\begin{aligned}
w_{j}^{t+}&\eqdef\int_{[x_j^t,x_{j+1}^t]}\eta(x)~
\frac{||x-x_{j+1}^t||}{||x_{j}^t-x_{j+1}^t||}\,d\mathcal{H}^1(x)\,, \\
w_{j}^{t-}&\eqdef\int_{[x_{j}^t,x_{j-1}^t]}\eta(x)~
\frac{||x-
x_{j-1}^t||}{||x_{j}^t-x_{j-1}^t||}\,d\mathcal{H}^1(x)\,.
   \end{aligned}
\end{equation*}
One can actually show that the displacement $\theta^t$ we apply to the vertices of $E_{x^t}$ is such that
\begin{equation}
   \theta^t=\underset{||\theta||\leq 1}{\text{Argmax}}~~\underset{\alpha\to 0^+}{\text{lim}}~\frac{\mathcal{J}(E_{x^t+\alpha\theta})-\mathcal{J}(E_{x^t})}{\alpha}\,,
   \label{steepest_descent}
\end{equation}
i.e. that it is the displacement of steepest ascent for $\mathcal{J}$ at $E_{x^t}$. We provide a proof of this result in \Cref{shape_grad}.

To compute the integral of $\eta$ over $E_{x^t}$, we integrate~$\eta$ on each triangle of a sufficiently fine triangulation of~$E_{x^t}$ (this triangulation must be updated at each iteration, and sometimes re-computed from scratch to avoid the presence of ill-shaped triangles). The integral of $\eta$ on a triangle and $w_{j}^{t+}$, $w_{j}^{t-}$ are computed using standard numerical integration schemes for triangles and line segments. \revision{If $|\mathcal{T}|$ denotes the number of triangles in the triangulation of $E_{x^t}$, $|\mathcal{S}_T|$ (resp. $|\mathcal{S}_L|$) the number of points used in the numerical integration scheme for triangles (resp. line segments), the complexity of each iteration is of order $\mathcal{O}\left(m\left(|\mathcal{T}|\,|\mathcal{S}_T|+n\,|\mathcal{S}_L|\right)\right)$.}

\paragraph{Comments.} Two potential concerns about the above procedure are whether the iterates remain simple polygons \revision{(i.e. $x^t\in\mathcal{X}_n$ for all $t$)} and whether they converge to a global maximizer of $\mathcal{J}$ \revision{over $\mathcal{P}_n$}. We could not prove that the iterates remain simple polygons along the process, but since the initial polygon can be taken
arbitrarily close to a simple set solving \eqref{cheeger} (in terms of the Lebesgue measure of the symmetric difference), we do not expect nor observe in practice any change of topology during the optimization. Moreover, even if $\mathcal{J}$ could have non-optimal critical points\footnote{Here critical point is to be understood in the sense that the limit appearing in \eqref{steepest_descent} is equal to zero for every $\theta$.}, the above initialization allows us to start our local descent with a polygon that hopefully lies in the basin of attraction of a global maximizer. Additionally, we stress again that to carry out Line 5 of Algorithm \ref{sfw_sliding}, thanks to \Cref{approx_jaggi}, we only need to find a set with near optimal value in \eqref{cheeger}.

An interesting problem is to quantify the distance (e.g. in the Hausdorff sense) of a maximizer of $\mathcal{J}$ \revision{over~$\mathcal{P}_n$} to a maximizer of $\mathcal{J}$. We discuss in Section \ref{toy} the simpler case of radial measurements. In the general case, if the sequence of polygons defined above converges to a simple polygon $E_x$, then $E_x$ is such that
\begin{equation}
   w_{j}^+=w_{j}^-=\frac{\int_{E_x} \eta}{P(E_x)}~\text{tan}\left(\frac{\theta_j}{2}\right)
	\label{polygon_opt_cheeger}
\end{equation}
for all $j$, where $\theta_j$ is the $j$-th exterior
angle of the polygon (the angle between $x_{j}-x_{j-1}$ and $x_{j+1}-x_j$). This can be seen as a discrete version of the following first order optimality condition for solutions of~\eqref{cheeger}:
\begin{equation}
\eta=\frac{\int_{E}\eta}{P(E)}~H_E \text{ on } \partial^*E\,.
\label{opt_cheeger}
\end{equation}
Note that \eqref{opt_cheeger} is similar to the optimality condition for the classical Cheeger problem (i.e.\ with~${\eta=1}$ and the additional constraint $E \subseteq \Omega$), namely $H_E= P(E)/\abs{E}$ in the free boundary of $E$ (see \cite{alterEvolutionCharacteristicFunctions2005} or \cite[Prop. 2.4]{pariniIntroductionCheegerProblem2011}).

\subsection{Sliding step}
The implementation of the sliding step (Line 13 in Algorithm~\ref{sfw_sliding}) is similar to what is described above for refining crude approximations of Cheeger sets. We use a first order optimization method on the mapping
\begin{equation}
   (a,x)\mapsto T_{\lambda}\left(\sum\limits_{i=1}^N a_i\,\mathbf{1}_{E_{x_i}}\right).
   \label{discrete_sliding_obj}
\end{equation}
Given a step size $\alpha^t$, a vector $a^t\in\mathbb{R}^{N}$ and \revision{$x_1^t,...,x_N^t$ in~$\mathcal{X}_n$}, we set~$u^t\eqdef\sum_{i=1}^N a^t_i\,\mathbf{1}_{E_{x_i^t}}$ and perform the following update:
\begin{equation*}
 \begin{aligned}
a_i^{t+1}&\eqdef a^t_i-
\alpha^t\,h_i^t\,,\\
h_i^t&\eqdef \left\langle\Phi\mathbf{1}_{E_{x_i^t}},\Phi u^t -
y\right\rangle+\lambda\,P\left(E_{x_i^t}\right)~\text{sign}\left(a_i^t\right
),\\
x_{i,j}^{t+1}&\eqdef x_{i,j}^t-
\alpha^t\,\theta^t_{i,j}\,, \\
\theta^t_{i,j}&\eqdef a_i^t\left[
\theta^t_{\text{data},i,j} -
\lambda\,\text{sign}(a_i^t)\left(\tau_{i,j}^t-\tau_{i,j-1}^t\right)\right],\\
\theta^t_{\text{data},i,j}&\eqdef \langle\Phi u^t-y,\,{w_{i,j}^{t-}}\rangle \,\nu_{i,j-1}^t+\langle\Phi u^t-y,\, w_{i,j}^{t+}\rangle\,\nu_{i,j}^t\,,
\end{aligned}
\end{equation*}
where $\tau_{i,j}^t$, $\nu_{i,j}^t$ are respectively the unit tangent and outer normal vectors on the edge $[x_{i,j}^t,x_{i,j+1}^t]$ and
\begin{equation*}
\begin{aligned}
w_{i,j}^{t+}&\eqdef\int_{[x_{i,j}^t,x_{i,j+1}^t]}\varphi(x)~
\frac{||x-x_{i,j+1}^t||}{||x_{i,j}^t-x_{i,j+1}^t||}\,d\mathcal{H}^1(x)\,, \\
w_{i,j}^{t-}&\eqdef\int_{[x_{i,j}^t,x_{i,j-1}^t]}\varphi(x)~
\frac{||x-
x_{i,j-1}^t||}{||x_{i,j}^t-x_{i,j-1}^t||}\,d\mathcal{H}^1(x)\,.
   \end{aligned}
\end{equation*}

\revision{Using the notations of \Cref{refinement_cheeger}, the complexity of each iteration is of order $\mathcal{O}\left(N\,m\left(|\mathcal{T}|\,|\mathcal{S}_T|+n\,|\mathcal{S}_L|\right)\right)$.}

\paragraph{Comments.}
\revision{We first stress that the above update is similar to the evolution formally described in \eqref{grad_flow}.} Now, unlike the local optimization we perform to approximate Cheeger sets, the sliding step may tend to induce topology changes (see \Cref{topo_change} for an
example). This is of course linked to the possible appearance of singularities  mentioned in \Cref{pres_sliding}. Typically, a simple set may tend to split in two simple sets over the
course of the descent. This is a major difference (and challenge) compared to
the sliding steps used in sparse spikes recovery (where the optimization is
carried out over the space of Radon measures) \cite{brediesInverseProblemsSpaces2013,boydAlternatingDescentConditional2017,denoyelleSlidingFrankWolfeAlgorithm2019}.
This phenomenon is closely linked to topological properties of the faces of the total (gradient) variation unit ball: its extreme points do not
form a closed set for any reasonable topology (e.g. the weak $\LD$ topology),
nor do its faces of dimension $d\leq k$ for any~${k \in  \NN}$.
As a result, when moving continuously on the set of faces of dimension $d=k$, it is possible to ``stumble upon'' a point which only belongs to a  face of dimension~$d>k$.

Our current implementation does not allow to handle these topology changes in a consistent way, and finding a way to deal with them ``off-the-grid'' is an interesting avenue for future research. It is important to note that not allowing topological changes during the sliding step is not an issue, since all convergence guarantees hold as soon as the output of the sliding step decreases the energy more than the standard update. One can hence stop the local descent at any point before any change of topology occurs, which avoids having to treat them. Still, in order to yield iterates that are as sparse as possible (and probably to decrease the objective as quickly as possible), it seems preferable to allow topological changes.

\section{Numerical experiments}

\subsection{Recovery examples}

Here, we investigate the practical performance of Algorithm \ref{sfw_sliding}. We focus on the case where $\Phi$ is a sampled Gaussian convolution operator, i.e.
\begin{equation*}
   \forall x\in\mathbb{R}^2,~\varphi(x)= \left(\text{exp}\left(-\frac{||x-x_i||^2}{2\sigma^2}\right)\right)_{i=1}^m
\end{equation*}
for a given $\sigma> 0$ and a sampling grid $(x_i)_{i=1}^m$. The noise is drawn from a multivariate Gaussian with zero mean and isotropic covariance matrix $\tau^2\,I_m$. We take $\lambda$ of the order of $\sqrt{2\,\text{log}(m)\,\tau^2}$.

Numerically certifying that a given function is an approximate solution of~\eqref{prob} is difficult. However, as the sampling grid becomes finer, $\Phi$ tends to the convolution with the Gaussian kernel, which is injective. Relying on a $\Gamma$-convergence argument, one may expect that if $u_0$ is a piecewise constant image and~$w$ is some small additive noise, the solutions of \eqref{prob} with~${y=\Phi u_0+w}$ are all close to $u_0$, modulo the regularization effects of the total variation.

We also assess the performance of our algorithm by comparing its output to that of a primal dual algorithm minimizing a discretized version of~\eqref{prob} on a pixel grid, \revisionbis{where the total variation term is replaced by the discrete isotropic total variation or Condat's discrete total variation}\footnote{\revisionbis{Condat's total variation is introduced in \cite{condatDiscreteTotalVariation2017}. See also \cite{chambolleChapterApproximatingTotal2021} for a review of discretizations of the total variation.}}. \revision{To minimize discretization artifacts, we artificially introduce a downsampling in the forward operator, so that the reconstruction is performed on a grid four times larger than the sampling one.}

Our first experiment consists in recovering a function $u_0$ that is a linear combination of three indicator functions (see \Cref{exp1_res,exp1_res_unfolding}). During each of the three iterations required to obtain a good approximation of~$u_0$, a new atom is added to its support. One can see the sliding step is crucial: the large atom on the left, added during the second iteration, is significantly refined during the sliding step of the third iteration, when enough atoms have been introduced.

\begin{figure*}
	\centering
	\includegraphics[width=0.99\textwidth]{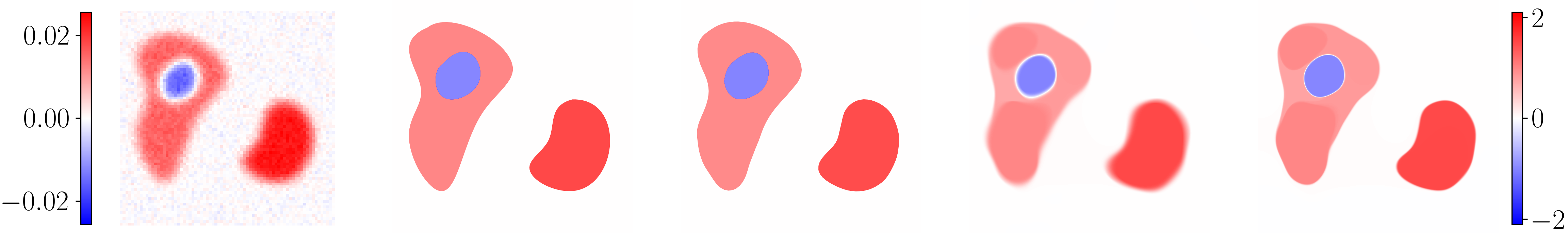}
	\caption{From left to right: observations, unknown function, output of Algorithm \ref{sfw_sliding}, outputs of the fixed grid method using the isotropic and Condat's total variation}
	\label{exp1_res}
\end{figure*}

\begin{figure*}
	\centering
	\includegraphics[width=0.99\textwidth]{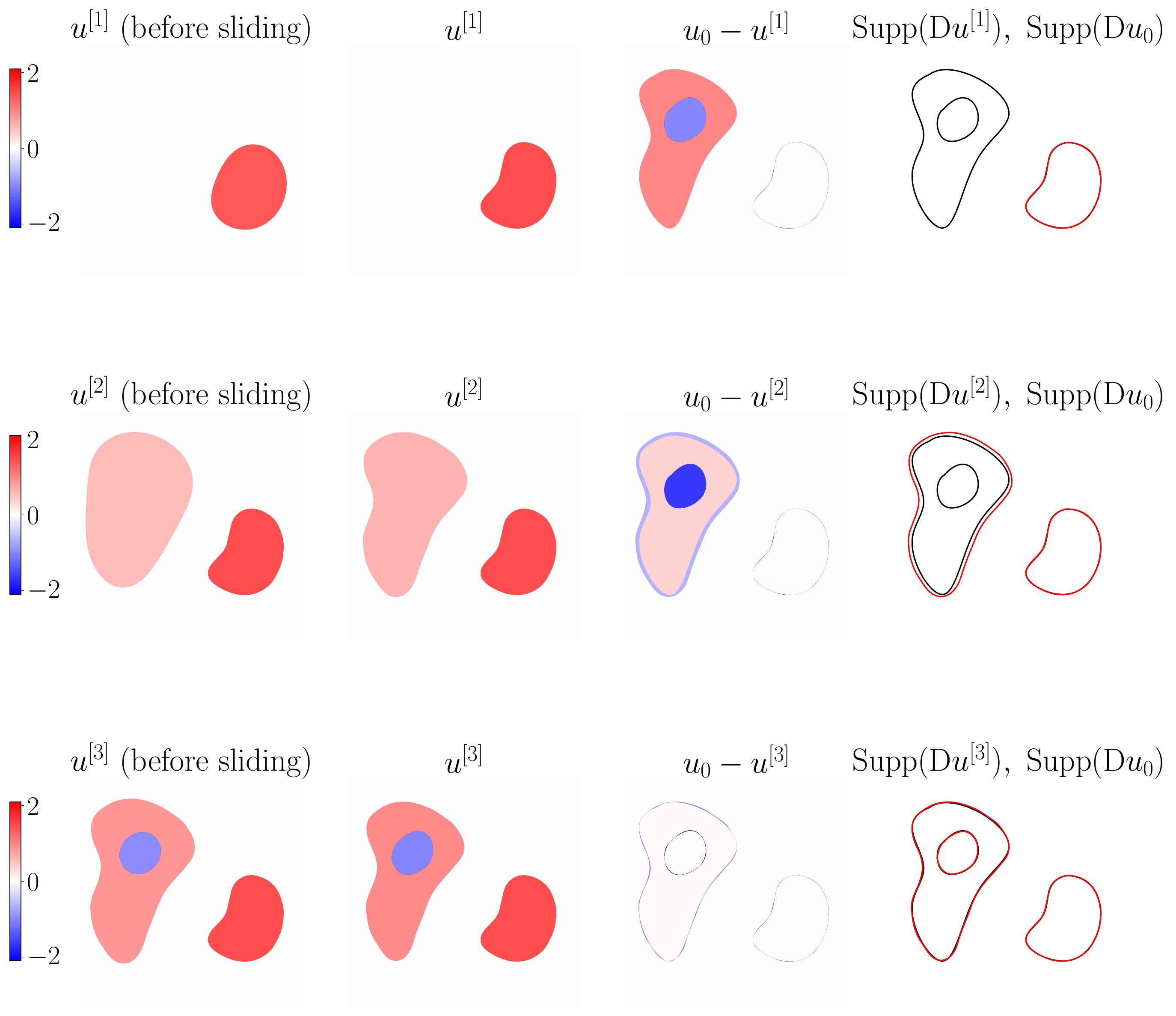}
	\caption{Unfolding of Algorithm \ref{sfw_sliding} for the first experiment ($u^{[k]}$ denotes the $k$-th iterate)}
	\label{exp1_res_unfolding}
\end{figure*}

The second experiment (see Figure \ref{pacman}) consists in recovering the indicator function of a set with a hole (which can also be seen as the sum of two indicator functions of simple sets). The support of $u_0$ and its gradient are accurately estimated. Still, the typical effects of total (gradient) variation regularization are noticeable: corners are slightly rounded, and there is a ``loss of contrast'' in the eye of the pacman.

\begin{figure*}
	\centering
   \newcommand\x{0.165}
	\subfloat{{\includegraphics[height=\x\textwidth]{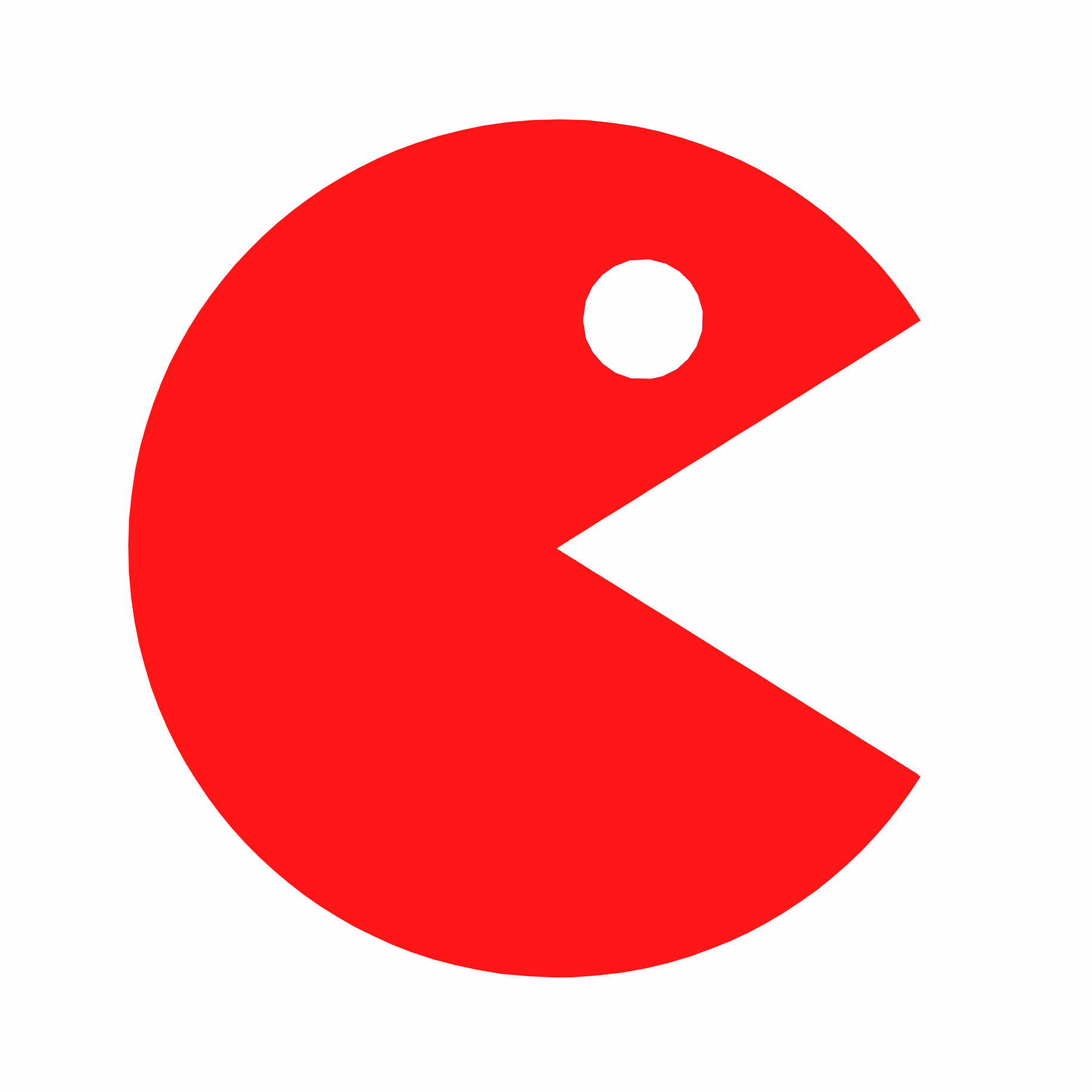}}}
	\subfloat{{\includegraphics[height=\x\textwidth]{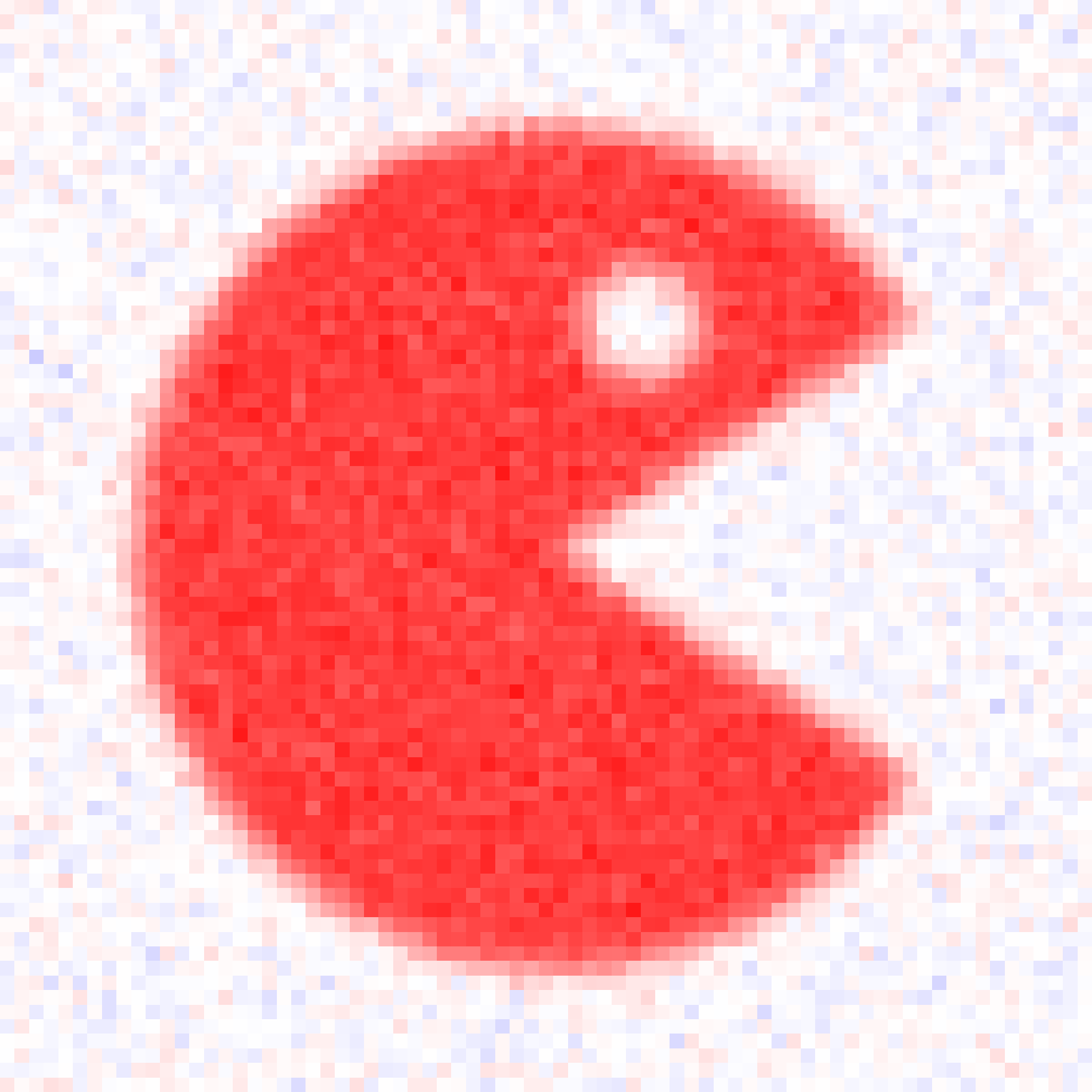}}}
	\subfloat{{\includegraphics[height=\x\textwidth]{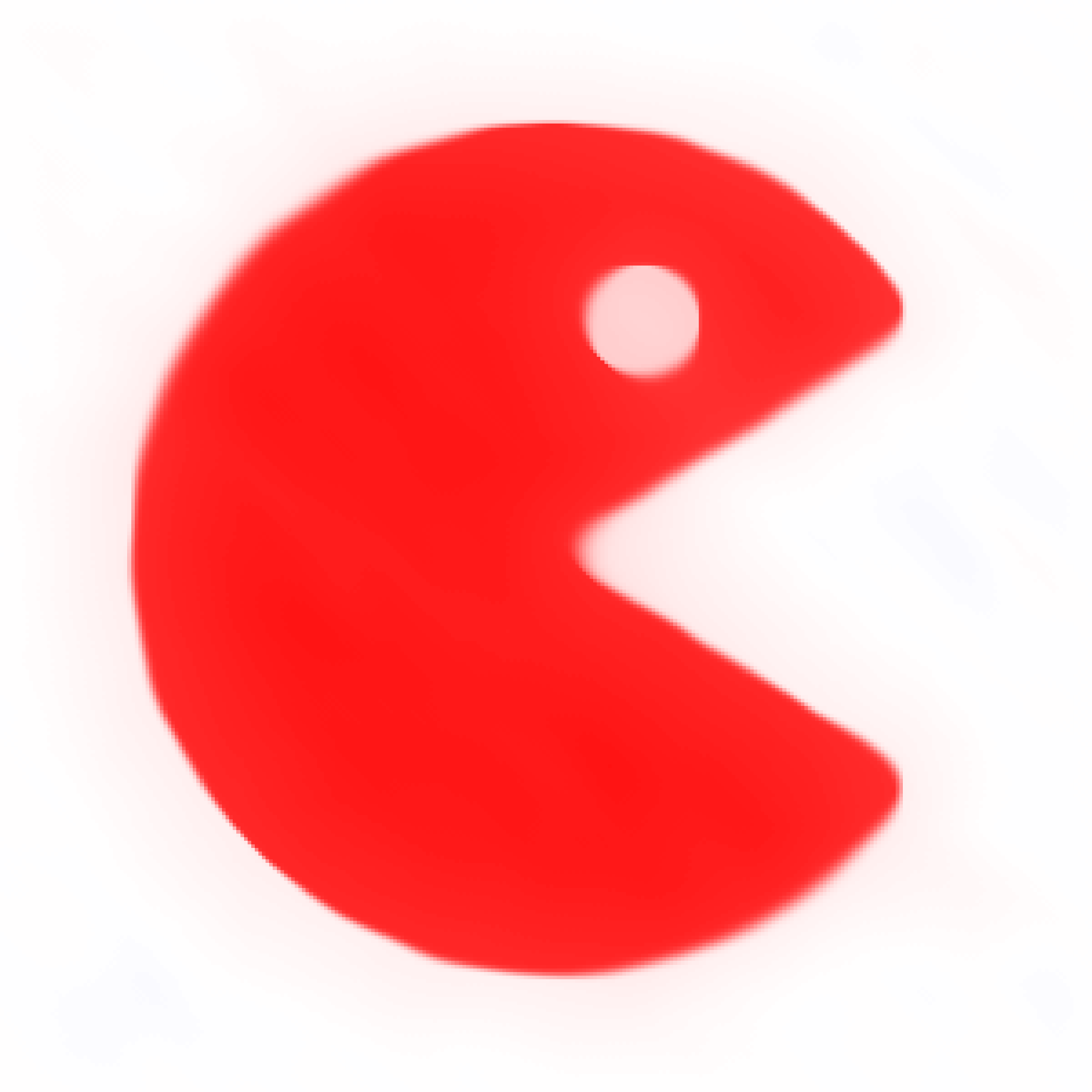}}}
   \subfloat{{\includegraphics[height=\x\textwidth]{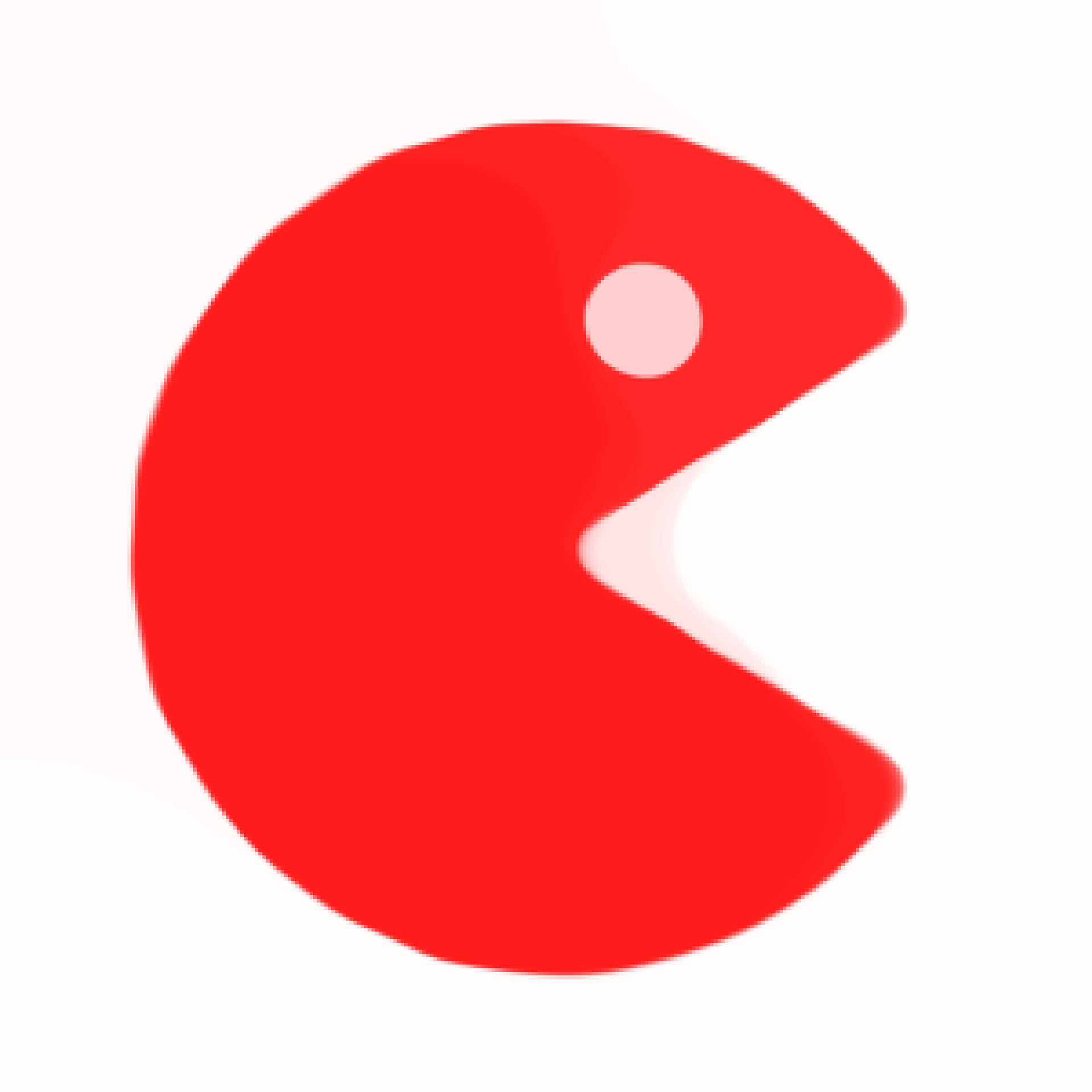}}}
	\subfloat{{\includegraphics[height=\x\textwidth]{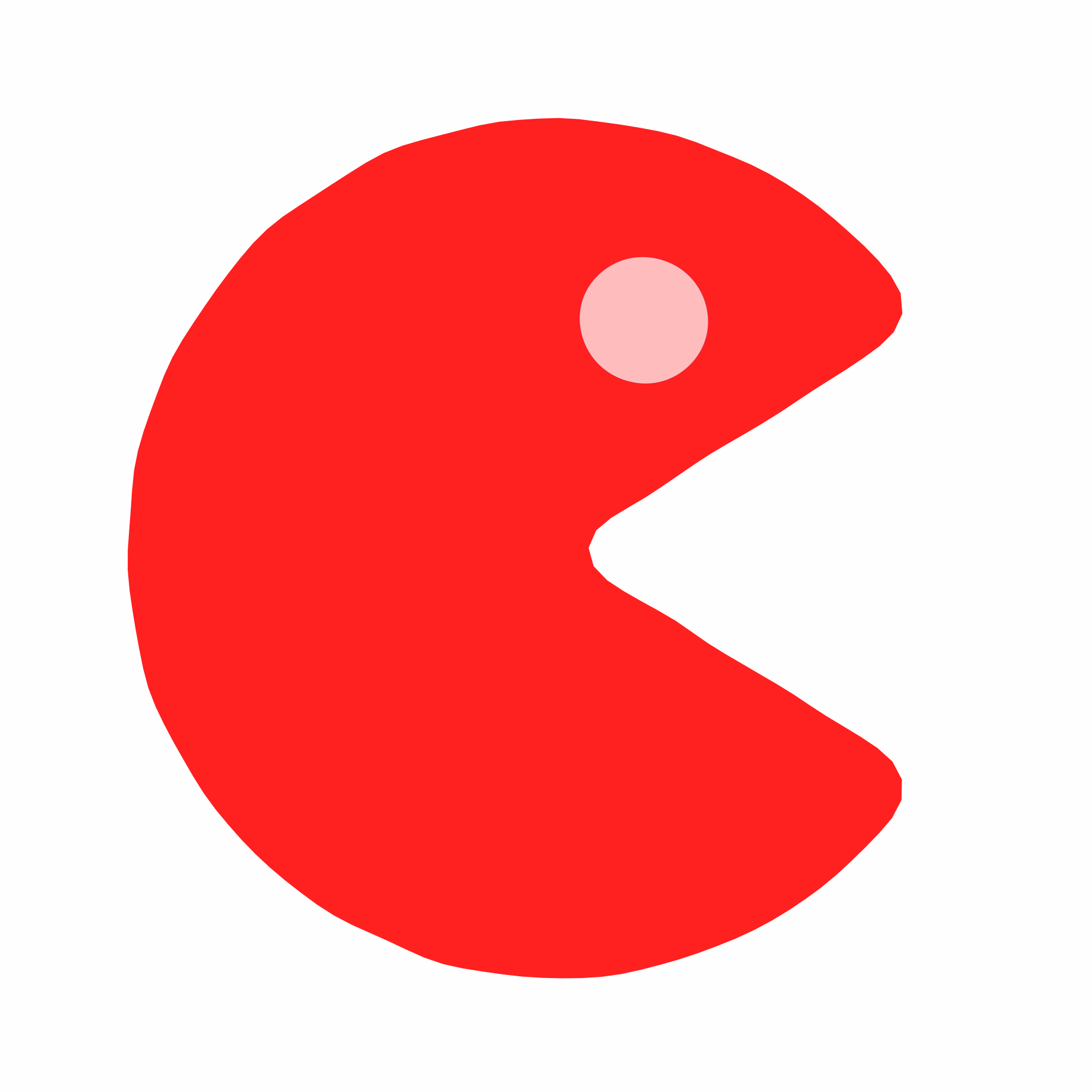}}}
	\subfloat{{\includegraphics[height=\x\textwidth]{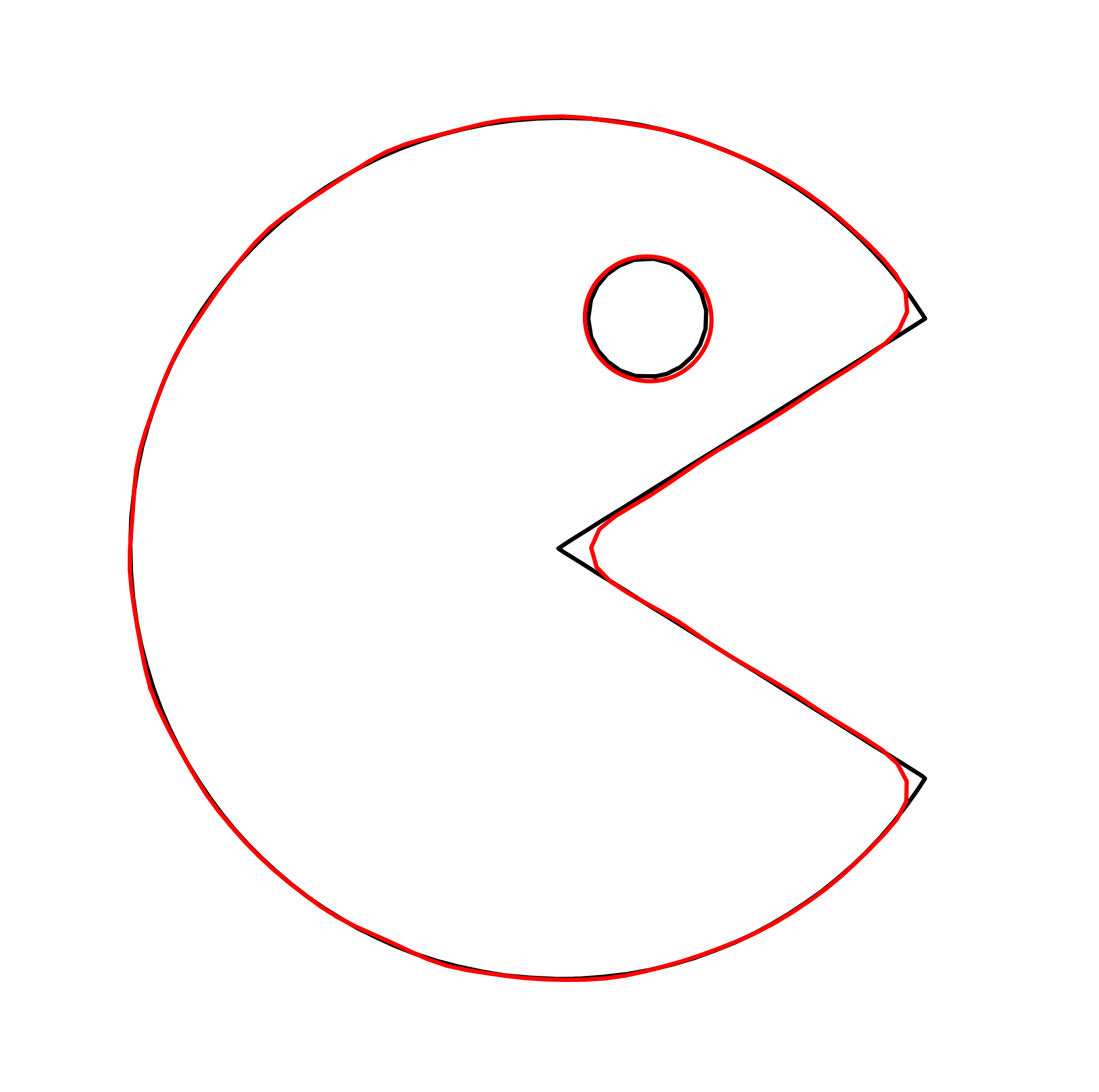}}}
	\caption{From left to right: unknown function, observations, outputs of the fixed grid method using the isotropic and Condat's total variation, output of Algorithm \ref{sfw_sliding}, gradients support (red: output of Algorithm \ref{sfw_sliding}, black: unknown)}
	\label{pacman}
\end{figure*}

The third experiment (Fig.~\ref{cross}) also showcases the rounding of corners, and highlights the influence of the regularization parameter: as $\lambda$ decreases, the curvature of the edge set increases.

\begin{figure*}
	\centering
	\includegraphics[width=.95\textwidth]{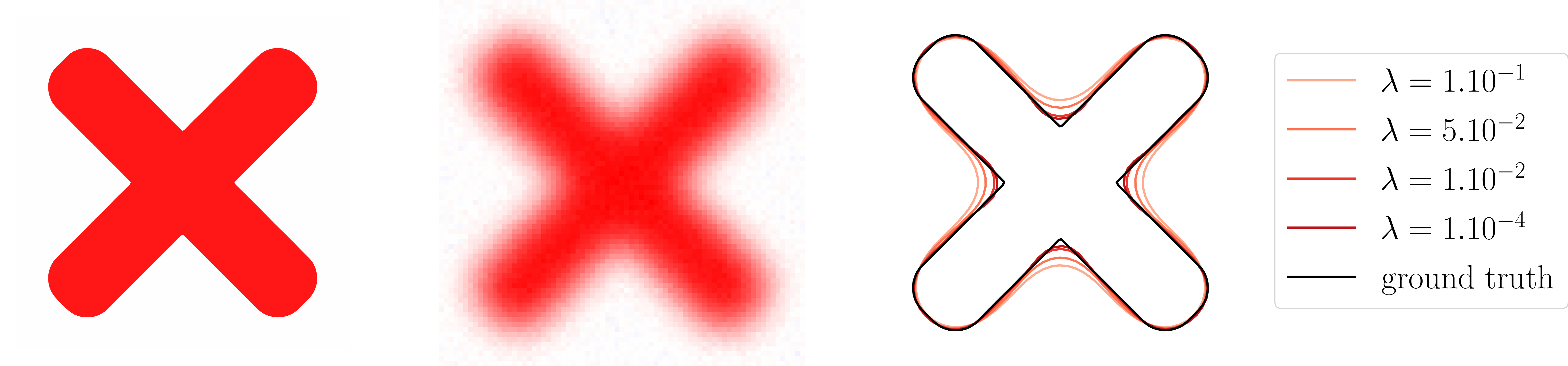}
	\caption{Left: unknown function, middle: observations, right: output of Algorithm \ref{sfw_sliding} for different values of $\lambda$}
	\label{cross}
\end{figure*}

\revision{Finally, we provide in Fig. \ref{cameraman} the results of an experiment on a more challenging task, which consists in reconstructing a natural grayscale image.}

\revisionbis{}

\revision{\paragraph{Choice of parameters.} The number of observations in the first experiment is $60\times 60$, $75\times 75$ in the second and third ones, and $64\times 64$ in the last one. In all experiments, we solved \eqref{cheeger_relaxed_discrete} on a grid of size $80\times 80$. In both local descent steps (for approximating Cheeger sets and for the sliding step), the  simple polygons have a number of vertices of order $30$ times the length of their boundary ($100$ for the last experiment), and the maximum area of triangles in their inner mesh is $10^{-2}$ (the domain being a square of side $1$). \revisionbis{The inner triangulation of a simple polygon is obtained by using Richard Shewchuk's Triangle library.} The boundary of the polygons are resampled every $30$ iterations. Line
integrals are computed using the Gauss-Patterson scheme of order $3$ ($15$ points) and triangle integrals using the Hammer-Marlowe-Stroud scheme of order $5$ ($7$ points).}

\begin{figure*}
	\centering
	 \includegraphics[width= 0.99\textwidth]{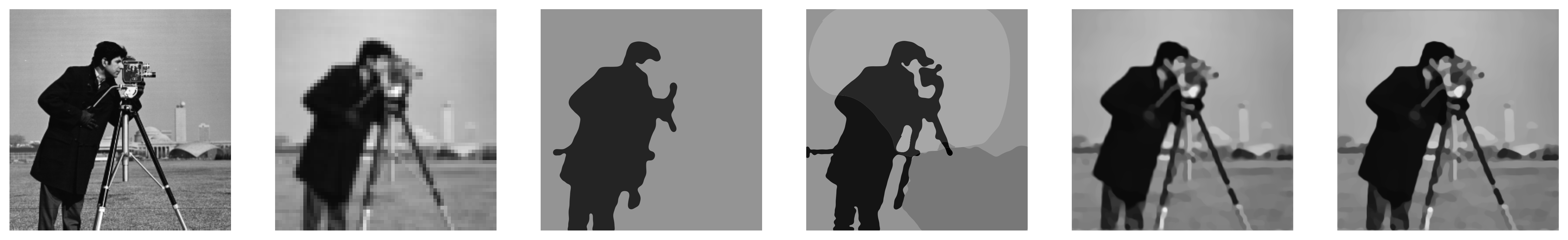}
	\caption{From left to right: original image, observations, iterates $u^{[k]}$ ($k=1,4$) produced by Algorithm \ref{sfw_sliding}, outputs of the fixed grid method using the isotropic and Condat's total variation}
	\label{cameraman}
\end{figure*}

\subsection{Topology changes during the sliding step}
\label{topo_change}

Here, we illustrate the changes of topology that may occur during the sliding step (Line 13 of Algorithm \ref{sfw_sliding}). All relevant plots are given in Figure \ref{splitting}. The unknown function (see (a)) is the sum of two indicator functions: \begin{equation*}u_0=\mathbf{1}_{B((-1, 0), 0.6)}+\mathbf{1}_{B((1, 0), 0.6)}\,,\end{equation*}
and observations are shown in (b). The Cheeger set computed at Line 5 of the first iteration covers the two disks (see (c)).

In this setting, our implementation of the sliding step converges to a function similar to (f)\footnote{This only occurs when $\lambda$ is small enough. For
higher values of $\lambda$, the output is similar to (d) or (e).}, and we obtain a valid update that decreases the objective more than the standard Frank-Wolfe update. The next iteration of the algorithm will then consist in adding a new atom to the approximation, with negative amplitude, so as to compensate for the presence of the small bottleneck.

However, it seems natural that the support of (f) should split into two disjoint simple sets, which is not possible with our current implementation. To investigate what would happen in this case, we manually split the two
sets (see (g)) and let them evolve independently. The support of the approximation converges to the union of the two disks, which produces an update that decreases the objective even more than (f).

\begin{figure*}
	\centering
	\includegraphics[width=.99\textwidth]{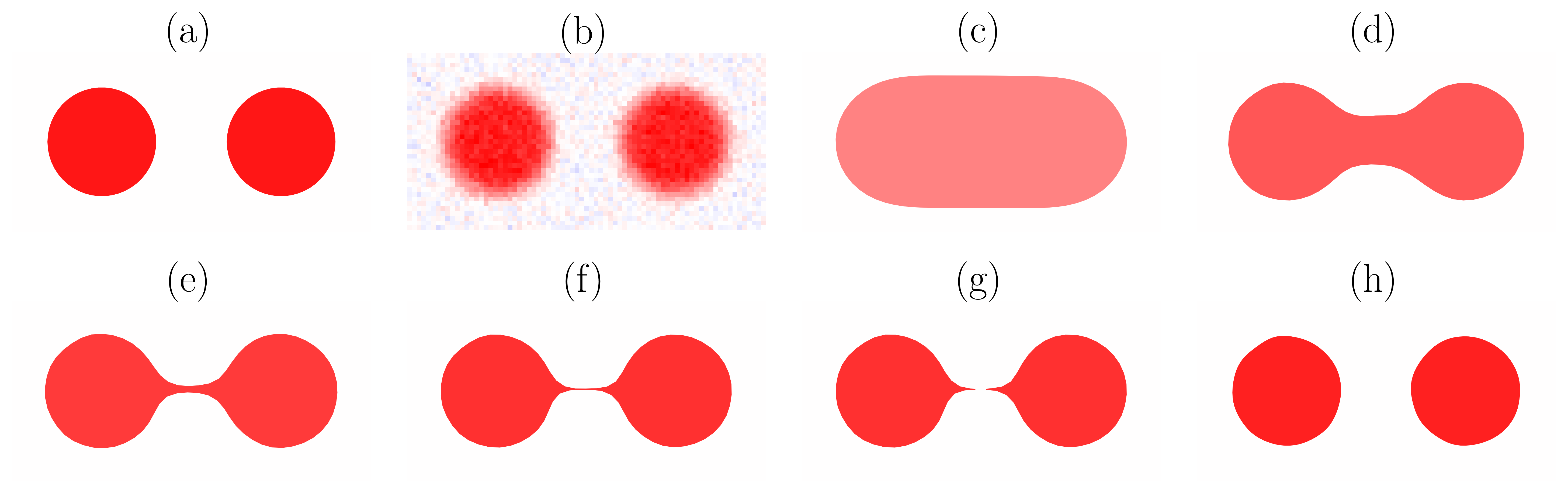}
	\caption{Topology change experiment. (a): unknown signal, (b): observations, (c): weighted Cheeger set, (d,e,f,g): sliding step iterations (with splitting), (h): final function.}
	\label{splitting}
\end{figure*}

\section{The case of a single radial measurement}
\label{toy}

In this section, we study a particular setting, where the number of
observations $m$ is equal to $1$, and the unique sensing function is
radial, i.e. the measurement operator is given by \eqref{measurement_op}
with $\varphi:\mathbb{R}^2\to\mathbb{R}$ a radial function\footnote{We
say that~${f:\RR^2\to\RR}$ is radial if there
exists~${g:[0,+\infty[\to\RR}$ such that~$f(x)=g(\|x\|)$ for almost
every~$x\in\RR^2$.}. We first state a proposition
about the solutions of \eqref{prob} in this setting, before carrying on
with
results that will require more assumptions on $\varphi$. Unless otherwise
specified, sets that differ by a Lebesgue negligible set and functions
that are equal almost everywhere are identified.

For every $u\in \LD$, we define the radialisation $\tilde{u}$ of $u$ by
\begin{equation*}
  \tilde{u}(x)=\int_{\mathbb{S}^1}u(\|x\|\,e)\,d\mathcal{H}^1(e)\,.
\end{equation*}
We note that in our setting $\Phi u$ only depends on $u$ through $\tilde{u}$, that is:
\begin{equation*}\Phi u =\int_{\RR^2}\varphi\,u=\int_{\RR^2}\tilde{\varphi}\,\tilde{u}\,.
\end{equation*}
Using the fact $|\Diff \tilde{u}|(\RR^2)\leq |\Diff u|(\RR^2)$ for any $u\in \LD$ such that $|\Diff u|(\RR^2)<+\infty$ with equality if and only if $u$ is radial (see \Cref{radial_tv} for a proof of this statement), we may state the following result:
\begin{proposition}
	Every solution of \eqref{prob} is radial, and there exists a solution that is proportional to the indicator of a disk centered at the origin.
\end{proposition}
\begin{proof}
	The first part of the result is a direct consequence of the above statements.
	Then, using \cite[Corollary 2 and Theorem 2]{boyerRepresenterTheoremsConvex2019}, we have that there exists a solution of \eqref{prob} which is proportional to the indicator function of a simple set. The result follows from the fact that every simple set whose indicator function is radial is a disk centered at the origin.
\end{proof}

We will now assume $\varphi$ is positive, continuous and decreasing\footnote{\revision{In all the following, by decreasing we mean \emph{strictly} decreasing.}} along rays. For any~${r\in\RR_+}$, we will denote by an abuse of notation~${\tilde{\varphi}(r)}$ the value of~${\tilde{\varphi}}$ at any point~${x\in\RR^2}$ such that $\|x\|=r$. We may also invoke the following assumption:

\medskip

\noindent\textbf{Assumption 1.} The function $f:r\mapsto r\,\tilde{\varphi}(r)$ is continuously differentiable on~$\mathbb{R}_+^*$, ${rf(r)\to 0}$ when~ ${r\to+\infty}$, and there exists $\rho_0>0$ such that~${f'(r)>0}$ on $]0, \rho_0[$ and $f'(r)<0$ on $]\rho_0, +\infty[$.\footnote{Assumption 1 is for example satisfied
by~${\varphi:x\mapsto \text{exp}\left(-||x||^2/(2\sigma^2)\right)}$ for any~${\sigma>0}$.}

\medskip

In the rest of this section we first explain what each step of Algorithm
\ref{sfw} should theoretically return in this particular setting, without
worrying about approximations made for implementation matters. Then, we compare
those with the output of each step of the practical algorithm.

\subsection{Theoretical behavior of the algorithm}

The first step of Algorithm \ref{sfw_sliding} consists in solving the Cheeger problem \eqref{cheeger} associated to  $\eta\overset{\text{def}}{=}\frac{1}{\lambda}\Phi^*y=\frac{y}{\lambda}\varphi$ (or equivalently to $\varphi$). \revision{To describe the solutions of this problem, we rely on Steiner symmetrization. If $E$ is a set of finite perimeter with finite measure, $\nu\in\mathbb{S}^1$ and~${z\in\RR}$, we denote
\begin{equation*}
   E_{\nu,z}\eqdef\{t\in\RR\,\rvert\,z\,\nu+t\,\nu^\perp\in E\}\,.
\end{equation*}
The Steiner symmetrization of $E$ with respect to the line through the origin and directed by $\nu$, denoted $E_{\nu}^s$, is then defined by
\begin{equation*}
   E_{\nu}^s\eqdef \{x\in\RR^2\,\rvert\,|\langle x,\nu^\perp\rangle|\leq \mathcal{L}^1(E_{\nu,\langle x,\nu\rangle})/2\}\,,
\end{equation*}
where $\mathcal{L}^1$ denotes the Lebesgue measure on $\RR$. The fundamental property of Steiner symmetrization is that it preserves volume and does not increase perimeter (see \cite[section 14.1]{maggiSetsFinitePerimeter2012} for more details).}
Using this, and denoting by $B(0,R)$ the disk of radius $R$ centered at the origin, we may state\footnote{This result can be proved using the radialisation operation previously introduced. We here however rely on classical arguments used in the analysis of geometric variational problems, which we will moreover also use later in this section.}:
\begin{proposition}\label{cheeger-radial}
All the solutions of the Cheeger problem \eqref{cheeger} associated
to~${\eta\overset{\text{def}}{=}\varphi}$ are disks centered at the origin. Under Assumption $1$ the unique solution is the disk $B(0,R^*)$ with $R^*$ the unique maximizer of \begin{equation*}R\mapsto \left[\int_{0}^R r\,\tilde{\varphi}(r)\,dr\right]/R\,.\end{equation*}
\end{proposition}
\begin{proof}
	We first stress that existence of solutions was already briefly discussed in \Cref{sec_fw} (it can either be obtained by purely geometric arguments, or by showing the existence of solutions of \eqref{cheeger_relaxed} by the direct method of calculus of variations and then using Krein-Milman theorem).

	Now if $E\subset\RR^2$ is such that $0<P(E)<+\infty$ and~${\nu\in\mathbb{S}^1}$ we have (see \Cref{steiner}):
	\begin{equation*}
		\frac{\int_{E}\eta}{P(E)}\leq \frac{\int_{E^s_{\nu}}\eta}{P(E^s_{\nu})}\,,
	\end{equation*}
	with equality if and only if $|E\triangle E^s_{\nu}|=0$. Hence if $E^*$ solves \eqref{cheeger}, arguing as in \cite[section 14.2]{maggiSetsFinitePerimeter2012}, we get that $E^*$ is a convex set which is invariant by reflection with respect to any line through the origin, and hence that $E^*$ is a ball centered at the origin.

 	Now for any $R>0$
	we have \begin{equation*}
	\mathcal{G}(R)\eqdef\frac{\int_{B(0,R)}\varphi}{P(B(0,R))}=\frac{1}{R}\int_{0}^R r\,\tilde{\varphi}(r)dr\,,
\end{equation*}
	and the last part of the result follows from a simple analysis of the variations of $\mathcal{G}$ under Assumption 1, which is given in \Cref{proofs_toy}.

\end{proof}

The second step (Line 10) of the algorithm then consists in solving
\begin{equation}
	\underset{a \in \mathbb{R}}{\text{inf}} ~~ \frac{1}{2}\left(a\int_{E^*}\varphi-y\right)^2+\lambda\,P(E^*)\,|a|\,,
	\label{amplitude_prob}
\end{equation}
where $E^*=B(0,R^*)$. The solution $a^*$ has a closed form which writes:
\begin{equation}
	a^*=\frac{\text{sign}(y)}{\int_{E^*}\varphi}\,\left(|y|-\lambda\, \frac{P(E^*)}{\int_{E^*}\varphi}\right)^+,
	\label{amplitude}
\end{equation}
where $x^+=\text{max}(x,0)$.

The next step should be the sliding one (Line 13). However, in this specific setting, one can show that the constructed function is already optimal, as stated by the following proposition:
\begin{proposition}
	Under Assumption 1, Problem \eqref{prob} has a unique solution  $a^*\,\mathbf{1}_{E^*}$ with~${E^*=B(0,R^*)}$ the solution of the Cheeger problem given by Prop. \ref{cheeger-radial}, and~$a^*$ given by \eqref{amplitude}.
\end{proposition}
\begin{proof}
If $u^*\in\LD$ solves \eqref{prob} then \begin{equation*}
\Phi^*p^*=p^*\varphi\in\partial J(u^*)\,,\end{equation*} with $p^*=-\frac{1}{\lambda}(\Phi u^*-y)$. Now from \Cref{lvlset_sousdiff} we know $p^*\varphi\in\partial J(u^*)$ implies that  $p^*\varphi\in\partial J(0)$ and that the level sets of $u^*$ satisfy
\begin{equation*}
	P(U_*^{(t)})=\left|\int_{U_*^{(t)}}p^*\varphi\right|.
\end{equation*}
This means that the non trivial level sets of $u^*$ are all solutions of the Cheeger problem associated to~${p^*\varphi}$ (or equivalently to $\varphi$), and are hence equal to $B(0,R^*)$. This shows there exists $a\in\RR$ such that~${u^*=a\,\mathbf{1}_{B(0,R^*)}}$, and the result easily follows.
\end{proof}

To summarize, with a single observation and a radial sensing function, a solution is found in a single iteration, and its support is directly identified by solving the Cheeger problem.

\subsection{Study of implementation approximations}
\label{implement_approx}
In practice, instead of solving \eqref{cheeger}, we look for \revision{an element of $\mathcal{P}_n$ (a simple polygon with at most $n$ sides)}
maximizing $\mathcal{J}$, for some given integer~${n\geq 3}$. It is hence natural to investigate the proximity of this optimal polygon  with~${B(0, R^*)}$. Solving classical geometric
variational problems restricted to the set of $n$-gons is involved, as the
Steiner symmetrization procedure might increase the number of sides \cite[Sec. 7.4]{polyaIsoperimetricInequalitiesMathematical1951a}.
However, using a trick from P\'olya and Szeg\"o, one may prove:
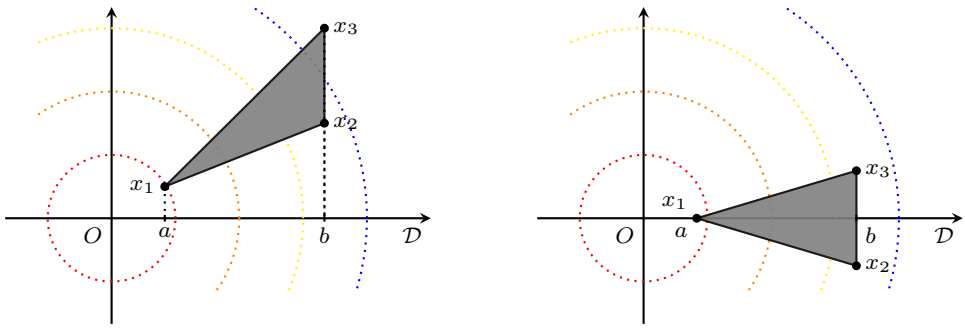
\begin{figure*}
	\centering
	\begin{tikzpicture}[scale=1.4,axis/.style={->,dashed},thick]
      \coordinate  (Ox1) at (-1,0){};
      \coordinate  (Ox2) at (3,0){};
      \coordinate  (Oy1) at (0,-1){};
      \coordinate  (Oy2) at (0,2){};

      \draw[->,>=stealth] (Ox1) -- (Ox2) node[anchor=north east]{};
      \draw[->,>=stealth] (Oy1) -- (Oy2) node[anchor=north east]{};
      \node[below left]  at (0,0) {$O$};
      \node[below left]  at (Ox2) {$\mathcal{D}$};

      \begin{scope}[color=blue]
      \clip (-0.7,-0.7) rectangle (3,2);
      \draw[dotted,color=red] (0,0) circle (0.6);
      \draw[dotted,color=orange] (0,0) circle (1.2);
      \draw[dotted,color=yellow] (0,0) circle (1.8);
      \draw[dotted,color=blue] (0,0) circle (2.4);
      \end{scope}

      \coordinate  (A) at (0.5,0.3){};
      \coordinate  (PA) at (0.5,0){};
      \coordinate  (B) at (2,0.9){};
      \coordinate  (PB) at (2,0){};
      \coordinate  (C) at (2,1.8){};

      \draw[fill=gray,opacity=0.9] (A) -- (B) -- (C) --cycle;
      \node[left] at (A) {$x_1$};
      \node[right] at (B) {$x_2$};
      \node[right] at (C) {$x_3$};
      \draw[fill=black] (A) circle (0.1em);
      \draw[fill=black] (B) circle (0.1em);
      \draw[fill=black] (C) circle (0.1em);

      \node[below] at (PA) {$a$};
      \node[below] at (PB) {$b$};

      \draw[dotted] (A) -- (PA);
      \draw[dotted] (B) -- (PB);
      \draw[dotted] (C) -- (PB);
      \draw[thick] ($(PA)+(0,0.1em)$) -- ($(PA)+(0,-0.1em)$);
      \draw[thick] ($(PB)+(0,0.1em)$) -- ($(PB)+(0,-0.1em)$);

	\begin{scope}[shift={(5,0)}]
      \coordinate  (Ox1) at (-1,0){};
      \coordinate  (Ox2) at (3,0){};
      \coordinate  (Oy1) at (0,-1){};
      \coordinate  (Oy2) at (0,2){};
      \draw[->,>=stealth] (Ox1) -- (Ox2) node[anchor=north east]{};
      \draw[->,>=stealth] (Oy1) -- (Oy2) node[anchor=north east]{};
      \node[below left]  at (0,0) {$O$};
      \node[below left]  at (Ox2) {$\mathcal{D}$};
      \begin{scope}[color=blue]
      \clip (-0.7,-0.7) rectangle (3,2);
      \draw[dotted,color=red] (0,0) circle (0.6);
      \draw[dotted,color=orange] (0,0) circle (1.2);
      \draw[dotted,color=yellow] (0,0) circle (1.8);
      \draw[dotted,color=blue] (0,0) circle (2.4);
      \end{scope}

      \coordinate  (PA) at (0.5,0){};
      \coordinate  (PB) at (2,0){};

      \coordinate (SA) at (0.5,0){};
      \coordinate (SB) at (2,-0.45){};
      \coordinate (SC) at (2,0.45){};

      \coordinate  (PA) at (0.5,0){};
      \coordinate  (PB) at (2,0){};

      \draw[fill=gray,opacity=0.9] (SA) -- (SB) -- (SC) --cycle;
      \node[above left] at (SA) {$x_1$};
      \node[right] at (SB) {$x_2$};
      \node[right] at (SC) {$x_3$};
      \draw[fill=black] (SA) circle (0.1em);
      \draw[fill=black] (SB) circle (0.1em);
      \draw[fill=black] (SC) circle (0.1em);

      \node[below left] at (PA) {$a$};
      \node[below right] at (PB) {$b$};

      \draw[thick] ($(PA)+(0,0.1em)$) -- ($(PA)+(0,-0.1em)$);
      \draw[thick] ($(PB)+(0,0.1em)$) -- ($(PB)+(0,-0.1em)$);
	\end{scope}

	\end{tikzpicture}
	\caption{Steiner symmetrization of triangles}
	\label{fig:steiner_triangles}
\end{figure*}
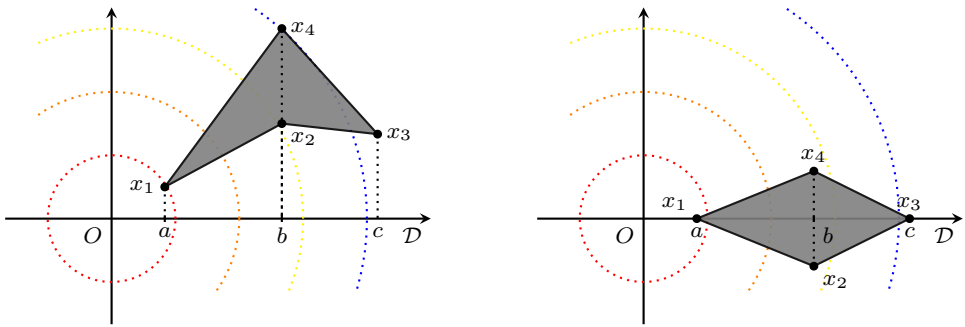
\begin{figure*}
	\centering
	\begin{tikzpicture}[scale=1.4,axis/.style={->,dashed},thick]
      \coordinate  (Ox1) at (-1,0){};
      \coordinate  (Ox2) at (3,0){};
      \coordinate  (Oy1) at (0,-1){};
      \coordinate  (Oy2) at (0,2){};

      \draw[->,>=stealth] (Ox1) -- (Ox2) node[anchor=north east]{};
      \draw[->,>=stealth] (Oy1) -- (Oy2) node[anchor=north east]{};
      \node[below left]  at (0,0) {$O$};
      \node[below left]  at (Ox2) {$\mathcal{D}$};

      \begin{scope}[color=blue]
      \clip (-0.7,-0.7) rectangle (3,2);
      \draw[dotted,color=red] (0,0) circle (0.6);
      \draw[dotted,color=orange] (0,0) circle (1.2);
      \draw[dotted,color=yellow] (0,0) circle (1.8);
      \draw[dotted,color=blue] (0,0) circle (2.4);
      \end{scope}

      \coordinate  (A) at (0.5,0.3){};
      \coordinate  (PA) at (0.5,0){};
      \coordinate  (B) at (1.6,0.9){};
      \coordinate  (PB) at (1.6,0){};
      \coordinate  (C) at (2.5,0.8){};
      \coordinate  (PC) at (2.5,0){};
      \coordinate  (D) at (1.6,1.8){};

      \draw[fill=gray,opacity=0.9] (A) -- (B) -- (C) -- (D) --cycle;
      \node[left] at (A) {$x_1$};
      \node[below right] at (B) {$x_2$};
      \node[right] at (C) {$x_3$};
      \node[right] at (D) {$x_4$};
      \draw[fill=black] (A) circle (0.1em);
      \draw[fill=black] (B) circle (0.1em);
      \draw[fill=black] (C) circle (0.1em);
      \draw[fill=black] (D) circle (0.1em);

      \node[below] at (PA) {$a$};
      \node[below] at (PB) {$b$};
      \node[below] at (PC) {$c$};

      \draw[dotted] (A) -- (PA);
      \draw[dotted] (B) -- (PB);
      \draw[dotted] (D) -- (PB);
      \draw[dotted] (C) -- (PC);

      \draw[thick] ($(PA)+(0,0.1em)$) -- ($(PA)+(0,-0.1em)$);
      \draw[thick] ($(PB)+(0,0.1em)$) -- ($(PB)+(0,-0.1em)$);

	\begin{scope}[shift={(5,0)}]
      \coordinate  (Ox1) at (-1,0){};
      \coordinate  (Ox2) at (3,0){};
      \coordinate  (Oy1) at (0,-1){};
      \coordinate  (Oy2) at (0,2){};

      \draw[->,>=stealth] (Ox1) -- (Ox2) node[anchor=north east]{};
      \draw[->,>=stealth] (Oy1) -- (Oy2) node[anchor=north east]{};
      \node[below left]  at (0,0) {$O$};
      \node[below left]  at (Ox2) {$\mathcal{D}$};

      \begin{scope}[color=blue]
      \clip (-0.7,-0.7) rectangle (3,2);
      \draw[dotted,color=red] (0,0) circle (0.6);
      \draw[dotted,color=orange] (0,0) circle (1.2);
      \draw[dotted,color=yellow] (0,0) circle (1.8);
      \draw[dotted,color=blue] (0,0) circle (2.4);
      \end{scope}

      \coordinate  (PA) at (0.5,0){};
      \coordinate  (PB) at (1.6,0){};
      \coordinate  (PC) at (2.5,0){};
      \coordinate (SA) at (0.5,0){};
      \coordinate (SB) at (1.6,-0.45){};
      \coordinate (SC) at (2.5,0){};
      \coordinate (SD) at (1.6,0.45){};

      \draw[fill=gray,opacity=0.9] (SA) -- (SB) -- (SC) -- (SD) --cycle;
      \node[above left] at (SA) {$x_1$};
      \node[below right] at (SB) {$x_2$};
      \node[above] at (SC) {$x_3$};
      \node[above] at (SD) {$x_4$};
      \draw[fill=black] (SA) circle (0.1em);
      \draw[fill=black] (SB) circle (0.1em);
      \draw[fill=black] (SC) circle (0.1em);
      \draw[fill=black] (SD) circle (0.1em);

      \node[below] at (PA) {$a$};
      \node[below right] at (PB) {$b$};
      \node[below] at (PC) {$c$};

      \draw[dotted] (SA) -- (PA);
      \draw[dotted] (SB) -- (PB);
      \draw[dotted] (SD) -- (PB);
      \draw[dotted] (SC) -- (PC);

      \draw[thick] ($(PA)+(0,0.1em)$) -- ($(PA)+(0,-0.1em)$);
      \draw[thick] ($(PB)+(0,0.1em)$) -- ($(PB)+(0,-0.1em)$);

	\end{scope}

	\end{tikzpicture}
	\caption{Steiner symmetrization of quadrilaterals}
	\label{fig:steiner_quadrilaterals}
\end{figure*}
\begin{proposition}
	Let $n\in\{3,4\}$. Then all the maximizers of $\mathcal{J}$ \revision{over $\mathcal{P}_n$} are regular and inscribed in a circle centered at the origin.
	\label{cheeger_polygon}
\end{proposition}
\begin{proof}~\medskip

	\noindent\textbf{Triangles:} let $E^*$ be a maximizer of $\mathcal{J}$ among triangles. Then the Steiner symmetrization of $E^*$ with respect to any of its heights through the origin (see \Cref{fig:steiner_triangles}) is still a triangle, and \Cref{steiner} ensures it has a higher energy except if this operation leaves $E^*$ unchanged. As a consequence, $E^*$ must be symmetric with respect to all its heights through the origin. This shows $E^*$ is equilateral and inscribed in a circle centered at the origin.

\medskip
\noindent\textbf{Quadrilaterals:} we notice that if $E$ is a simple quadrilateral, then its Steiner symmetrization with respect to any line perpendicular to one of its diagonals (see \Cref{fig:steiner_quadrilaterals}) is still a simple quadrilateral. We can then proceed exactly as for triangles to prove any maximizer~$E^*$ of~$\mathcal{J}$ \revision{over $\mathcal{P}_4$} is symmetric with respect to every line through the origin and perpendicular to one of its diagonals. This shows $E^*$ is a rhombus centered at the origin. We can now symmetrize with respect to any line through the origin perpendicular to one of its sides to finally obtain that $E^*$ must be a square centered at the origin.
\end{proof}

Our proof does not extend to $n\geq 5$, but  the following conjecture is natural:
\begin{conjecture}\label{conj-polygon}
	The result stated in \Cref{cheeger_polygon} holds for all~${n\geq 3}$.
\end{conjecture}
For $n \in \{3,4\}$ or if Conjecture \ref{conj-polygon} holds, it remains to  compare the optimal polygons with $B(0,R^*)$. If we define~${\mathcal{G}(R)\eqdef \mathcal{J}(B(0,R))}$ and
$\mathcal{G}_n(R)$ the value of $\mathcal{J}$ at any regular $n$-gon inscribed in a circle of radius $R$ centered at the origin, then we can state the following result (its proof is given in \Cref{proof_cheeger_polygon_radius}):
\begin{proposition}
	\label{cheeger_polygon_radius}
	Under Assumption 1, we have that \begin{equation*}||\mathcal{G}_n-\mathcal{G}||_{\infty}=O\left(\frac{1}{n^2}\right).\end{equation*}
	Moreover, if $f$ is of class~$C^2$ and $f''(\rho_0)<0$ , then for $n$ large enough $\mathcal{G}_n$ has a unique maximizer $R^*_n$ and \begin{equation*}
	|R^*_n-R^*|=O\left(\frac{1}{n}\right).
\end{equation*}
	If $\varphi$ is the function defined by
	\begin{equation*}
	\varphi:x\mapsto\text{exp}\left(-||x||^2/(2\sigma^2)\right),
 	\end{equation*}
	then this last result holds for all $n\geq 3$.
\end{proposition}

Now, the output of our method for approximating Cheeger sets, described in \Cref{implementation}, is a polygon that is obtained by locally maximizing $\mathcal{J}$ using a first order method. Even if we carefully
initialize this first order method, the possible existence of non-optimal
critical points makes
its analysis challenging. However, in our setting (a
radial weight function), the simple polygons that are critical points\footnote{We recall that critical point is here to be understood
in the sense that the limit appearing in \eqref{steepest_descent} is equal to
zero for every $\theta$.} of $\mathcal{J}$ coincide with its
global maximizers \revision{over $\mathcal{P}_n$} (at least for small $n$). The proof of this result is given in \Cref{critical_poly_cheeger_proof}.
\begin{proposition}
   \begin{sloppypar}
	Let $n\in\{3,4\}$. Under Assumption 1, if $f$ is of class $C^2$ and~$f''(\rho_0)<0$, the \revision{elements of $\mathcal{P}_n$} that are critical points of $\mathcal{J}$ are the regular $n$-gons inscribed in the circle of radius $R^*_n$ centered at the origin.
   \label{critical_poly_cheeger}
   \end{sloppypar}
\end{proposition}
We make the following conjecture:
\begin{conjecture}
   The result stated in \Cref{critical_poly_cheeger} holds for all $n\geq 3$.
	 \label{critical_poly_cheeger_conj}
\end{conjecture}
If $n\in\{3,4\}$, or if Conjecture \ref{critical_poly_cheeger_conj} holds, we may therefore expect our polygonal approximation to be at Hausdorff distance of order $O\left(\frac{1}{n}\right)$ to $B(0,R^*)$.

\section{Conclusion}

As shown in the present exploratory work, solving total variation regularized inverse problems in a gridless manner is highly beneficial, as it allows to preserve structural properties of their solutions, which cannot be achieved by traditional numerical solvers. The price to pay for going ``off-the-grid'' is an increased complexity of the analysis and the implementation of the algorithms. Furthering their theoretical study and improving their practical efficiency and reliability is an interesting avenue for future research. \revision{Investigating extensions to higher dimensions (e.g. 3D) could also be promising. Although the computational cost of each step might be large, it seems that the proposed algorithm could be transposed to this new setting.}

\begin{acknowledgements}
The authors thank Robert Tovey for carefuly reviewing the code used in the numerical experiments section, and for suggesting several modifications that significantly improved the results presented therein.
\end{acknowledgements}

\bibliography{ref}
\bibliographystyle{apalike}

\appendix

\section{Derivation of Algorithm \ref{sfw}}
\label{derivation_algo}

See \cite[Sec. 4.1]{denoyelleSlidingFrankWolfeAlgorithm2019} for the case of the sparse spikes problem.

\begin{lemma} Problem
  \eqref{prob} is equivalent to
  \begin{equation}
     \begin{aligned}
      & \underset{(u,t)\in C}{\text{min}} &&
      \tilde{T}_{\lambda}(u,t)\overset{def}{=}\frac{1}{2}||\Phi u-y||^2 +\lambda t \\
     \end{aligned}
     \label{prob_epi}
     \tag{$\tilde{\mathcal{P}}_{\lambda}$}
  \end{equation}
  with \begin{equation*}C\eqdef\left\{(u,t)\in
  \LD\times\mathbb{R}\,\big\rvert\, |\Diff u|(\RR^2)\leq t\leq
  M \right\}.\end{equation*} and $M\eqdef||y||^2/(2\lambda)$,
  i.e. if $u$ is a solution of \eqref{prob} then we have that ${(u,|\Diff u|(\mathbb{R}^2))}$ is a solution of
  \eqref{prob_epi}, and conversely any solution of \eqref{prob_epi} is of the form $(u,|\Diff u|(\mathbb{R}^2))$ with $u$ a solution of \eqref{prob}.
\end{lemma}
\begin{proof}
  If $u^*$ is a solution of \eqref{prob},
  then \begin{equation*}T_{\lambda}(u^*)\leq T_{\lambda}(0)=||y||^2/2\,.\end{equation*}
  Hence we have that $|\Diff u^*|(\RR^2)\leq
  M$, which shows the feasible set of
  \eqref{prob} can be restricted to functions $u$ which are such that~${|\Diff u|(\RR^2)\leq
  M}$. It is then straightforward to show
  that the resulting program is equivalent to \eqref{prob_epi}, in the sense
  defined above.
\end{proof}

The objective $\widetilde{T}_{\lambda}$ of \eqref{prob_epi} is now convex,
differentiable and we have for all $(u,t)\in \LD\times\RR$
\begin{equation*}
   \begin{aligned}
   d\tilde{T}_{\lambda}(u,t)\colon \LD\times \mathbb{R} &\to \mathbb{R} \\
   (v,s) & \mapsto \left[\int_{\RR^2} \Phi^*(\Phi u-y)\,v\right]+\lambda s\,.
   \end{aligned}
\end{equation*}
Moreover, the feasible set $C$ is weakly compact. We can therefore apply Frank-Wolfe algorithm to \eqref{prob_epi}. The following result shows how the linear minimization step (Line 2 of Algorithm \ref{fw}) one has to perform at step $k$ amounts to solving the Cheeger problem \eqref{cheeger} associated to $\eta\eqdef-\frac{1}{\lambda}\Phi^*(\Phi u^{[k]}-y)$.

\begin{proposition}
   Let $(u,t)\in C$ and $\eta\eqdef-\frac{1}{\lambda}\Phi^*(\Phi u-y)$. We also denote
   \begin{equation}
       \alpha\eqdef \underset{E \subset{\mathbb{R}^2} }{\text{sup}} ~ \frac{\left|\int_E\eta\right|}{P(E)}~~\text{s.t.}~~ 0<|E|<+\infty,~P(E)<+\infty\,.
     \label{cheeger_bis}
   \end{equation}
   Then, if $\alpha\leq 1$, $(0,0)$ is a minimizer of $d\tilde{T}_{\lambda}(u,t)$ over $C$. Otherwise, there exists a simple set $E$ achieving the supremum in \eqref{cheeger_bis} such that, denoting $\epsilon=\text{sign}\left(\int_{E}\eta\right)$,
   $\left(\frac{\epsilon M}{P(E)}\mathbf{1}_E, M\right)$ is a minimizer of $d\tilde{T}_{\lambda}(u,t)$ on $C$.
\end{proposition}
\begin{proof}
The extreme points of $C$ are $(0,0)$ and the elements of
\begin{equation*}
	\left\{\left(\pm\frac{M}{P(E)}\mathbf{1}_E,M\right)~\bigg|~E \text{ is simple},~0<|E|<+\infty\right\}.
\end{equation*}
Since
$d\tilde{T}_\lambda(u,t)$ is linear, it reaches its minimum on $C$ at least at
one of these extreme points. We hence have that \begin{equation*}
(0,0)\in\underset{(v,s)\in
C}{\text{Argmin}}~d\tilde{T}_{\lambda}(u,t)(v,s)\end{equation*}
or that a minimizer can be
found by finding an element of
\begin{equation*}
\begin{aligned}
   &\underset{\substack{E\text{ simple}\\\epsilon\in\{-1,1\}}}{\text{Argmin}}~
   \left\langle\Phi u-y,
   \frac{\epsilon M}{P(E)}\Phi\mathbf{1}_E\right\rangle+\lambda M\\
=&\underset{\substack{E\text{ simple}\\\epsilon\in\{-1,1\}}}{\text{Argmin}}~
\left\langle\Phi u-y,\frac{\epsilon}{\lambda
P(E)}\Phi\mathbf{1}_E\right\rangle\\
=&\underset{\substack{E\text{ simple}\\\epsilon\in\{-1,1\}}}{\text{Argmin}}~
\frac{\epsilon}{P(E)}
\int_E \frac{1}{\lambda}\Phi^*\left(\Phi u-y\right).
\end{aligned}
\end{equation*}
This last problem is equivalent to finding an element of
\begin{equation*}
\underset{E\text{
simple}}{\text{Argmax}}~~\frac{1}{P(E)}\left|\int_E \eta\right|,
\end{equation*}
in the sense that $E^*$ is optimal for the latter if and only if the couple $\left(E^*,\text{sign}\left(\int_{E^*}\eta\right)\right)$ is optimal for the
former. We can moreover show that $(0,0)$ is optimal if and only if for all~${E\subset\RR^2}$ such that $0<|E|<+\infty$ and $P(E)<+\infty$ we have:\begin{equation*}
\frac{1}{P(E)}\left|\int_{E}\eta\right|\leq 1\,.
\end{equation*}
\end{proof}

\section{Discussion on Line 10 of Algorithms \ref{sfw} and \ref{sfw_sliding}}
\label{lasso_amplitudes}
\revision{Consdering \Cref{derivation_algo} and Algorithm \ref{fw}, the standard Frank-Wolfe update at iteration $k$ would be to take $u^{[k+1]}$ equal to~${\tilde{u}^{[k+1]}}$ with:
\begin{equation*}
   \tilde{u}^{[k+1]}\eqdef(1-\gamma_k)\, u^{[k]}+\gamma_k\frac{M\,\epsilon_*}{P(E_*)}\mathbf{1}_{E_*}\,,
\end{equation*}
where ${\gamma_k=\frac{2}{k+2}}$, $E_*$ is the set obtained at Line 5  and $\epsilon_*$ is the sign of~$\int_{E_*}\eta^{[k]}$. Now, one can write $\tilde{u}^{[k+1]}$ as a linear combination of indicator functions of its level sets, and then apply the decomposition
mentionned in \Cref{preli_sets} to each level set. This allows to
find a family $(E_i)_{i=1}^N$ of simple sets of positive measure and $(a_i)_{i=1}^N\in\RR^N$ such that~${\tilde{u}^{[k+1]}=\sum_{i=1}^N a_i\,\mathbf{1}_{E_i}}$ and
\begin{equation*}
\left|\Diff\left(\sum\limits_{i=1}^N a_i\,\mathbf{1}_{E_i}\right)\right|(\mathbb{R}^2)=\sum\limits_{i=1}^N |a_i|\,P(E_i)\,.
\end{equation*}
Moreover, it is possible to prove (see \cite{hdr}) that for every~$i\neq j$
\begin{enumerate}
   \item Either $E_i\subset E_j$, $E_j\subset E_i$ or $E_i\cap E_j=\emptyset$.
   \item If $\mathrm{sign}(a_i)=\mathrm{sign}(a_j)$ and $E_i\cap E_j=\emptyset$ then it holds that~${\mathcal{H}^1(\partial^* E_i\cap \partial^* E_j)=0}$.
   \item If $\mathrm{sign}(a_i)=-\mathrm{sign}(a_j)$ and $E_i\subset E_j$ then it holds again that~${\mathcal{H}^1(\partial^* E_i\cap \partial^* E_j)=0}$.
\end{enumerate}
We hence deduce that for every $b\in\RR^N$ such that\begin{equation*}\forall i\in\{1,...,N\},~\mathrm{sign}(a_i)=\mathrm{sign}(b_i)\,,\end{equation*}we have:
\begin{equation}
   \left|\Diff\left(\sum\limits_{i=1}^N b_i\,\mathbf{1}_{E_i}\right)\right|(\mathbb{R}^2)=\sum\limits_{i=1}^N |b_i|\,P(E_i)\,.
\end{equation}
This shows that if $E^{[k+1]}=(E_1,...,E_N)$ and
\begin{equation}
  \begin{aligned}
    a^{[k+1]}\in~& \underset{b \in \RR^N}{\text{Argmin}} && \frac{1}{2}||\Phi_E\,b -y||^2+\lambda\,\sum\limits_{i=1}^N P(E_i)\,|b_i| \\
    & \text{s.t.} && \forall i\in\{1,...,N\},~\mathrm{sign}(b_i)=\mathrm{sign}(a_i)\,,
  \end{aligned}
\end{equation}
then, defining $u^{[k+1]}=\sum_{i=1}^N a_i^{[k+1]}\,\mathbf{1}_{E_i^{[k+1]}}$, we finally obtain~${T_{\lambda}(u^{[k+1]})\leq T_{\lambda}(\tilde{u}^{[k+1]})}$, which ensures the validity of this update.

As a final note, let us mention that applying the decomposition mentionned in \Cref{preli_sets} to the level sets of $\tilde{u}^{[k+1]}$ is a computationally challenging task. However, we stress again that, generically,  $\mathcal{H}^1(\partial^*E_i\cap\partial^*E_j)=0$ for every $i\neq j$, so that the above procedure is never required in practice.}

\section{Existence of maximizers of the Cheeger ratio among simple polygons with at most $n$ sides}
\revision{
\label{ex_cheeger_polygon}
Let $\eta\in\LD\cap C^0(\RR^2)$ and $n\geq 3$. We want to prove the existence of maximizers of the Cheeger ratio $\mathcal{J}$ associated to~$\eta$ among simple polygons with at most $n$ sides. We will in fact prove a slightly stronger result, namely the existence of maximizers of a relaxed energy which coincides with $\mathcal{J}$ on simple polygons, and the existence of a simple polygon maximizing this relaxed energy.

We first begin by defining relaxed versions of the perimeter and the (weighted) area. To be able to deal with polygons with a number of vertices smaller than $n$, which will be useful in the following, we define for all $m\geq 2$ and ${x\in\RR^{m\times 2}}$ the following quantities:
\begin{equation*}
      \mathcal{P}(x)\eqdef\sum\limits_{i=1}^m \|x_{i+1}-x_i\| ~~\text{and}~~
      \mathcal{A}(x)\eqdef\int_{\RR^2}\eta\,\chi_x\,,
\end{equation*}
where $\chi_x(y)$ denotes the index (or winding number) of any parametrization of the polygonal curve~${[x_1,x_2],...,[x_m,x_1]}$ around $y\in\RR^2$. In particular, for every $x\in\mathcal{X}_m$ (i.e. for every~${x\in\RR^{m\times 2}}$ defining a simple polygon), we have
\begin{equation*}
      \mathcal{P}(x)=P(E_x) ~~\text{and}~~
      \left|\mathcal{A}(x)\right|=\left|\int_{E_x}\eta\right|,
\end{equation*}
and hence, as soon as $\mathcal{P}(x)>0$:
\begin{equation*}
\mathcal{J}(E_x)=\frac{|\mathcal{A}(x)|}{\mathcal{P}(x)}\,.
\end{equation*}
This naturally leads us to define
\begin{equation*}
   \mathcal{Y}_m\eqdef\left\{x\in\RR^{m\times 2}\,\big\rvert\,\mathcal{P}(x)>0\right\}
\end{equation*}
and to denote, abusing notation, $\mathcal{J}(x)=\left|\mathcal{A}(x)\right|/\mathcal{P}(x)$ for every~${x\in\mathcal{Y}_m}$.

The function $\chi_x$ is constant on each connected component of $\RR^2\setminus \Gamma_x$ with $\Gamma_x\eqdef\cup_{i=1}^m[x_i,x_{i+1}]$. It takes values in~${\{- m,...,m\}}$ and is equal to zero on the only unbounded connected component. We also have~${\partial\,\text{supp}(\chi_x)\subset \Gamma_x}$. Moreover $\chi_x$ has bounded variation and for $\mathcal{H}^1$-almost every~${y\in\Gamma_x}$
there exists $u_{\Gamma}^+(y),u_{\Gamma}^-(y)$ in $\{-m,...,m\}$ such that \begin{equation*}
\Diff \chi_x=(u_{\Gamma_x}^+ -u_{\Gamma_x}^-)\,\nu_{\Gamma_x}\mathcal{H}^1\mres \Gamma_x\,.
\end{equation*}

Now we define
\begin{equation*}
   \alpha\eqdef\underset{x\in\mathcal{Y}_n}{\text{sup}}\mathcal{J}(x)\,.
\end{equation*}
If $\eta=0$, then the existence of maximizers is trivial. Otherwise, there exists a Lebesgue point $x_0$ of $\eta$ at which $\eta$ is non-zero. Now the family of regular $n$-gons inscribed in any circle centered at $x_0$ has bounded eccentricity. Hence, if $x_{n,r}$ defines a regular $n$-gon inscribed in a circle of radius $r$ centered at $x_0$, Lebesgue differentiation theorem ensures that
\begin{equation*}
   \underset{r\to 0^+}{\text{lim}}~\frac{\left|\int_{E_{x_{n,r}}}\eta\right|}{|E_{x_{n,r}}|}>0\,,
\end{equation*}
and the fact that $\alpha>0$ easily follows.

\begin{lemma}
\label{lemma_bounds_polygon}
   Let $C>0$. There exists $R>0$ and $c>0$ such that
   \begin{equation*}
      \forall x\in\mathcal{Y}_n,~
      \mathcal{J}(x)\geq C\implies \mathcal{P}(x)\geq c \text{ and } \|x_i\|\leq R \text{ for all } i.
   \end{equation*}
\end{lemma}
\begin{proof}
   The proof is similar to that of \Cref{lemma_bounds}.

   \medskip\noindent\textbf{Upper bound on the perimeter:}
   the integrability of $\eta^2$ yields that for every $\epsilon>0$ there exists $R_1>0$ such that \begin{equation}
   \int_{\RR^2\setminus B(0,R_1)}\eta^2\leq \epsilon^2\,.\label{cond_integ}\end{equation}
   Let $\epsilon>0$ and $R_1>0$ such that \eqref{cond_integ} holds. We have
   \begin{equation*}
      \begin{aligned}
         \mathcal{P}(x)&\leq \frac{1}{C}\left|\mathcal{A}(x)\right|\\
         &\leq \frac{1}{C}\left[\left|\int_{\RR^2\cap B(0,R)}\eta\,\chi_x\right|+\left|\int_{\RR^2\setminus B(0,R)}\eta\,\chi_x\right|\right]\\
         &\leq \frac{1}{C}\left[\|\eta\|_{\mathrm{L}^2}\,\|\chi_x\|_{\mathrm{L}^\infty}\,\sqrt{|B(0,R)|}+\epsilon\,\|\chi_x\|_{\mathrm{L}^2}\right]\\
         &\leq \frac{1}{C}\left[\|\eta\|_{\mathrm{L}^2}\,n\,\sqrt{|B(0,R)|}+\epsilon\,\frac{1}{\sqrt{c_2}}|\Diff\chi_x|(\RR^2)\right]\\
         &\leq \frac{1}{C}\left[\|\eta\|_{\mathrm{L}^2}\,n\,\sqrt{|B(0,R)|}+\epsilon\,\frac{2n}{\sqrt{c_2}} \mathcal{P}(x)\right].
      \end{aligned}
   \end{equation*}
   Now,
   taking
   \begin{equation*}
      \epsilon\eqdef \frac{C\,\sqrt{c_2}}{4n} ~~\text{and}~~ c'=\frac{2n}{C}\|\eta\|_{\mathrm{L}^2}\sqrt{|B(0,R)|}\,,
   \end{equation*} we finally get that~${\mathcal{P}(x)\leq c'}$.

   \medskip\noindent\textbf{Inclusion in a ball:} we take $\epsilon=\frac{\sqrt{c_2}}{4n}$ and fix~${R_2>0}$ such that $\int_{\mathbb{R}^2\setminus B(0,R_2)}\eta^2\leq \epsilon^2$.
   Let us show that
   \begin{equation*}
      \text{supp}(\chi_x)\cap B(0,R_2)\neq \emptyset\,.
   \end{equation*}
   By contradiction, if $\text{supp}(\chi_x)\cap B(0,R_2)= \emptyset$, we would have:
   \begin{equation*}
      \begin{aligned}
      \mathcal{P}(x)&\leq \frac{1}{C}\left|\mathcal{A}(x)\right|\\
      &=\frac{1}{C}
       \left|\int_{\RR^2\setminus B(0,R_2)}\eta\,\chi_x\right| \\
       &\leq \sqrt{\int_{\mathbb{R}^2\setminus B(0,R_2)}\eta^2}~\|\chi_x\|_{\mathrm{L}^2}\\
      &\leq\frac{\epsilon}{\sqrt{c_2}}|D\chi_x|(\RR^2)\leq \frac{2n\,\epsilon}{\sqrt{c_2}}\,\mathcal{P}(x)\,.
   \end{aligned}
   \end{equation*}
   Dividing by $\mathcal{P}(x)>0$ yields a contradiction. Now since \begin{equation*}
   \partial\, \text{supp}(\chi_x)\subset\Gamma_x\,,
   \end{equation*}
   we have $\text{diam}(\text{supp}(\chi_x))\leq\mathcal{P}(x)\leq c'$ which shows
   \begin{equation*}
   \text{supp}(\chi_x)\subset B(0,R) \text{ with } R\eqdef c'+R_2\,.
   \end{equation*}
   This in turn implies that $\|x_i\|\leq R$ for all $i$.

   \medskip\noindent\textbf{Lower bound on the perimeter:} the integrability of~$\eta^2$ shows that, for every $\epsilon>0$, there exists $\delta>0$ such that
   \begin{equation*}
      \forall E\subset\RR^2,~|E|\leq \delta\implies \left|\int_{E}\eta^2\right|\leq \epsilon^2\,.
   \end{equation*}
   Taking $\epsilon\eqdef C\,\sqrt{c_2}/2$, we obtain that if $|\text{supp}(\chi_x)|\leq\delta$
   \begin{equation*}
      \begin{aligned}
      \mathcal{P}(x)&\leq \frac{1}{C}\left|\mathcal{A}(x)\right|=\frac{1}{C}\left|\int_{\text{supp}(\chi_x)}\eta\right|\\
      &\leq\frac{1}{C}\sqrt{\int_{\text{supp}(
      \chi_x)}\eta^2}~\sqrt{|\text{supp}(\chi_x)|}\\
      &\leq \frac{\epsilon}{C\sqrt{c_2}}\,P(\text{supp}(\chi_x))\\
      &\leq \frac{\epsilon}{C\sqrt{c_2}}\,\mathcal{P}(x)\,,
      \end{aligned}
   \end{equation*}
   the last inequality holding because  $\partial\,\text{supp}(\chi_x)\subset\Gamma_x$. We get a contradiction since $\mathcal{P}(x)$ is positive.
\end{proof}
Applying \Cref{lemma_bounds_polygon} with e.g. $C=\alpha/2$, and defining
\begin{equation*}
   \mathcal{Y}'_n\eqdef\left\{x\in\RR^{n\times 2}\,\big\rvert\,\mathcal{P}(x)\geq c \text{ and } \|x_i\|\leq R\text{ for all }i\right\},
\end{equation*}
 we see that any maximizer of $\mathcal{J}$ over $\mathcal{Y}'_n$ (if it exists) is also a maximizer of $\mathcal{J}$ over $\mathcal{Y}_n$, and conversely.

\begin{lemma}
   Let $x\in\RR^{n\times 2}$. Then for every $a\in\RR^2$ we have
   \begin{equation*}
      \begin{aligned}
      &\mathcal{A}(x)=\sum\limits_{i=1}^n\mathrm{sign}(\mathrm{det}(x_i-a~x_{i+1}-a))\int_{a x_i x_{i+1}}\eta\\
      &=\sum\limits_{i=1}^n \mathrm{det}(x_i-a~x_{i+1}-a)\int_{T_1}\eta((x_i-a~x_{i+1}-a)\,y)\,dy\,,
   \end{aligned}
   \end{equation*}
   where $a x_i x_{i+1}$ denotes the triangle with vertices  $a,x_i,x_{i+1}$ and $T_1\eqdef\left\{(\alpha,\beta)\in\left(\RR_+\right)^2\,\big\rvert\,\alpha+\beta\leq 1\right\}$ is the unit triangle.
   \label{lemma_signed_area}
\end{lemma}
\begin{proof}
   Let us show that for all~${a\in\RR^2}$ we have
   \begin{equation}
      \chi_x=\sum\limits_{i=1}^n\mathrm{sign}(\mathrm{det}(x_i-a~x_{i+1}-a))\,\mathbf{1}_{a x_i x_{i+1}}
      \label{winding_eq}
   \end{equation}
   almost everywhere. First, we have that $y\in\RR^2$ is in the (open) triangle $ax_ix_{i+1}$ if and only if the ray issued from $y$ directed by $y-a$ intersects $]x_i,x_{i+1}[$. Moreover, if $y$ is in this triangle, then
   \begin{equation*}
      \text{sign}(\text{det}(x_i-a~x_{i+1}-a))=\text{sign}\left((y-a)\cdot (x_{i+1}-x_i)^{\perp}\right).
   \end{equation*}
   The above hence shows that, if $y\in \RR^2\setminus\cup_{i=1}^n[x_i,x_{i+1}]$ does not belong to any of the segments $[a,x_i]$, evaluating the right hand side of \eqref{winding_eq} at $y$ amounts to computing the winding number  $\chi_x(y)$ by applying the ray-crossing algorithm described in  \cite{hormannPointPolygonProblem2001}. This in particular means that
   \eqref{winding_eq} holds almost everywhere, and the result follows.
\end{proof}

 From \Cref{lemma_signed_area}, we deduce that $\mathcal{A}$ is continuous on $\RR^{n\times 2}$. This is also the case of $\mathcal{P}$. Now $\mathcal{Y}'_n$ is compact and included in $\mathcal{Y}_n$, hence the existence of maximizers of $\mathcal{J}$ over~$\mathcal{Y}'_n$, which in turn implies the existence of maximizers of $\mathcal{J}$ over $\mathcal{Y}_n$.

Let us now show there exists a maximizer which belongs to $\mathcal{X}_n$. To do so, we rely on the following lemma
\begin{lemma}
   Let $m\geq 3$ and $x\in\mathcal{Y}_m\setminus \mathcal{X}_m$. Then there exists~$m'$ with $2\leq m'<m$ and $y\in\mathcal{Y}_{m'}$ such that
   \begin{equation*}
      \mathcal{J}(x)\leq\mathcal{J}(y)\,.
   \end{equation*}
\end{lemma}
\begin{proof}
   If $x\in\mathcal{Y}_m\setminus \mathcal{X}_m$ then $[x_1,x_2],...,[x_m,x_1]$ is not simple. If there exists $i$ with $x_i=x_{i+1}$ then \begin{equation*}y=(x_1,...,x_i,x_{i+2},...,x_m)\end{equation*}is suitable, and likewise if $x_1=x_m$ then \begin{equation*}y=(x_1,...,x_{m-1})\end{equation*} is suitable. Otherwise we distinguish the following cases:

   \medskip\noindent\textbf{If there exists $i<j$ with $x_i=x_j$:} we define
   \begin{equation*}
      \begin{aligned}
      y&=(x_1,...,x_i,x_{j+1},...,x_m)\in\RR^{m-(j-i)}\,,\\
      z&=(x_i,x_{i+1},...,x_{j-1})\in\RR^{j-i}\,.
      \end{aligned}
   \end{equation*}
   We notice that $2\leq j-i<m$ and $2\leq m-(j-i)<m$.

   \medskip\noindent\textbf{If there exists $i<j$ with $x_i\in]x_j,x_{j+1}[$:} we necessarily have $(i,j)\neq (1,m)$. We define
   \begin{equation*}
      \begin{aligned}
      y&=(x_1,...,x_i,x_{j+1},...,x_m)\in\RR^{m-(j-i)}\,,\\
      z&=(x_i,x_{i+1},...,x_{j})\in\RR^{j-i+1}\,.
      \end{aligned}
   \end{equation*}
   We again have $2\leq m-(j-i)<m$, and since $(i,j)\neq (1,m)$, we have $j-i<m-1$ which shows $2\leq j-i+1<m$.

   \medskip\noindent\textbf{If there exists $i<j$ with $x_j\in]x_i,x_{i+1}[$:} we necessarily have $j>i+1$. We define
   \begin{equation*}
      \begin{aligned}
      y&=(x_1,...,x_i,x_j,...,x_m)\in\RR^{m-(j-i)+1}\,,\\
      z&=(x_{i+1},...,x_{j})\in\RR^{j-i}\,.
      \end{aligned}
   \end{equation*}
   We again have $2\leq j-i<m$, and since $j>i+1$ we obtain that~$2\leq {m-(j-i)+1<m}$.

   \medskip\noindent\textbf{If there exists $i<j$ with $x'\in]x_i,x_{i+1}[\,\cap\,]x_j,x_{j+1}[$:} if we have $j=i+1$ then either $x_{i+2}\in]x_i,x_{i+1}[$ or $x_i\in]x_{i+1},x_{i+2}[$ and in both cases we fall back on the previously treated cases. The same holds if $(i,j)=(1,m)$. Otherwise, we define
   \begin{equation*}
      \begin{aligned}
      y&=(x_1,...,x_i,x',x_{j+1},...,x_m)\in\RR^{m-(j-i)+1}\,,\\
      z&=(x',x_{i+1},...,x_{j})\in\RR^{j-i+1}\,.
      \end{aligned}
   \end{equation*}
   Since $j>i+1$ and $(i,j)\neq (1,m)$ we get $2\leq m-(j-i)+1<m$ and $2\leq j-i+1<m$.

   \medskip\noindent Now, one can see that in each case we have~${\mathcal{P}(x)=\mathcal{P}(y)+\mathcal{P}(z)}$ and  $\chi_x=\chi_y+\chi_z$ almost everywhere, which in turn gives that~${\mathcal{A}(x)=\mathcal{A}(y)+\mathcal{A}(z)}$. We hence get that $\mathcal{P}(y)=0$ or $\mathcal{P}(z)=0$, and in this case $\mathcal{J}(x)=\mathcal{J}(y)$ or
   $\mathcal{J}(x)=\mathcal{J}(z)$, or that $\mathcal{P}(y)>0$ and $\mathcal{P}(z)>0$, which yields
   \begin{equation*}
      \begin{aligned}
      \frac{\left|\mathcal{A}(x)\right|}{\mathcal{P}(x)}&\leq \frac{\left|\mathcal{A}(y)\right|+\left|\mathcal{A}(z)\right|}{\mathcal{P}(y)+\mathcal{P}(z)}\\
      &=\frac{\mathcal{P}(y)}{\mathcal{P}(y)+\mathcal{P}(z)}\frac{\left|\mathcal{A}(y)\right|}{\mathcal{P}(y)}+\frac{\mathcal{P}(z)}{\mathcal{P}(y)+\mathcal{P}(z)}\frac{\left|\mathcal{A}(z)\right|}{\mathcal{P}(z)}\,.
   \end{aligned}
   \end{equation*}
   Hence $\mathcal{J}(x)$ is smaller than a convex combination of $\mathcal{J}(y)$ and~$\mathcal{J}(z)$, which gives that it is smaller than $\mathcal{J}(y)$ or~$\mathcal{J}(z)$. This shows that $y$ or $z$ is suitable.
\end{proof}
We can now prove our final result, i.e. that there exists~${x_*\in\mathcal{X}_n}$ such that
\begin{equation*}
\forall x\in\mathcal{Y}_n,~\mathcal{J}(x_*)\geq \mathcal{J}(x)\,.
\end{equation*}
Indeed, repeatedly applying the above lemma starting with a maximizer $x_*$ of $\mathcal{J}$ over $\mathcal{Y}_n$, we either have that
there exists~$m$ with $3\leq m\leq n$ and~${x'_*\in\mathcal{X}_m}$ such that $\mathcal{J}(x_*)=\mathcal{J}(x'_*)$, or that there exists $y\in \mathcal{Y}_2$ such that $\mathcal{J}(x_*)\leq\mathcal{J}(y)$, which is impossible since in that case $\mathcal{J}(y)=0$ and
$\mathcal{J}(x_*)=\alpha>0$. We hence have $x'_*\in\mathcal{X}_m$ such that
\begin{equation*}
   \forall x\in\mathcal{Y}_n,~\mathcal{J}(x'_*)=\mathcal{J}(x_*)\geq \mathcal{J}(x)\,.
\end{equation*}
We can finally build $x''_*\in\mathcal{X}_n$ such that $\mathcal{J}(x''_*)=\mathcal{J}(x'_*)$ by adding dummy vertices to $x'_*$, which finally allows to conclude.
}

\section{Proof of Proposition \ref{cheeger_grid}}
\label{proof_cheeger_grid}

First, let us stress that for any function $v$ that is piecewise constant on $(C_{i,j})_{(i,j)\in[ 1,N]^2}$ and that is equal to $0$ outside~${[-R,R]^2}$, we have $J(v)=h\,\|\nabla^h v\|_{1,1}$ where by abuse of notation $\nabla^h v$ is given by~\eqref{discrete_grad} with $v_{i,j}$ the value of $v$ in~$C_{i,j}$. Hence $J^h(u^h) \leq 1$ for all
$h$ implies that $J(u^h)$ (and hence
$\|u^h\|_{\mathrm{L}^2}$) is
uniformly bounded in~$h$. There hence exists a (not relabeled) subsequence that converges strongly in~$\mathrm{L}^1_{loc}(\RR^2)$ and weakly in~$\LD$ to a function $u$, with moreover $\Diff u^h\overset{\ast}{\rightharpoonup}\Diff u$.

Let us now take $\phi=(\phi^{(1)},\phi^{(2)})\in C^{\infty}_c(\RR^2,\RR^2)$  such that~${||\phi||_{\infty}\leq 1}$. The weak-* convergence of the gradients give us that
\begin{equation*}
\begin{aligned}
&\int_{\RR^2}\phi\cdot dDu
=\underset{h\to 0}{\text{lim}}~\int_{\RR^2}\phi\cdot dDu^h\\
=&\underset{h\to 0}{\text{lim}}~\sum\limits_{i=0}^N\sum\limits_{j=0}^N \begin{pmatrix}\int_{C^h_{i,j}\cap C^h_{i+1,j}}\phi^{(1)}\,d\mathcal{H}^1\\\int_{C^h_{i,j}\cap C^h_{i,j+1}}\phi^{(2)}\,d\mathcal{H}^1\end{pmatrix}\cdot \nabla^h u^h_{i,j}\,.
\end{aligned}
\end{equation*}
One can moreover show there exists $C>0$ such that for~$h$ small enough and  all $(i,j)$ we have:
\begin{equation*}\begin{aligned}
&\left|\left[\int_{C^h_{i,j}\cap C^h_{i+1,j}}\phi^{(1)}\,d\mathcal{H}^1\right]-h\,\phi^{(1)}(x^h_{i+1,j+1})\right|\leq Ch^2\,,\\
&\left|\left[\int_{C^h_{i,j}\cap C^h_{i,j+1}}\phi^{(2)}\,d\mathcal{H}^1\right]-h\,\phi^{(2)}(x^h_{i+1,j+1})\right|\leq Ch^2\,,\\
\end{aligned}\end{equation*}
with $x_{i,j}\eqdef(-R+i\,h,-R+j\,h)$. We use the above inequalities and the fact $\|\phi(x)\|\leq 1$ for all $x$ to obtain the existence of $C'>0$ such that for~$h$ small enough and for all $(i,j)$ we have:
\begin{equation*}
\left\|\begin{pmatrix}\int_{C^h_{i,j}\cap C^h_{i+1,j}}\phi^{(1)}\,d\mathcal{H}^1\\\int_{C^h_{i,j}\cap C^h_{i,j+1}}\phi^{(2)}\,d\mathcal{H}^1\end{pmatrix}\right\|_2\leq h\,\sqrt{1+C'h}\,.
\end{equation*}
This finally yields
\begin{equation*}\begin{aligned}
&\sum\limits_{i=0}^N\sum\limits_{j=0}^N \begin{pmatrix}\int_{C^h_{i,j}\cap
C^h_{i+1,j}}\phi^{(1)}\,d\mathcal{H}^1\\\int_{C^h_{i,j}\cap
C^h_{i,j+1}}\phi^{(2)}\,d\mathcal{H}^1\end{pmatrix}\cdot \nabla^h
u^h_{i,j}\\
\leq &\sum\limits_{i=0}^N\sum\limits_{j=0}^N h\,\sqrt{1+C'h} ~\|\nabla^h u^h_{i,j}\|=\sqrt{1+C'h}~J^h(u^h)\,,
\end{aligned}\end{equation*}
which gives
\begin{equation*}\int_{\RR^2}\phi\cdot dDu\leq \underset{h\to 0}{\text{lim sup}}~ \sqrt{1+C'h}~J^h(u^h)\leq 1\,.\end{equation*}

We now have to show that
\begin{equation*}
\forall v\in\LD,~J(v)\leq 1\implies\int_{\RR^2}\eta\,u\leq \int_{\RR^2}\eta\,v\,.
\end{equation*}
Let $v\in C^{\infty}([-R,R]^2)$ be
such that $J(v)\leq 1$. We define \begin{equation*}
v^h\eqdef\left(v\left(\left(i+\frac{1}{2}\right)h,\,\left(j+\frac{1}{2}\right)h\right)\right)_{(i,j)\in[ 1,N]^2}.
\end{equation*}
One can then show that \begin{equation*}
\underset{h\to 0}{\text{lim}}~J^h(v^h)=J(v)=1\,,\end{equation*}
so that for every $\delta>0$ we have~${J^h\left(\frac{v^h}{1+\delta}\right)\leq1}$ for $h$ small enough. Now this yields
\begin{equation*}
\begin{aligned}
   \int_{[-R,R]^2}\eta\,u&=\underset{h\to 0}{\text{lim}}~\int_{[-R,R]^2}\eta\,u^h\\&\leq\underset{h\to 0}{\text{lim}}~\int_{[-R,R]^2}\eta\,\frac{v^h}{1+\delta}\\&=\int_{[-R,R]^2}\eta\,\frac{v}{1+\delta}\,.
\end{aligned}
\end{equation*}
Since this holds for all $\delta>0$ we get that \begin{equation}\int_{[-R,R]^2}\eta\,u\leq
\int_{[-R,R]^2}\eta\,v\,.\label{opt_cheeger_relaxed}\end{equation} Finally, if $v\in\LD$ is such that~$v=0$ outside
$[-R,R]^2$ and~${J(v)\leq 1}$, by standard approximation results (see \cite[remark 3.22]{ambrosioFunctionsBoundedVariation2000}) we also have that
\eqref{opt_cheeger_relaxed} holds, and hence $u$ solves \eqref{cheeger_relaxed}. Finally, since $u$ solves \eqref{cheeger_relaxed}, its support is included in $[-R,R]^2$, which shows the strong
$\mathrm{L}^1_{loc}(\RR^2)$ convergence of $(u^h)$ towards $u^*$ in fact implies its strong $\mathrm{L}^1(\RR^2)$ convergence.

\section{First variation of the perimeter and weighted area functionals for simple polygons}
\label{shape_grad}

\revision{We stress that since $\mathcal{X}_n$ is open, for every~${x\in\mathcal{X}_n}$} the functions~${h\mapsto P(E_{x+h})}$ and $h\mapsto \int_{E_{x+h}}\eta$ are well-defined in a neighborhood of zero (for any locally integrable function $\eta$). We now compute the first variation of these two quantities.

\begin{proposition}
	Let \revision{$x\in\mathcal{X}_n$}. Then we have
	\begin{equation}
		P(E_{x+h})=P(E_x)-\sum\limits_{i=1}^n \left\langle h_i,\tau_i-\tau_{i-1}\right\rangle + o\left(\|h\|\right),
	\end{equation}
	where $\tau_i\eqdef \frac{x_{i+1}-x_i}{\|x_{i+1}-x_i\|}$ is the unit tangent vector to $[x_i,x_{i+1}]$.
\end{proposition}
\begin{proof}
	If $\|h\|$ is small enough we have:
	\begin{equation*}
		\begin{aligned}
			&P(E_{x+h})=\sum\limits_{i=1}^n\|x_{i+1}-x_i+h_{i+1}-h_i\|\\
			=&\sum\limits_{i=1}^n\sqrt{\|x_{i+1}-x_i+h_{i+1}-h_{i}\|^2}\\
			=&\sum\limits_{i=1}^n\|x_{i+1}-x_i\|\left(1+\frac{\langle x_{i+1}-x_i,h_{i+1}-h_i\rangle}{\|x_{i+1}-x_i\|^2}+o\left(\|h\|\right)\right)\\
			=&~P(E_x)+\sum\limits_{i=1}^n\langle \tau_i,h_{i+1}-h_i\rangle+o\left(\|h\|\right),\\
		\end{aligned}
	\end{equation*}
	and the result follows by re-arranging the terms in the sum.
\end{proof}

\begin{proposition}
	Let \revision{$x\in\mathcal{X}_n$} and $\eta\in C^0(\RR^2)$. Then we have
	\begin{equation}
		\int_{E_{x+h}}\eta=\int_{E_x}\eta+\sum\limits_{i=1}^n \left\langle h_i,w_i^-\nu_{i-1}+w_i^+\nu_i\right\rangle + o\left(\|h\|\right),
	\end{equation}
	where $\nu_i$ is the outward unit normal to \revision{$E_x$ on $]x_i,x_{i+1}[$} and
	\begin{equation*}
	\begin{aligned}
	w_{i}^{+}&\eqdef\int_{[x_i,x_{i+1}]}\eta(x)~
	\frac{\|x-x_{i+1}\|}{\|x_{i}-x_{i+1}\|}\,d\mathcal{H}^1(x)\,, \\
	w_{i}^{-}&\eqdef\int_{[x_{i-1},x_{i}]}\eta(x)~
	\frac{\|x-
	x_{i-1}\|}{\|x_{i}-x_{i-1}\|}\,d\mathcal{H}^1(x)\,.
	   \end{aligned}
	\end{equation*}
\end{proposition}
\begin{proof}
	Our proof relies on the following identity \revision{(see \Cref{lemma_signed_area} for a proof of a closely related formula)}:
	\begin{equation*}
\int_{E_x}\eta=\text{sign}\left(\sum\limits_{i=1}^n\det(x_i~x_{i+1})\right)\sum\limits_{i=1}^n \omega(x_i,x_{i+1})\,,
\end{equation*}
with
\begin{equation*}
	\omega(a_1,a_2)\eqdef\det(a_1~a_2)\int_{T_1}\eta((a_1~a_2)\,y)\,dy\,,
\end{equation*}
where $T_1\eqdef\left\{(\alpha,\beta)\in\left(\RR_+\right)^2\,\big\rvert\,\alpha+\beta\leq 1\right\}$ is the unit triangle. Assuming $\eta\in C^1(\RR^2)$ \revision{and denoting $\mathrm{adj}(A)$ the adjugate of a matrix $A$}, we have:
\begin{equation*}
\begin{aligned}
			&\omega(a_1+h_1,a_2+h_2)\\ =~&\omega(a_1,a_2)+\mathrm{det}(a_1~a_2)\int_{T_1}\nabla\eta((a_1~a_2)\,y)\cdot((h_1~h_2)y)\,dy\\
			+~&\mathrm{tr}\left(\mathrm{adj}(a_1~a_2)^T(h_1~h_2)\right)\int_{T_1}\eta((a_1~a_2)\,y)\,dy+o(\|h\|)\\
			=~&\omega(a_1,a_2)\\
			+~&\mathrm{sign}(\mathrm{det}(a_1~a_2))\int_{Oa_1a_2}\nabla\eta(y)\cdot((h_1~h_2)\,(a_1~a_2)^{-1})\,y)\,dy\\
			+~&\frac{\mathrm{tr}\left(\mathrm{adj}(a_1~a_2)^T(h_1~h_2)\right)}{|\mathrm{det}(a_1~a_2)|}\int_{Oa_1 a_2}\eta(y)\,dy+o(\|h\|)\,.
\end{aligned}
\end{equation*}
	Denoting $g(y)\eqdef (h_1~h_2)(a_1,a_2)^{-1}\,y$, we obtain:
	\begin{equation*}\begin{aligned}
	&\omega(a_1+h_1,a_2+h_2)\\
	=~&\omega(a_1,a_2)\\
	+~&\mathrm{sign}(\mathrm{det}(a_1~a_2))
	\int_{Oa_1a_2}\left[\nabla\eta\cdot g+\eta\,\mathrm{div}g\right]+o(\|h\|)\\
	=~&\omega(a_1,a_2)\\
	+~&\mathrm{sign}(\mathrm{det}(a_1~a_2))
	\int_{\partial(Oa_1a_2)}\eta\,(g\cdot\nu_{Oa_1a_2})\,d\mathcal{H}^1+o(\|h\|)\,,
	\end{aligned}\end{equation*}
	where we used Gauss-Green theorem to obtain the last equality. Now if $\|h\|$ is small enough then
	\begin{equation*}
	\sum\limits_{i=1}^n\det(x_i+h_i~~x_{i+1}+h_{i+1}) ~\text{and}~ \sum\limits_{i=1}^n\det(x_i~x_{i+1})
	\end{equation*}
	have the same sign, so that, defining
	\begin{equation*}
		g_i:y\mapsto((h_i~h_{i+1})(x_i~x_{i+1})^{-1}\,y)\,,
	\end{equation*}
	we get
	\begin{equation*}
		d\left(\int_{E_{\bullet}}\eta\right)(x)\,.\,h=\epsilon\sum\limits_{i=1}^n\mathrm{sign}(\mathrm{det}(x_i~x_{i+1}))~\omega_i\,,
	\end{equation*}
	with
	\begin{equation*}
   \begin{aligned}
      \epsilon&\eqdef\text{sign}\left(\sum\limits_{i=1}^n\det(x_i~x_{i+1})\right),\\
	\omega_i&\eqdef\int_{\partial^*(Ox_ix_{i+1})}\eta\,(g_i\cdot\nu_{Ox_ix_{i+1}})\,d\mathcal{H}^1\,.
   \end{aligned}
	\end{equation*}
	Then one can decompose each integral in the sum and show the integrals over $[0,x_i]$ cancel out each other, which allows to obtain
	\begin{equation*}
		d\left(\int_{E_{\bullet}}\eta\right)(x)\,.\,h=\sum\limits_{i=1}^n\int_{[x_i,x_{i+1}]}\eta\,(g_i\cdot\nu_i)\,d\mathcal{H}^1\,.
	\end{equation*}
	But now if $y\in[x_i,x_{i+1}]$ then
	\begin{equation*}
		(x_i~x_{i+1})^{-1}y=\frac{1}{\|x_{i+1}-x_i\|}\begin{pmatrix}\|y-x_{i+1}\|\\\|y-x_i\|\end{pmatrix}\,,
	\end{equation*}
	and the result follows by re-arranging the terms in the sum. One can then use an approximation argument \revision{as in \cite[Proposition 17.8]{maggiSetsFinitePerimeter2012}} to show it also holds when $\eta$ is only continuous.
\end{proof}

\section{Results used in Section \ref{toy}}
\label{proofs_toy}

\subsection{Properties of the radialisation operator}
\label{radial_tv}

The goal of this subsection, based on \cite[II.1.4]{dautrayMathematicalAnalysisNumerical2012}, is to prove the following result:
\begin{proposition}
	\label{radialisation_prop}
  Let $u\in \LD$ be s.t.~${|\Diff u|(\RR^2)<+\infty}$. Then~${|\Diff\tilde{u}|(\RR^2)\leq |\Diff u|(\RR^2)}$ with
  equality if and only if~$u$ is radial.
\end{proposition}

First, one can show that for every $u\in \LD$, the radialisation $\tilde{u}$ of $u$ defined in \Cref{toy} by
\begin{equation}
  \tilde{u}(x)=\int_{\mathbb{S}^1}u(\|x\|\,e)\,d\mathcal{H}^1(e)
\end{equation}
is well defined and belongs to $\LD$. Then a change of variables in polar coordinates shows that, as stated in the following lemma, the radialisation operator is self-adjoint.
\begin{lemma}
	\label{radialisation_autoadj}
	We have
	\begin{equation*}
		\forall u,v\in\LD,~\int_{\RR^2}\tilde{u}(x)\,v(x)\,dx=\int_{\RR^2}u(x)\,\tilde{v}(x)\,dx\,.
	\end{equation*}
\end{lemma}
We now state a useful identity:
\begin{lemma}
	For every $\varphi\in C^{\infty}_c(\RR^2\setminus\{0\},\RR^2)$, we have:
	\begin{equation*}
		\left\langle \Diff\tilde{u},\varphi\right\rangle=\left\langle \Diff
	u,\widetilde{\varphi\cdot\frac{x}{\|x\|}}\right\rangle,
	\end{equation*}
	where $\varphi\cdot\frac{x}{\|x\|}$ denotes the mapping $x\mapsto \left(\varphi(x)\cdot \frac{x}{\|x\|}\right)$.
	\label{lemma_grad_radial}
\end{lemma}
\begin{proof}
	From \Cref{radialisation_autoadj} we get
	\begin{equation*}\left\langle  \Diff\tilde{u},\varphi\right\rangle=\int_{\RR^2}\tilde{u}\,\div\varphi=\int_{\RR^2}u\,\widetilde{\div\varphi}\,.\end{equation*}
	Using polar coordinates, defining \begin{equation*}h(r,\theta)\eqdef (r\cos(\theta),r\sin(\theta))\,,\end{equation*}
	we get
	\begin{equation}
		\begin{aligned}
		\left(\div\varphi\right)(h(r,\theta))=&~~~~\frac{1}{r}\frac{\partial}{\partial r}(r(\varphi_r \circ h))(r,\theta)\\
		&\,+\frac{1}{r}\frac{\partial}{\partial\theta}(\varphi_{\theta} \circ h)(r,\theta)\,,
	\end{aligned}
			\label{div_polar}
	\end{equation}
	where $\varphi_r$ and $\varphi_{\theta}$ respectively denote the radial and orthoradial components of $\varphi$, i.e. \begin{equation*}
	\varphi_r(x)=\varphi(x)\cdot\frac{x}{\|x\|} \text{ and }\varphi_{\theta}(x)=\varphi(x)\cdot\frac{x^{\perp}}{\|x\|}\,.\end{equation*}
	The second term in \eqref{div_polar} has zero circular mean. Interchanging derivation and integration we get that the radialisation
	of the first term equals
	$\frac{1}{r}\frac{\partial}{\partial r}\left(r\left(\widetilde{\varphi_r}\circ h\right)\right)$, which yields
	\begin{equation*}
		\left(\widetilde{\div\varphi}\right)(x)=\div\left(\widetilde{\varphi\cdot\frac{x}{\|x\|}}\right)(x)\,.
	\end{equation*}
\end{proof}

We now introduce the radial and orthoradial components of the gradient, which are Radon measures on~${U\eqdef\RR^2\setminus\{0\}}$ defined by
\begin{equation*}\begin{aligned}
	\forall\psi\in C^{\infty}_c(U),~\left\langle\Diffrad u,\psi\right\rangle&=\left\langle \Diff u ,\psi\,\frac{x}{|x|}\right\rangle,
\\
\left\langle\Difforth u,\psi\right\rangle&=\left\langle \Diff u ,\psi\,\frac{x^\perp}{|x|}\right\rangle .
\end{aligned}\end{equation*}
\begin{proposition}
	There exist two $|\Diff u|$-measurable mappings from $U$ to $\RR$, denoted $g_{\mathrm{rad}}$ and $g_{\mathrm{orth}}$, such that
	\begin{equation*}
	g_{\mathrm{rad}}^2+g_{\mathrm{orth}}^2\leq 1~|\Diff u|\text{-almost everywhere}
	\end{equation*}
	and
	\begin{equation}
		\begin{aligned}
		\forall\psi\in C^{\infty}_c(U),~\langle \Diffrad u,\psi\rangle&=\int_U \psi(x)\,g_{\mathrm{rad}}(x)\,d|\Diff u|(x)\,, \\
		\langle \Difforth u,\psi\rangle&=\int_U \psi(x)\,g_{\mathrm{orth}}(x)\,d|\Diff u|(x)\,.
	\end{aligned}
	\label{orth_rad}
\end{equation}
\end{proposition}
\begin{proof}
	The existence of the $|\Diff u|$-measurable mappings $g_{\mathrm{rad}}$ and  $g_{\mathrm{orth}}$, as well as \eqref{orth_rad}, come from Lebesgue differentiation theorem and the fact $\Diffrad u$ and $\Difforth u$ are absolutely continuous with respect to $\Diff u$. Now for every open set  $A\subset U$ we have:
	\begin{equation*}
		\begin{aligned}
			|\Diff u|(A)&=\sup\left\{\left\langle \Diff u,\varphi\right\rangle\,\big\rvert\,\varphi\in C^{\infty}_c(A,\RR^2),~\|\varphi\|_{\infty}\leq 1\right\}\\
				&=\sup\bigg\{\left\langle \Diff u,\varphi_1\frac{x}{\|x\|}\right\rangle+
				\left\langle \Diff u,\varphi_2\frac{x^{\perp}}{\|x\|}\right\rangle \text{ s.t.}\\
				&\qquad\qquad~\varphi_i\in C^{\infty}_c(A),~\left\|\varphi_1^2+\varphi_2^2\right\|_{\infty}\leq 1\bigg\}\\
				&=\sup\bigg\{\left\langle \Diffrad u,\varphi_1\right\rangle+
				\left\langle \Difforth u,\varphi_2\right\rangle\text{ s.t. }\\
				&\qquad\qquad~\varphi_i\in C^{\infty}_c(A),~\left\|\varphi_1^2+\varphi_2^2\right\|_{\infty}\leq 1\bigg\}\,.
		\end{aligned}
	\end{equation*}
	Hence for $\varphi_i\in C^{\infty}_c(A)$ such that $\left\| \varphi_1^2+\varphi_2^2\right\|_{\infty}\leq 1$ we have:
	\begin{equation*}
		\int_{A}1\,d|\Diff u|\geq \int_{A}(g_{\mathrm{rad}}\,\varphi_1+g_{\mathrm{orth}}\,\varphi_2)\,d|\Diff u|\,.
	\end{equation*}
	If we had $g_{\mathrm{rad}}^2+g_{\mathrm{orth}}^2> 1$ on a set of non zero measure $|\Diff u|$, we would have a contradiction.
\end{proof}

We can now prove \Cref{radialisation_prop}. Indeed, since $\{0\}$ is $\mathcal{H}^1$-negligible, we have that
\begin{equation*}|\Diff u|(\{0\})=|\Diff \tilde{u}|(\{0\})=0\,,\end{equation*} and moreover
\begin{equation}
	|\Diff\tilde{u}|(\RR^2\setminus\{0\})\leq |\Diffrad u|(\RR^2\setminus\{0\})\leq |\Diff u|(\RR^2\setminus\{0\})\,.
\end{equation}
The first equality comes from \Cref{lemma_grad_radial}, while the second is easily obtained from the definition of $\Diffrad$. Now if we have~${|\Diffrad u|(U)=|\Diff u|(U)}$, then we get
\begin{equation*}
	\int_{U}g_{\mathrm{rad}}\,d|\Diff u|=\int_{U}\sqrt{g_{\mathrm{rad}}^2+g_{\mathrm{orth}}^2}\,d|\Diff u|=\int_{U}d|\Diff u|\,.
\end{equation*}
This yields $g_{\mathrm{orth}}=0$ (and $|g_{\mathrm{rad}}|=1$) $|\Diff u|$-almost everywhere. Hence $\Difforth u=0$.

Let us now show this implies that $u$ is radial. If we define~${A\eqdef\,]0,+\infty[\times]-\pi,\pi[}$, then we have that the mapping given by~${h:(r,\theta)\mapsto (r\cos\theta,r\sin\theta)}$ is a
$C^{\infty}$-diffeomorphism from $A$ to $\RR^2\setminus\left(\RR_{-}\times\{0\}\right)$. Now if $\xi\in C_c^{\infty}(A)$ we have that~${\xi\circ h^{-1}\in C_c^{\infty}(h(A))}$ and
\begin{equation*}
	\begin{aligned}
		0&=\left\langle \Difforth u,\xi\circ h^{-1}\right\rangle\\
		&= \int_{\RR^2}u\,\div\left(\left(\xi\circ h^{-1}\right)\frac{x^{\perp}}{\|x\|}\right)\\
		&=\int_{0}^{+\infty}\int_{-\pi}^{\pi}\left(u\circ h\right)(r,\theta)\,\left(\frac{1}{r}\frac{\partial}{\partial\theta}(\xi)(r,\theta)\right)r\,d\theta\,dr\,.
	\end{aligned}
\end{equation*}
This means that $\frac{\partial\theta}{\partial}(u\circ h)=0$ in the sense of distributions, and hence that there exists\footnote{To see this, notice that if we convolve $u\circ h$ with an approximation of unity $\rho_{\epsilon}$, then we have \begin{equation*}\frac{\partial}{\partial\theta}\left((u\circ h)\star\rho_{\epsilon}\right)=(u\circ h)\star \frac{\partial}{\partial
\theta}\rho_{\epsilon}=0\,,\end{equation*}
hence the smooth function $(u\circ h)\star \rho_{\epsilon}$ is equal to some function~$g_{\epsilon}$ that depends only on
$r$. Letting $\epsilon\to 0^+$, we see that for almost every $(r,\theta)$, $u\circ h$ only depends on $r$.} a mapping~${g:\,]0,+\infty[\to\RR}$ such that for almost every $(r,\theta)\in A$, $(u\circ h)(r,\theta)=g(r)$. We finally get~${u(x)=g(\|x\|)}$ for almost every $x\in h(A)$, which shows $u$ is radial.

\subsection{Lemmas used in the proof of \Cref{cheeger-radial}}

We take $\eta\eqdef\varphi$ and keep the assumptions of \Cref{toy}.

\begin{lemma}
  Let $f:\RR\to\RR_+$ be square integrable, even and decreasing on $\RR_+$. Then for every measurable set $A$ such that~${|A|<+\infty}$ we have
  \begin{equation*}\int_{A}f\leq\int_{A^s}f\,,\end{equation*}
  where $A^s\eqdef[-\frac{|A|}{2},\frac{|A|}{2}]$. Moreover, equality holds if and only if $|A\triangle A^s|=0$.
  \label{steiner_1d}
\end{lemma}
\begin{proof}
  We have
	\begin{equation*}
		\int_{A}f=\int_{0}^{+\infty}\left|\left\{f\,\mathbf{1}_A\geq t\right\}\right|dt=\int_{0}^{+\infty}\left|\left\{f\geq t\right\}  \cap A\right|dt\,.
	\end{equation*}
  For all $t>0$ there exists $\alpha$ such that $\{f\geq t\}=[-\alpha,\alpha]$, so that we have
	\begin{equation*}
		\begin{aligned}
		\left|\left\{f\geq t\right\}  \cap A\right|&=\left|[-\alpha,\alpha]  \cap A\right|\leq \text{min}(2\alpha,|A|)\\&=\left|[-\alpha,\alpha]\cap [-|A|/2,|A|/2]\right|\\&=\left|\left\{f\geq t\right\}  \cap A^s\right|.
	\end{aligned}
	\end{equation*}
  Hence \begin{equation*}\int_{A}f\leq \int_{0}^{+\infty}\left|\left\{f\geq t\right\}  \cap A^S\right|dt=\int_{A^s}f\,.\end{equation*}
  Now if $|A\triangle A^s|>0$ then $|A\setminus A^s|=|A^s\setminus A|>0$ and we have
  \begin{equation*}
		\begin{aligned}
		\int_{A^s}f&=\int_{A\cap A^s}f+\int_{A^s\setminus A}f \\&>\int_{A\cap A^s}f+f\left(\frac{|A|}{2}\right)\,|A^s\setminus A|\\&\geq \int_{A\cap A^s}f+\int_{A\setminus A^s}f=\int_{A}f\,,
	\end{aligned}
	\end{equation*}
  which proves the second part of the result.
\end{proof}

\begin{lemma}
  Let $E\subset\RR^2$ be s.t. $0<|E|<\infty$ and~${P(E)<\infty}$. Then for any $\nu\in\mathbb{S}^1$, denoting $E^s_{\nu}$ the Steiner symmetrization of $E$
  with respect to the line through the origin directed by $\nu$, we have
  \begin{equation*}\frac{\int_{E^s_{\nu}}\eta}{P(E^s_{\nu})}\geq \frac{\int_{E}\eta}{P(E)}\,,\end{equation*}
  with equality if and only if $|E\triangle E^s_{\nu}|=0$.
  \label{steiner}
\end{lemma}
\begin{proof}
From \cite[theorem 14.4]{maggiSetsFinitePerimeter2012} we know that we have~${P(E_{\nu}^s)\leq P(E)}$. We now perform a change of coordinates in order to
have $E^s_{\nu}=\left\{(x_1,x_2)\in\RR^2\,\rvert\,|x_2|\leq\frac{\mathcal{L}^1(E_{x_1})}{2}\right\}$ with \begin{equation*}E_{x_1}\eqdef\left\{x_2\in\RR\,\rvert\, (x_1,x_2)\in E\right\}.\end{equation*}
Now we have
\begin{equation*}\begin{aligned}
	\int_{E}\eta&=\int_{-\infty}^{+\infty}\left(\int_{-\infty}^{+\infty}\eta(x_1,x_2)\,\mathbf{1}_{E}(x_1,x_2)\,dx_2\right)dx_1\\
&=\int_{-\infty}^{+\infty}\left(\int_{E_{x_1}}\eta(x_1,\cdot)\right)dx_1\,,\end{aligned}\end{equation*}
with $E_{x_1}=\left\{x_2\in\RR\,\rvert\, (x_1,x_2)\in E\right\}$. For almost every~${x_1\in\RR}$ we have that~${E_{x_1}}$ is measurable, has finite measure, and that~${\eta(x_1,\cdot)}$ is nonnegative, square integrable, even and decreasing on $\RR_+$. We can hence apply \Cref{steiner_1d} and get that
\begin{equation}
	\int_{E}\eta\geq \int_{-\infty}^{+\infty}\left(\int_{\left(E_{x_1}\right)^s}\eta(x_1,\cdot)\right)d_{x_1}=\int_{E^s_{\nu}}\eta\,.
   \label{ineq_aux_steiner}
\end{equation}
Moreover, if $|E\triangle E^s_{\nu}|>0$, then since
\begin{equation*}
  \begin{aligned}
    |E\triangle E^s_{\nu}|&=\int_{0}^{+\infty}\left(\int_{0}^{+\infty}|\mathbf{1}_E(x_1,x_2)-
\mathbf{1}_{E^s_{\nu}}(x_1,x_2)|\,dx_2\right)dx_1\\
&=\int_{0}^{+\infty}\left(\int_{0}^{+\infty}\left|\mathbf{1}_{E_{x_1}}(x_2)-
\mathbf{1}_{\left(E_{x_1}\right)^s}(x_2)\right|\,dx_2\right)dx_1\\
&=\int_{0}^{+\infty}\left|E_{x_1}\triangle \left(E_{x_1}^s\right)\right|\,dx_1\,,
\end{aligned}
\end{equation*}
we get that $\mathcal{L}^1\left(\left\{x_1\in\RR\,\rvert\,\left|E_{x_1}\triangle\left(E_{x_1}\right)^s\right|>0\right\}\right)>0$ and hence that \eqref{ineq_aux_steiner} is strict.
\end{proof}

\begin{lemma}
  Under Assumption 1, the mapping \begin{equation*}\mathcal{G}:R\mapsto \frac{1}{R}\int_{0}^R r\,\tilde{\varphi}(r)\,dr\end{equation*} has a unique maximizer.
	\label{lemma_variations}
\end{lemma}
\begin{proof}
  Since $\varphi$ (and hence $\tilde{\varphi}$) is continuous, we have that $\mathcal{G}$ is $C^1$ on $R_+^*$ and
 \begin{equation*}
	 \mathcal{G}'(R)=\frac{R\,(R\,\tilde{\varphi}(R))-\int_{0}^R r\,\tilde{\varphi}(r)\,dr}{R^2}\,.
\end{equation*}
 Now an integration by part yields that for any continuously differentiable function $h:]0,+\infty[\to\RR$ and for any $x>0$ we have \begin{equation*}H(x)\eqdef x\,h(x)-\int_{0}^x h=\int_{0}^x t\,h'(t)\,dt\,,\end{equation*}
 which shows $H'(x)=x\,h'(x)$. This means the mappings
 \begin{equation*}
	 R\mapsto R\,(R\,\tilde{\varphi}(R))-\int_{0}^R r\,\tilde{\varphi}(r)\,dr \text{ and } R\mapsto f(R)=R\,\tilde{\varphi}(R)
 \end{equation*}
	 have the same variations. Under Assumption 1, it is then easy to show there exists $R_0>0$ such that~${\mathcal{G}'(R_0)=0}$, $\mathcal{G}'$ is positive on $]0,R_0[$ and negative on $]R_0,+\infty[$, hence the result.
\end{proof}

\subsection{Proof of \Cref{cheeger_polygon_radius}}
\label{proof_cheeger_polygon_radius}
We define \begin{equation*}
R_n(\theta)\eqdef R\,\frac{\cos\left(\pi/n\right)}{\cos\left(\left(\theta\text{ mod }2\pi/n\right)-\pi/n\right)},
\end{equation*}
\revision{so that in polar coordinates an equation of the boundary of a regular $n$-gon of radius $R$ with a vertex at $(0,0)$ is given by~${r(\theta)=R_n(\theta)}$}. Under Assumption 1 we have, for all $R>0$:
\begin{equation*}
	\begin{aligned}
		&2\pi R\,\frac{\tan\left(\pi/n\right)}{\pi/n}\left|\mathcal{G}(R)-\mathcal{G}_n(R)\right|\\
		&=\bigg|\frac{\tan\left(\pi/n\right)}{\pi/n}\,\int_{0}^{2\pi}\int_{0}^R
		r\tilde{\varphi}(r)\,dr\, d\theta\\
		&\qquad\qquad\quad-\int_{0}^{2\pi}\int_{0}^{R_n(\theta)} r\tilde{\varphi}(r)\,dr\,
		d\theta\bigg|\\
		&=\bigg|\int_{0}^{2\pi}\int_{R_n(\theta)}^R r\tilde{\varphi}(r)\,dr\,
		d\theta\\
		&\qquad-\left(1-\frac{\tan\left(\pi/n\right)}{\pi/n}\right)\int_{0}^{2\pi}\int_{0}^{R} r\tilde{\varphi}(r)\,dr\,
		d\theta\bigg|\\
		&\leq \bigg[2\pi\,\underset{\theta\in[0,2\pi]}{\sup}\left|R-
		R_n(\theta)\right| \|f\|_{\infty}\\
		&\qquad+\left(1-\frac{\tan\left(\pi/n\right)}{\pi/n}\right)\,2\pi R\,\|f\|_{\infty}\bigg]\\
		&\leq \|f\|_{\infty}\left[(1-
		\cos\left(\pi/n\right)) +\left(1-\frac{\tan\left(\pi/n\right)}{\pi/n}\right)\right].
	\end{aligned}
\end{equation*}
We hence obtain that $\left|\mathcal{G}(R)-\mathcal{G}_n(R)\right|_{\infty}=O\left(\frac{1}{n^2}\right)$.

Now assuming $f$ is of class $C^2$ and $f''(\rho_0)<0$ we want to prove that for $n$ large enough, $\mathcal{G}_n$ has a unique maximizer~$R^*_n$ and~${|R^*_n-R^*|=O\left(\frac{1}{n}\right)}$. Denoting $\alpha_n(s)\eqdef\frac{\cos\left(\pi/n\right)}{\cos\left(\pi s/n\right)}$, we have:
\begin{equation*}
	\begin{aligned}
		\mathcal{G}_n(R)&=\frac{1}{2\pi R\,\frac{\tan\left(\pi/n\right)}{\pi/n}}\int_{0}^{2\pi}\int_{0}^{R_n(\theta)}r\tilde{\varphi}(r)\,dr\,d\theta\\
		&=\frac{1}{2\pi R\,\frac{\tan\left(\pi/n\right)}{\pi/n}}~n\int_{0}^{2\pi/n}\int_{0}^{R\frac{\cos\left(\pi/n\right)}{\cos\left(\theta-\pi/
		n\right)}}r\tilde{\varphi}(r)\,dr\,d\theta \\
		&=\frac{\pi/n}{R\,\tan\left(\pi/n\right)}~\int_{0}^{1}\int_{0}^{R\,\alpha_n(s)}r\tilde{\varphi}(r)\,dr\,ds\\
		&=\frac{\pi/n}{\tan\left(\pi/n\right)}\frac{1}{R}\int_{0}^{R}r\left[\int_{0}^{1}\alpha_n(s)^2\,\tilde{\varphi}(r\,\alpha_n(s))\,ds\right]dr\,.
	\end{aligned}
\end{equation*}
Considering \Cref{lemma_variations} and defining \begin{equation*}f_n:r\mapsto r\left[\int_{0}^{1}\alpha_n(s)^2\,\tilde{\varphi}(r\,\alpha_n(s))\,ds\right],\end{equation*}
we see that showing $f_n'$ is positive on $]0,\rho_1[$ and negative on~${]\rho_1,+\infty[}$ for some $\rho_1$  is sufficient to prove $\mathcal{G}_n$ has a unique maximizer. Now we have
\begin{equation*}f_n'(r)=\int_{0}^1 \alpha_n(s)^2\left(\tilde{\varphi}(r\,\alpha_n(s))+r\,\alpha_n(s)\,\tilde{\varphi}'(r\,\alpha_n(s))\right)ds\,.\end{equation*}
The image of $[0,1]$ by $s\mapsto r\,\alpha_n(s)$ is $[r\cos\left(\pi/n\right),r]$. Since the mapping $r\mapsto \tilde{\varphi}(r)+r\tilde{\varphi}'(r)=(r\tilde{\varphi})'(r)$ is positive on~${]0,\rho_0[}$ and negative on $]\rho_0,+\infty[$, we get that $f_n'$ is
positive on~${]0,\rho_0[}$ and negative on~${]\rho_0/\cos\left(\pi/n\right),+\infty[}$ and it hence remains to investigate its sign on~${[\rho_0,\rho_0/\cos\left(\pi/n\right)]}$. But since $f$ is of class $C^2$ and~ ${f''(\rho_0)<0}$ there exists
$\epsilon>0$ s.t.~${f''(r)<0}$ on~${]\rho_0-\epsilon,\rho_0+\epsilon[}$. For $n$ large enough, we hence have \begin{equation*}[\rho_0\cos\left(\pi/n\right),\rho_0/\cos\left(\pi/n\right)]\subset\,]\rho_0-\epsilon,\rho_0+\epsilon[\,,\end{equation*}
which implies that
\begin{equation*}\forall r\in[\rho_0,\rho_0/\cos\left(\pi/n\right)],~ r\,\alpha_n(s)\in\,]\rho_0-\epsilon,\rho_0+\epsilon[\,,\end{equation*} and hence $f_n''(r)<0$. This finally shows there exists $\rho_1$ such that $f_n'$ is positive on $]0,\rho_1[$ and negative on $]\rho_1,+\infty[$, and the result follows as in the proof of \Cref{lemma_variations}.

Now $R^*$ and $R_n^*$ and  are respectively the unique solutions of $F(0,R)=0$ and $F(\pi/n,R)=0$ with
\begin{equation*}
\begin{aligned}
	F(t,R)&\eqdef\left[\int_{0}^R f_t\right] - R\,f_t(R)\,,\\f_t(r)&\eqdef r\int_{0}^1 \alpha(t,s)^2\,\tilde{\varphi}(r\,\alpha(t,s))\,ds\,,\\ \alpha(t,s)&\eqdef\frac{\cos t}{\cos(ts)}\,.
\end{aligned}
\end{equation*}
One can then show $\frac{\partial}{\partial R}F(0,R)=0$ if and only if $f_0'(R)=0$, i.e. if and only if $R=\rho_0$. But from the proof of \Cref{lemma_variations} and the above, it is easy to see neither $R^*$ nor $R_n^*$ equals $\rho_0$. We can hence apply the implicit function theorem to finally get that $|R^*-R_n^*|=O\left(\frac{1}{n^2}\right)$.

\subsection{Proof of \Cref{critical_poly_cheeger}}
\label{critical_poly_cheeger_proof}
\subsubsection{Triangles}
Let $T$ be a triangle. Up to a rotation of the axis, we can assume that there exist~${a<b}$ and two affine functions $u,v$
such that $v\geq u$ and $u(a)=v(a)$ with
\begin{equation*}
T=\{(x,y)\in\RR^2\,\big\rvert\, x\in[a,b],~u(x)\leq y\leq v(x)\}\,.
\end{equation*}
The Steiner symmetrization $T_s$ of $T$ with respect to the line through the origin perpendicular to the side $\{b\}\times [u(b),v(b)]$ is hence obtained by replacing $u$ and $v$ in the definition of $T$ by $(u-v)/2$ and $(v-u)/2$. For all $\theta\in[0,1]$, we define
\begin{equation*}
	\begin{aligned}
u_{\theta}&\eqdef(1-\theta)\,u+\theta\,(-v)\,,\\v_{\theta}&\eqdef\theta\,(-u)+(1-\theta)\,v\,,
\end{aligned}\end{equation*}
and
\begin{equation*}
	T_{\theta}\eqdef\left\{(x,y)\in\RR^2\,\big\rvert\, x\in[a,b],~u_{\theta}(x)\leq y\leq v_{\theta}(x)\right\},
\end{equation*}
so that $T_{1/2}=T_s$. Let us now show that  \begin{equation*}
\frac{d}{d\theta}\mathcal{J}\,(T_{\theta})\big\rvert_{\theta=0}\leq 0\,,
\end{equation*}
with equality if and only if $T$ is symmetric with respect to the symmetrization line.

\medskip
\noindent\textbf{Weighted area term:} first, we have: \begin{equation*}
\int_{T_{\theta}}\eta=\int_{a}^b \left(\int_{u_{\theta}(x)}^{v_{\theta}(x)}\eta\left(\sqrt{x^2+y^2}\right)dy\right)dx\,.
\end{equation*}
Hence
\begin{equation*}
	\frac{d}{d\theta}\int_{T_{\theta}}\eta\bigg\rvert_{\theta=0}=
-\int_{a}^b(u+v)(x)\,(g_x(|v|(x))-g_x(|u|(x)))\,dx\,,
\end{equation*} with $g_x(p)\eqdef\eta\left(\sqrt{x^2+p^2}\right)$. Our assumptions on $\eta$ ensure that~${p\mapsto g_x(p)}$ is decreasing, so that $g_x(|v|(x))-g_x(|u|(x))$ and $|u|(x)-|v|(x)$ have the same sign. But since
\begin{equation*}-(u(x)+v(x)) ~\text{ and }~ -(|u|(x)-|v|(x))\end{equation*}  also have the same sign, we have that
\begin{equation*}-(u+v)(x)\,(g_x(|v|(x))-g_x(|u|(x)))<0\,,
\end{equation*}
unless $u(x)=v(x)$ or $u(x)=-v(x)$. Since $u$ and $v$ are affine and $u(a)=v(a)$, the first equality can not hold for any~${x\in]a,b[}$ (otherwise we would have $u=v$ on $[a,b]$ and~$T$ would be flat). Moreover,~${u(x)=-v(x)}$ almost
everywhere on $[a,b]$ if and only if $T=T_s$. Hence $\frac{d}{d\theta}\int_{T_{\theta}}\eta\big\rvert_{\theta=0}\leq 0$ with equality if and only if $T=T_s$.

\medskip
\noindent\textbf{Perimeter term:} now, the perimeter of $T_{\theta}$ is given by
\begin{equation*}
	\begin{aligned}
	P(T_{\theta})&=\int_{a}^b \sqrt{1+|\nabla u_{\theta}|^2}+\int_{a}^b \sqrt{1+|\nabla v_{\theta}|^2}+v_{\theta}(b)-u_{\theta}(b)\\
	&=(b-a)\left(f(\nabla u_{\theta})+f(\nabla v_{\theta})\right)+(v(b)-u(b))\,,
	\end{aligned}
\end{equation*}
with $f(p)\eqdef\sqrt{1+\|p\|^2}$, this last function being strictly convex. Now since
\begin{equation*}\left\{\begin{aligned}
	\nabla u_{\theta}&=\nabla u-\theta(\nabla u+\nabla v)\,,\\
	\nabla v_{\theta}&=\nabla v-\theta(\nabla u+\nabla v)\,,
\end{aligned}\right.\end{equation*}
we get
\begin{equation*}
	\begin{aligned}
	&\frac{d}{d\theta}P(T_{\theta})\bigg\rvert_{\theta=0}=(b-a)\left[\nabla f(\nabla u)+\nabla f(\nabla v)\right]\cdot\left[-(\nabla u+\nabla v)\right]\\
	&\qquad\qquad~=-(b-a)\left[\nabla f(\nabla u)-\nabla f(-\nabla v)\right]\cdot\left[\nabla u-(-\nabla v)\right],
	\end{aligned}
\end{equation*}
and the strict convexity of $f$ hence shows \begin{equation*}
\frac{d}{d\theta}P(T_{\theta})\bigg\rvert_{\theta=0}\leq 0\,,
\end{equation*}
with equality if and only if $\nabla u=-\nabla v$, which means, up to a translation, that $T$ is equal to $T_s$.

\medskip
\noindent Applying the above arguments to all three sides finally yields the result.

\subsubsection{Quadrilaterals}
Let $Q$ be a simple quadrilateral. Up to a rotation of the axis, we can assume that there exist~${a<b<c}$ and four affine functions $u_1,v_1,u_2,v_2$
such that
\begin{equation*}\left\{\begin{aligned}
	v_1\geq u_1,&~v_2\geq u_2\,,\\ u_1(a)=v_1(a),&~u_2(c)=v_2(c)\,,\\ u_1(b)=u_2(b),&~v_1(b)=v_2(b)\,,
\end{aligned}\right.\end{equation*} with $Q=T_1\cup T_2$ and
\begin{equation*}
T_i\eqdef\{(x,y)\in\RR^2\,\big\rvert\, x\in[a,b],~u_i(x)\leq y\leq v_i(x)\}\,.
\end{equation*}
For all $\theta\in[0,1]$ and $i\in\{1,2\}$, we define \begin{equation*}\left\{\begin{aligned}
u_{i,\theta}&\eqdef(1-\theta)\,u_i+\theta\,(-v_i)\,, \\v_{i,\theta}&\eqdef\theta\,(-u_i)+(1-\theta)\,v_i\,,
\end{aligned}\right.\end{equation*}
and $Q_{\theta}=T_{1,\theta}\cup T_{2,\theta}$ with
\begin{equation*}T_{i,\theta}=\left\{(x,y)\in\RR^2\,\big\rvert\, x\in[a,b],~u_{i,\theta}(x)\leq y\leq v_{i,\theta}(x)\right\},\end{equation*}
so that the Steiner symmetrization $Q^s$ of $Q$ with respect to the ligne through the origin perpendicular to the diagonal~${\{b\}\times [u_1(b),v_1(b)]}$ satisfies $Q_{1/2}=Q^s$.

\medskip
\noindent\textbf{Weighted area term:} the fact $\frac{d}{d\theta}\int_{Q_{\theta}}\eta\big\rvert_{\theta=0}\leq 0$ with equality if and only if $Q=Q^s$ can easily be deduced from the case of triangles using the fact that $\int_{Q_{\theta}}\eta=\int_{T_{1,\theta}}\eta+\int_{T_{2,\theta}}\eta$.

\medskip
\noindent\textbf{Perimeter term:} now, the perimeter of $Q_{\theta}$ is given by:
\begin{equation*}
	\begin{aligned}
	P(Q_{\theta})&=\int_{a}^b \sqrt{1+|\nabla u_{1,\theta}|^2}+\int_{a}^b \sqrt{1+|\nabla v_{1,\theta}|^2}\\
	&+\int_{b}^c \sqrt{1+|\nabla u_{2,\theta}|^2}+\int_{b}^c \sqrt{1+|\nabla v_{2,\theta}|^2}\\
	&=(b-a)\left(f(\nabla u_{1,\theta})+f(\nabla v_{1,\theta})\right)\\&+(c-b)\left(f(\nabla u_{2,\theta})+f(\nabla v_{2,\theta})\right)
	\end{aligned}
\end{equation*}
with $f(p)\eqdef\sqrt{1+\|p\|^2}$ as before. We then get
\begin{equation*}
	\begin{aligned}
	&\frac{d}{d\theta}P(Q_{\theta})\bigg\rvert_{\theta=0}\\
	=&-(b-a)\left[\nabla f(\nabla u_1)-\nabla f(-\nabla
	v_1)\right]\cdot\left[\nabla u_1-(-\nabla v_1)\right]\\
	&-(c-b)\left[\nabla f(\nabla u_2)-\nabla f(-\nabla
	v_2)\right]\cdot\left[\nabla u_2-(-\nabla v_2)\right],
	\end{aligned}
\end{equation*}
and the strict convexity of $f$ hence shows \begin{equation*}\frac{d}{d\theta}P(Q_{\theta})\bigg\rvert_{\theta=0}\leq 0\,,\end{equation*}
with equality if and only if $\nabla u_1=-\nabla v_1$ and $\nabla u_2=-\nabla v_2$, which means, up to a translation, that $Q$ is equal to $Q^s$.

\end{document}